%% file: arxiv-merging.tex
\documentclass[journal,onecolumn]{IEEEtran}
\usepackage{graphicx,subfig,amsthm,amsmath,latexsym,amssymb,times,bbold}
\usepackage{float,epsfig,multirow,rotating,times,verbatim,wrapfig,booktabs}
\usepackage{color,array,hyperref, algpseudocode}
\usepackage{cite}

\usepackage{array}

\newtheorem{theorem}{Theorem}
\newtheorem{lemma}[theorem]{Lemma}
\newtheorem{corollary}[theorem]{Corollary}
\newtheorem{definition}[theorem]{Definition}

\newcommand{\R}{\mathbb{R}}

\newcommand{\V}{\mathbb{V}}
\newcommand{\PP}{\mathbb{P}}

\newcommand{\abs}[1]{|#1|}
\newcommand{\norm}[1]{\Vert#1\Vert}

\newcommand{\btheta}{\boldsymbol{\theta}}

\newcommand{\bXi}{\boldsymbol{\Xi}}

\newcommand{\bGamma}{\boldsymbol{\Gamma}}

\newcommand{\bepsilon}{\boldsymbol{\epsilon}}

\newcommand{\bmu}{\boldsymbol{\mu}}
\newcommand{\bzeta}{\boldsymbol{\zeta}}

\newcommand{\bOmega}{\boldsymbol{\Omega}}

\newcommand{\bE}{\boldsymbol{E}}

\newcommand{\bx}{\boldsymbol{x}}
\newcommand{\bX}{\boldsymbol{X}}
\newcommand{\by}{\boldsymbol{y}}

\newcommand{\mR}{\mathcal R}

\newcommand{\argmin}{\operatornamewithlimits{argmin}}

\newcommand{\genomega}{\ddot{\omega}}
\newcommand{\genhatomega}{\ddot{\hat\omega}}

\newcommand{\mO}{\mathcal O}
\newcommand{\mC}{\mathcal C}

\newcommand{\bvarepsilon}{\boldsymbol{\varepsilon}}

\newcommand{\bit}{\begin{itemize}}
\newcommand{\eit}{\end{itemize}}
\newcommand{\ben}{\begin{enumerate}}
\newcommand{\een}{\end{enumerate}}
\newcommand{\beqn}{\begin{equation}}
\newcommand{\eeqn}{\end{equation}}
\newcommand{\bea}{\begin{eqnarray*}}
\newcommand{\eea}{\end{eqnarray*}}
\newcommand{\bpf}{\begin{proof}}
\newcommand{\epf}{\end{proof}\ms}
\newcommand{\ms}{\medskip}

\newcommand{\mI}{\mathcal I}
\newcommand{\bmC}{\boldsymbol{\mathcal C}}
\newcommand{\E}{\mathbb{E}}
\newcommand{\citep}{\cite}
\newcommand{\citet}{\cite}
\newcommand{\citeyear}{\cite}
\newcommand{\citeauthor}{\cite}

\begin{document}
\title{Kernel-estimated Nonparametric Overlap-Based Syncytial Clustering}
\author{Israel~Almod\'ovar-Rivera~and~Ranjan~Maitra

  \thanks{I. Almod\'ovar-Rivera is with the Department of
    Biostatistics and Epidemiology at the University of Puerto Rico,
    Medical Science Campus, San Juan, Puerto Rico, USA.}
  \thanks{R.Maitra is with the Department of Statistics, Iowa State
    University, Ames, Iowa, USA.}
  \thanks{This research was supported in part by the
    National Institute of Biomedical Imaging and Bioengineering (NIBIB) of the National
Institutes of Health (NIH) under its Award No. R21EB016212,
I. Almod\'ovar-Rivera also acknowledges receipt of a fellowship from 
Iowa State University's Alliance for Graduate Education and the
Professoriate (AGEP) program for underrepresented graduate students
in STEM fields.
R. Maitra also acknowledges support from the United States Department
of Agriculture (USDA) National Institute of Food and Agriculture
(NIFA) Hatch project IOW03617 
The content of this paper however is solely the responsibility of the 
authors and does not represent the official views of either the
NIBIB, the NIH, the NIFA or the USDA.
}

}

\markboth{}%
{Almod\'ovar-Rivera and Maitra: Nonparametric Syncytial Clustering}


\IEEEcompsoctitleabstractindextext{%

    \input{abstract}

  \begin{IEEEkeywords}
BATSE, DEMP, DEMP+, DBSCAN*, density peaks algorithm,
  GRB, GSL-NN, $k$-clips, $k$-means, $k_m$-means, kernel density estimation, KNOB-SynC,
  MixModCombi, MGHD, MSAL, overlap, PGMM, SDSS, spectral clustering, TiK-means 
\end{IEEEkeywords}}
\maketitle

\IEEEdisplaynotcompsoctitleabstractindextext

\input{introduction}


\input{methodology}


\input{experiments}

\input{application}

\input{discussion}

\bibliographystyle{IEEEtran}
\bibliography{references}
\appendices
\input{new-proof}
\input{appendix-2d}

\end{document}


%% file: abstract.tex
\begin{abstract}
  Commonly-used clustering algorithms usually find
  ellipsoidal, spherical or other regular-structured clusters, but are
  more challenged  when the underlying groups lack formal structure or
  definition. Syncytial clustering is the name that we introduce for
  methods that merge  
groups obtained from standard clustering algorithms in order to reveal
complex group structure in the data. Here, we develop a distribution-free
fully-automated syncytial clustering algorithm that can be used with
$k$-means and other algorithms. Our approach estimates the cumulative
distribution function of the normed residuals from an appropriately
fit $k$-groups model and calculates  the estimated nonparametric
overlap between each pair of clusters. Groups with high pairwise
overlap are merged  as long as the estimated
generalized overlap decreases. Our methodology is always a top
performer in identifying groups with regular and irregular structures 
in several datasets and can be applied to datasets with scatter or
incomplete records. The approach is also used to identify the
distinct kinds of gamma ray bursts in the   
Burst and Transient Source Experiment 4Br catalog and  
the distinct kinds of activation in a functional Magnetic Resonance 
Imaging study.
\end{abstract} 

%% file: introduction.tex
\section{Introduction}\label{sintro}
Cluster analysis~\citep{ramey85,mclachlanandbasford88,kaufmanandrousseuw90,everittetal01,  melnykovandmaitra10,  xuandwunsch09,bouveryonetal19} is an
unsupervised learning method that partitions datasets into distinct
groups of homogeneous observations. 
Finding such structure in the absence of group information
can be challenging but is important in many applications, such as
taxonomical  classification~\citep{michenerandsokal57}, market
segmentation~\citep{hinneburgandkeim99}, software
management~\citep{maitra01} and so on. As such, a number of methods, ranging
from the
heuristic~\citep{johnson67,everittetal01,jainanddubes88,forgy65,macqueen67,kaufmanandrousseuw90} to the more formal, 
 model-based~\citep{titteringtonetal85,mclachlanandpeel00,melnykovandmaitra10, mcnicholas16,bouveryonetal19} approaches have been proposed and implemented.  

Most common clustering algorithms, whether model-agnostic methods like
$k$-means~\citep{macqueen67,hartiganandwong79,lloyd82} or model-based approaches
such as Gaussian mixture models~\citep{fraleyandraftery02,melnykovandmaitra10} 
yield clusters with regular dispersions or
structure. For instance, the $k$-means algorithm is   geared towards
finding homogeneous spherical clusters or   spherically-dispersed
groups of equal radius. Such algorithms are not designed to find
general-shaped or structured groups, therefore, many additional approaches have been suggested to identify
irregularly-shaped groups~(see, for example, \cite{dhillonetal04,fredandjain05,vonluxburg07,baudryetal10, hennig10, melnykov16,petersonetal17}). 
 Kernel $k$-means clustering~\citep{dhillonetal04}
enhances the $k$-means algorithm  by  using a kernel function
$\phi(\cdot)$ that nonlinearly maps the original (input) space to a
higher-dimensional feature space where it may be possible to linearly
separate clusters that were not linearly separable in the original
space. Spectral  
clustering~\citep{vonluxburg07} uses $k$-means on the first few
eigenvectors of a  
Laplacian  of the similarity matrix of the data. Both  methods need the number of clusters to be provided:
in the case of spectral clustering, \citet{vonluxburg07} 
suggests estimating this number as the one  with the highest gap
between successive eigenvalues.

A separate set of approaches modifies the distribution of the mixture
components in model-based clustering (MBC) by replacing the
commonly-used multivariate Gaussian component with other more general
distributions. Some of these approaches simply add dimension reduction
in the form of factor
models~\citep{ghahramaniandhinton97,mcnicholasandmurphy08} through
parsimonious Gaussian mixture models (PGMM). More generally,
\citet{franczaketal13} propose MBC using a 
mixture of asymmetric shifted Laplace distributions (MixSAL) while
\citet{browneandmcnicholas15} suggest using a mixture of generalized
hyperbolic distributions (MixGHD). These approaches more fully exploit MBC
but can be CPU intensive and are somewhat limited in capturing complex
structures.   

Evidence accumulation clustering or EAC \citep{fredandjain05} combines 
results from multiple runs of the $k$-means algorithm with the
underlying rationale that each partitioning provides
independent evidence of structure that is then extricated by
cross-tabulating the relative frequencies (out of the multiple
partitionings) that each observation pair is in the same group. This
relative frequency table serves as a similarity matrix for hierarchical 
clustering: however, implementation of this method can be
computationally demanding in terms of CPU speed and memory.
\citet{stuetzleandnugent10} developed a nonparametric clustering
approach under the premise that each group corresponds to a mode of 
the estimated  multivariate density of the observations. 
The high-density modes are located and hierarchically clustered with
dissimilarity between two modes calculated in terms of the lowest
density or number of common points in each mode's domain of
attraction. The ``density-based spatial clustering algorithm of
applications with noise'' (DBSCAN) algorithm~\citep{esteretal96}
groups together points 
in high-density regions while identifying points in low-density regions as
outliers. A refinement~(DBSCAN*, by~\citep{campelloetal13})
follows the same principle but  classifies so-called border
observations as outliers.  Both algorithms depend  on the minimum
cluster size and {\em reachability distance}, and also on a cut-off to
determine the border and outlying observations. The authors suggest
setting this cutoff at the knee of a plot of the $k$-nearest neighbor
distances of the observations.  In a
similar vein, ~\citet{rodriguezandlaio14} developed a fast Density
Peaks (DP) algorithm 
to determine cluster centers and find outliers while considering the
local density of each observation. DP uses the estimated multivariate 
density in order to classify observations into outliers and does 
not rely on an explicit cut-off value, but other parameters need
to be subjectively specified or estimated to graphically decide on the
number of groups. 
These methods all rely on density  estimates and are not immune from
the ravages of the curse of dimensionality. 

More recent work~\citep{baudryetal10, hennig10,
  melnykov16,petersonetal17} proposed merging groups found using MBC
or $k$-means. Such methods fall into the category of what we introduce in this paper as syncytial clustering algorithms, because they yield a cluster structure 
resembling a {\em syncytium}, a term that in cell biology refers to a multi-nucleated mass of cytoplasm inseparable into individual
cells and that can arise from multiple fusions of uninuclear
cells. Syncytial clustering algorithms are similar in that they merge or
fuse groups that originally corresponded to mixture model
components or $k$-means or other regular-structured groups. Resulting
partitions have  
groups with potentially multiple well-defined and structured
sub-groups. We outline a few such algorithms next.

MBC is premised on the idea of a one-to-one correspondence between a
mixture component of given density form and group.  Such injective
mapping assumptions are not  always tenable so some
authors~\citep{baudryetal10, hennig10, melnykov16} model each group as a mixture of (one or more) components. Operationally, we have a
syncytial clustering framework where identified mixture components that are not very distinct from 
each  other are merged~\citep{baudryetal10, hennig10, melnykov16} into a
cluster.  \citet{baudryetal10}  successively merge mixture
component pairs that result in the highest change in entropy,
continuing for as long as the entropy increases. This method,
abbreviated here as MMC,  is   implemented in the {\sf R}~\citep{R} 
  package {\sc RMixModCombi}~\citep{RmixmodCombi}. \citet{hennig10} 
developed the {\it directly estimated misclassification probabilities}
(DEMP) algorithm to identify candidate components for merging.
The author argued that the best measure of group similarity should
relate to the classification probability and so  proposed that
clusters with the highest pairwise misclassification probabilities be
merged. The DEMP+ method \citep{melnykov16} 
mimics DEMP but replaces the misclassification probabilities of DEMP
with the overlap measure of \citet{maitraandmelnykov10} for Gaussian mixture components. DEMP+ uses Monte Carlo simulation to determine pairwise overlap between
merged components and uses thresholds on the maximum pairwise overlap
to determine termination. The sliding threshold was empirically
suggested to be chosen to be inversely related to dimension.  

The MBC algorithms offer a principled approach to the partitioning of
observations into groups but are more  demanding in CPU time and
perhaps unnecessary to  use when the objective is simply to find the
most appropriate grouping with no particular dogma regarding shape or
structure and where using $k$-means as a starting point for an initial
clustering may be a fairly plausible but faster alternative. Perhaps
recognizing this aspect, \citet{melnykov16} 
contended that  DEMP+  can be applied to $k$-means output by
assuming equal mixing proportions and 
homogeneous spherical dispersions in the mixture model. The basis for
this assertion is the framing of the $k$-means algorithm of 
\citet{lloyd82} as a 
Classification Expectation-Maximization (CEM) Algorithm~(see
\citep{fraleyandraftery98} for details). But $k$-means clustering 
makes hard assignments of each observation and, indeed,  most
commonly-used  statistical software programs, such as 
R~\citep{R} use the efficient  
\citet{hartiganandwong79} algorithm that handles computations quite
differently and sparingly than  \citet{lloyd82}. In this vein, 
\citet{petersonetal17} provided the K-mH algorithm to merge
poorer-separated $k$-means groups. Such groups are identified  as per an easily-computed index that uses normal theory with spherical dispersion assumptions. 
However, the K-mH algorithm has a large number of settings and parameters: using default values and rules-of-thumb provided by the authors, we have found that this  method performs well in many datasets but not  as well in many others. Therefore, it would be worth investigating other syncytial clustering algorithms that use $k$-means groupings for clustering efficiency while also reducing the need to tune multiple parameter settings. 

A separate issue is the impact of Gaussian mixture model
assumptions in methods such as DEMP+ when applied to
regular-structured groups  found using, say, the multivariate $t$-mixture  or
other appropriate models. A nonparametric
method not taking recourse to such distributional assumptions would be
desirable in addressing this shortcoming. This paper therefore
proposes the Kernel-estimated 
Nonparametric Overlap-Based Syncytial Clustering (KNOB-SynC) algorithm  that successively
merges groups from a well-optimized $k$-means solution until 
some objective and nonparametric data-driven cluster overlap measure
vanishes or is no longer reduced. This  measure is calibrated through
the generalized overlap~\citep{maitra10,melnykovandmaitra11,melnykovetal12} calculated
using smooth estimation of the cumulative distribution function (CDF)
developed in Section~\ref{sec:methodology}. Our algorithm is illustrated and
comprehensively evaluated in Section~\ref{sec:eval}. Although motivated using
$k$-means, the method is general enough to apply to the output of
other partitioning algorithms, such as clustering using the \citet{mahalanobis36}
distance, or in scenarios with scatter~\citep{maitraandramler09} or
incomplete
records~\citep{lithioandmaitra18}. Section~\ref{sec:application} also
applies our methodology to two interesting settings: in the first
case, we identify the differents kinds of gamma ray
bursts in the  most recent Burst and Transient Source Experiment
(BATSE) 4Br catalog. Our second application uses KNOB-SynC to identify 
activation from single replications of a 
functional Magnetic Resonance Imaging (fMRI) study obtained from a
right-hand finger tapping experiment performed by a right-hand-dominant
male. We find our results to both be interpretable and with greater
reproducibility than current methods. 
The paper concludes with some discussion. An appendix provides
mathematical proofs for our derived theoretical properties of smooth 
estimation of the CDF using asymmetric kernel density estimation and
detailed graphical illustrations of experimental performance on
two-dimensional (2D) datasets and numerical summaries of 
performance on all datasets.

%% file: methodology.tex
\section{Methodological Development} \label{sec:methodology}
\subsection{Problem Setup}
Let $\bXi = \{\bX_1,\bX_2,\ldots,\bX_n\}$ be a random sample of $n$ $p$-dimensional observations, with each 
\begin{equation}
\bX_i \sim \prod_{c=1}^C  [f_c(\bx)]^{\zeta_{ic}},
\label{eq:gen.cluster}
\end{equation}
where $C$ is the number of groups, $\zeta_{ic}= \mI_{(\bX_i\in \bmC_c)}$ with $\mathcal I_{(\mathcal  Z)}=1$ if $\mathcal Z$ holds and $0$ otherwise, $f_c(\bx)$ is the
cluster-specific density of 
an observation in the $c$th cluster and  $\bmC_c$ is the set
of observations in the sample from that group. Our specification in
\eqref{eq:gen.cluster} refers to a hard clustering framework: we
marginally obtain a mixture model if we specify independent identical multinary prior distributions on each
$\zeta_{ic}$. Our objective is to  estimate $\zeta_{ic}$s (equivalently, $\bmC_c$s) for
each $c=1,2,\ldots,C$ with $C$ possibly unknown. We also assume that 
for each $c=1,2,\ldots,C$, the density $f_c(\bx)$ for  any $\bX_i\in
\bmC_c$ ({\em i.e.} $\zeta_{ic}=1$) can be further described by  
\begin{equation}
  f_c(\bx) =  \prod_{k=1}^{k_c}[ h(\norm{\bx-\bmu_k^{\bmC_c}})]^{\zeta^{\bmC_c}_{ik}},\label{eq:1cluster}
\end{equation}
where $h(\cdot)$ is defined on the positive half of the real line so
that $h(\norm{\bx})$ is a zero-centered density in $\R^p$ 
with spherical level hyper-surfaces. 
This means that each group in
the dataset can be further decomposed into multiple homogeneous
spherically-dispersed subgroups, and $\zeta_{ik}^{\bmC_c} = 1$ if
$\bX_i$ is in $\bmC_c$ and in the $k$th  subgroup inside $\bmC_c$, and zero otherwise. 
That is, we can model 
$\bX_i\in\bXi$ as $\bX_i\sim \prod_{c=1}^C\prod_{k=1}^{k_c}[ h(\norm{\bx-\bmu_k^{\bmC_c}})]^{\zeta^{\bmC_c}_{ik}}$, or equivalently as
\begin{equation}
\bX_i\sim \prod_{k=1}^K [ h(\norm{\bx-\bmu_k^\circ})]^{\zeta^\circ_{ik}},\label{eq:kmeans}
\end{equation}
where $\zeta^\circ_{ik}$ and $\bmu_k^\circ$ for $k=1,2,\ldots,K$ are
renumerations, respectively, of all the $\zeta^{\bmC_c}_{ik}$ and
$\bmu_k^{\bmC_c}$ for $k=1,2,\ldots,k_c, c=1,2,\ldots,C$. Therefore, 
$K=\sum_{c=1}^Ck_c$, $\zeta_{ic}=\sum_{k=1}^{k_c}\zeta_{ik}^{\bmC_c}$
for $c=1,2,\ldots,C$ and 
$\sum_{k=1}^K\zeta^\circ_{ik}\equiv\sum_{c=1}^C\sum_{k=1}^{k_c}\zeta^{\bmC_c}_{ik} =  1$
(however, both $K$ and $C$ are also unknown). The reformulation of
\eqref{eq:gen.cluster} in terms of \eqref{eq:kmeans} means that the 
$k$-means algorithm~\citep{forgy65,lloyd82,hartiganandwong79} can be
employed along with cluster-selection methods~(for
  example,~\citep{krzanowskiandlai88,sugarandjames03,maitraetal12}) to
obtain a first-pass clustering of the dataset where the observations
are partitioned into an estimated number ($\hat K$) 
of homogeneous spherically-dispersed groups. Our proposal is to
develop methods for identifying the supersets of these
$k$-means (homogeneous spherical) groups to obtain the 
clusters $\{\bmC_c;c=1,2,\ldots,C\}$ with $C$ also needing to be
estimated. These supersets will reveal the general-shaped clustering
structure in the data.

From the $\hat K$-groups solution, 
define the $i$th residual ($i=1,2,\ldots,n$) as
\begin{equation}
  \hat{\bepsilon}_i = \bX_i - \sum^{\hat K}_{k=1} \hat{\bmu}^\circ_k \hat{\zeta}^\circ_{ik};\label{residual}
\end{equation}
where $\hat{\bmu}^\circ_k$ is the multivariate mean vector of the
observations in the $k$th group and 
$\hat{\zeta}^\circ_{ik} = \mI_{(\bX_i\in \mbox{ $k$th $k$-means group})}$. From \eqref{residual}, we obtain
the normed residuals, that is, we obtain
\begin{equation}
  \hat{\Psi}_{i} =
\sqrt{\hat{\bepsilon}'_i\hat{\bepsilon}_i} = \norm{\bX_i - \sum^{\hat
    K}_{i=1}   \hat{\zeta}^\circ_{ik} \hat{\bmu}^\circ_k}
\label{normed.residual}
\end{equation}
for $ i =1,2,\ldots,n; k=1,2,\ldots,\hat K$. These $\hat{\Psi}_1,\hat{\Psi}_2,\ldots \hat{\Psi}_n$ may be viewed as a random
sample with density function $h(\cdot)$ and CDF $H(\cdot)$ and
having support in $[0,\infty)$. We now provide methods for estimating
$H(\cdot)$ under assumptions of a smooth CDF.

\subsection{Smooth estimation of the CDF of the normed  residuals}\label{meth:kde} 
We first introduce a smooth estimator for an univariate CDF.  Let $Y_1,Y_2,\ldots, Y_n$ be a
random sample having CDF $H(\cdot)$ and probability density function (PDF)
$h(\cdot)$. 
The natural and most common estimator is the empirical CDF
(ECDF) defined as
\begin{equation}
\hat{H}_n(y) = \frac{1}{n} \sum^n_{i=1} \mathbb{1}(Y_i \leq y) \label{eq:ECDF}.
\end{equation}
It is easy to see that $\hat{H}_n(y)$ is an unbiased estimator of
$H(y)$, that is, $\E[\hat{H}_n(y)] = H(y)$. Further, it converges
almost surely to the true CDF $H(\cdot)$. However, the ECDF
is a step function for any $n$ and so inappropriate for a smooth continuous CDF, even though it
is a smooth function in the limit as $n\rightarrow
\infty$~\citep{silverman86}. An alternative 
{\em kernel estimator}~\citep{rosenblatt56,parzen62,silverman86,wandandjones95} for
$H(\cdot)$  replaces the indicator function in~\eqref{eq:ECDF}  by its
smooth cousin. Strictly speaking, kernel density estimation is most
often employed in nonparametric contexts~\citep{silverman86} but  can
also be extended to smooth CDF estimation by integrating over the domain of the
kernel. Let $G(y) =  \int^y_{-\infty}K(u) du $ be the CDF of a kernel
function $K(\cdot)$. The kernel CDF estimator is 
then defined as
\begin{equation}
\hat{H}(y;b) = \frac{1}{n} \sum^n_{i=1} G\left(\frac{y - Y_i}{b}\right), \label{eq:ske}
\end{equation}
 where $b$ is the {\em   bandwidth} or the {\em smoothing
   parameter}. Equation \eqref{eq:ske} makes the popular
 assumption of a symmetric kernel, the most common examples of which are 
 the  Gaussian and  Epanechnikov
 \citep{epanechnikov69,azzalini81,reiss81} kernels.  
 However,  using a symmetric kernel when the  support of the  distribution is not on 
 the entire real line (as is the case with our  normed residuals) causes 
 weights to be assigned outside the domain of the 
 observations, resulting in boundary bias~\citep{bouezmarniandscaillet05}.
 So \citet{chen00} proposed using an asymmetric kernel in
 \eqref{eq:ske} based on the 
 gamma density, with behavior similar to the Gaussian kernel and  a
 comparable rate of convergence in terms  
 of the mean squared error. However, \citet{chen00}'s estimator is not
 a valid density for finite sample sizes~\citep{jeonandkim13} so we
 consider the Reciprocal Inverse Gaussian 
 (RIG) kernel density estimator~\citep{scaillet04}
\begin{equation}
\hat{h}(y;b) = \frac{1}{n} \sum^n_{i=1} K \left(y; Y_i, b\right), \label{eq:aked}
\end{equation}
with $K(y;Y_i,b) =   \frac{1}{\sqrt{2 \pi b Y_i}} \exp\{
-\frac{1}{2bY_i} [Y_i - (y-b)]^2 \}$. 
Since $K(y;Y_i,b)$ is a smooth function the CDF estimate is defined as $G(y;Y_i, b) =
\int^y_0 K(t;Y_i,b) d t$. Then, integrating $\hat{h}(y;b)$ with respect
to $y$  yields the  smooth CDF estimator
\begin{equation}
\hat{H}(y;b)=\frac{1}{n} \sum^n_{i=1} G\left( y;Y_i, b\right) = \frac{1}{n} \sum^n_{i=1}  \left[  \Phi \left(\frac{Y_i +b}{\sqrt{Y_i b}}\right) - \Phi \left(\frac{Y_i  -(y-b)}{\sqrt{Y_i b}}\right)\right], \label{eq:ake}
\end{equation}
where $\Phi(\cdot)$ is the standard Gaussian CDF. An added benefit of
using the RIG kernel over the gamma kernel is that the estimated CDF
is in closed form and can be readily evaluated using standard software. 
We now investigate some theoretical properties of the asymmetric RIG kernel CDF estimator.  
Before proceeding however, we revisit the definition of the Inverse
Gaussian and RIG densities for the sake of completeness and to fix ideas.
\begin{definition}
\label{def.ig}
A nonnegative random variable $U_{\mu,\lambda}$ is said to arise from the 
Inverse Gaussian distribution with parameters $(\mu, \lambda)$ if it
has the density 
\begin{equation}
  r(u;\mu,\lambda) = \begin{cases} \frac{\sqrt{\lambda}}{u\sqrt{2\pi
      u}}\exp{\left\{-\frac\lambda{2\mu}\left(\frac      u\mu - 2
        +\frac\mu u\right)\right\}}, &  u > 0 \\
0 & \mbox{ otherwise.} \end{cases}
\label{eq.ig}
\end{equation}
Notationally, we write $U_{\mu,\lambda}\sim
\mbox{IG}(\mu,\lambda)$. Also, we have $\E(U_{\mu,\lambda}) = \mu$ and
$\V\mbox{ar}(U_{\mu,\lambda})  ={\mu^3}/\lambda$. 
\end{definition}
\begin{definition}
\label{def.rig}
A nonnegative random variable $V_{\mu,\lambda}$ is said to be from the
Reciprocal Inverse
Gaussian distribution with parameters $(\mu, \lambda)$ if it has the
density 
\begin{equation}
  s(v;{\mu,\lambda}) = \begin{cases} \frac{\sqrt{\lambda}}{\sqrt{2\pi
      v}}\exp{\left\{-\frac\lambda{2\mu}\left(     v\mu - 2
        +\frac1{\mu v}\right)\right\}}, &  v > 0 \\
0 & \mbox{ otherwise.} \end{cases}
\label{eq.rig}
\end{equation}
Notationally, $V_{\mu.\lambda}\sim\mbox{RIG}(\mu,\lambda)$. Further,
 $V_{\mu.\lambda}$ is equivalent in law to $1/U_{\mu,\lambda}$ where 
 $U_{\mu,\lambda}\sim
\mbox{IG}(\mu,\lambda)$ and 
$\E(V_{\mu,\lambda}) = 1/{\mu}+1/\lambda$ while $\V\mbox{ar}(V_{\mu,\lambda})  = (\lambda+2\mu)/(\lambda^2\mu)$.
\end{definition}
We now develop some properties of the
asymmetric kernel RIG to estimate the CDF.
\begin{lemma}
Let $Y_1,Y_2,\ldots,Y_n$ be independent identically distributed
nonnegative-valued random variables with CDF $H(y)$, and PDF
$h(y)$ that is infinitely differentiable. 
Also, consider the RIG kernel density defined by
$K(t;Y_i,b)  = \phi[(Y_i - (t-b))/\sqrt{bY_i}]/\sqrt{bY_i}$ where
$\phi(z)$ is the standard normal density 
evaluated at $z$. Consider estimating $H(y)$ using $\hat{H}(y;b)$ as defined in 
\eqref{eq:ake}. Then, as $b\rightarrow0$,  $\E[\hat{H}(y;b)] = H(y) + b[yh'(y)-h(y)]/2+o(b)\equiv H(y) +  \mathcal O(b)$ and $\V\mbox{ar}[\hat{H}(y;b)] 
 \approx  
H(y) (1 - H(y))/n - H(y)/(2n) - H(y) \sqrt{b}/(2n\sqrt{2\pi y}) +  o(\sqrt b)
\equiv
H(y) (1 - H(y))/n - H(y)/(2n) + \mathcal O(b)$.
\label{lemma:kernel.EVar}
\end{lemma}
\begin{proof}
  See Appendix \ref{proof-lemma3}.
  \end{proof}
Lemma \ref{lemma:kernel.EVar} shows that $\hat H(y;b)$ has lower variance than the ECDF and has point-wise Mean Squared
Error (MSE) at $y$ that is given by $\mbox{MSE}[\hat{H}(y;b)] = \V\mbox{ar}[\hat{H}(y;b)]+[\mbox{Bias}\{\hat{H}(y;b)\}]^2$. 

\subsubsection{Bandwidth selection}\label{sec:band}
\citet{scaillet04} minimized the  Mean Integrated Squared Error (MISE)
to provide a   rule-of-thumb bandwidth selector for the RIG kernel
density estimator of the form  
\begin{equation}
  \hat b = \left[\frac{2 \int^\infty_0 y^{-1/2}h(y) \mathrm{d} y}{\sqrt\pi\int^\infty_0 y^2\{h''(y)\}^2 \mathrm{d} y} \right]^{2/5} n^{-2/5}.
\label{eq:bw}
\end{equation} 
However, \eqref{eq:bw} involves knowledge of the true density $h(\cdot)$
and is directly unusable. \citet{scaillet04} proposed obtaining $\hat
b$ by assuming 
an initial parametric density, say $h(\cdot;\btheta)$, for $h(\cdot)$ and
estimating the parameters $\btheta$ of the density from the
sample. Exact derivations using a lognormal density for
$h(\cdot;\btheta)$ were provided~\citep{scaillet04} but this approach
has been found to produce estimates that are biased
downwards. We therefore adopt \citet{scaillet04}'s approach but use an
initial gamma density $h(y;\vartheta,\tau) =
\exp{(-y/\tau)}y^{\vartheta-1}\tau^\vartheta/\Gamma{(\vartheta)}$ for
$y>0$ and zero otherwise. Under this setup, $\int^{\infty}_0 y^{-1/2}
h(y;\vartheta,\tau) \mathrm dy =
\Gamma(\vartheta-{1}/{2})\sqrt{\tau}/\Gamma(\vartheta)$ and
$\int^{\infty}_0 y^{2} \{h''(y;\vartheta,\tau)\}^2 \mathrm dy =
(6\vartheta -4)(\vartheta -
1)\Gamma(2\vartheta)/\{4^\vartheta\tau^3\Gamma^2(\vartheta)(2\vartheta -
1)\}$.  Therefore, we have
\begin{equation}
\hat b = n^{-\frac25}\left[\frac{2^{2\hat\vartheta+1}\hat\tau^{7/2}(2\hat\vartheta-1)\Gamma(\hat\vartheta - \frac12)\Gamma(\hat\vartheta)  }{\sqrt\pi(6\hat\vartheta -4)(\hat\vartheta -
1)\Gamma(2\hat\vartheta)}\right]^{\frac25}
\label{eq:bw.gamma}
\end{equation}
with $\hat\vartheta$ and $\hat\tau$ 
estimated from the sample $Y_1,Y_2,\ldots,Y_n$ using, for example, the
method of moments. This  $\hat b$ is used in \eqref{eq:ake} to obtain
our smoothed RIG-kernel CDF estimator.

The development of this section,  when applied to the normed residuals
$\hat\Psi_1,\hat\Psi_2,\ldots,\hat\Psi_n$ (in place of
$Y_1,Y_2,\ldots,Y_n$), yields a smooth nonparametric
kernel-based estimator of their CDF. We use this kernel-estimated 
CDF in our development of the nonparametric estimation of the overlap
measure between  groups. 

\subsection{A nonparametric estimator of overlap between groups}
Overlap between two groups is an indicator of the extent to which
they are indistinguishable from each other. \citet{maitraandmelnykov10}
defined the pairwise overlap of two  mixture components as the sum of
the misclassification probabilities 
$\omega_{lk} \equiv \omega_{kl} = \omega_{l|k} + \omega_{k|l}$ with 
\begin{equation}
\omega_{l|k} = \PP[\bX \mbox{ is assigned to }\bmC_l \mid \bX \mbox{ is truly in } \bmC_k]. \label{eq:pairwiseoverlap}
\end{equation}
For any two mixture components with densities $f(\bx; \btheta_k)$
and $f(\bx; \btheta_l)$ and mixing proportions $\pi_k$ and $\pi_l$, we have
$$\omega_{k|l} = \PP\left[\pi_k f(\bx_i;\btheta_k) < \pi_l
  f(\bx_i;\btheta_l)|\bx_i \in f(\bx_i;\btheta_l)\right],$$ where
$\btheta_k$ and $\btheta_l$ are the parameter sets associated with the
$k$th and $l$th mixture components. 

\citet{maitraandmelnykov10} calculated \eqref{eq:pairwiseoverlap} for Gaussian
mixture densities, but the 
definition itself is general enough to include other clustering
situations including those as general as when we have cluster
distributions given by  densities of the type in 
\eqref{eq:gen.cluster}. 
For an equal-proportioned mixture of homogeneous spherical Gaussian
densities, \citet{maitraandmelnykov10} showed that 
$\omega_{k|l} = 
\Phi(\norm{\bmu_l-\bmu_k}/2\sigma)$ between the $k$th and the $l$th
cluster where $\Phi(\cdot)$ is the
standard Gaussian CDF, $\bmu_k$ and $\bmu_l$ are the
$k$th and the $l$th cluster means and $\sigma$ is the common
(homogeneous) standard deviation for each group, estimated unbiasedly
as $WSS_K/\{(n-K)p\}$ with $WSS_K$ being the optimized value of the
within-sums-of-squares (WSS) of the $K$-groups solution. The sum of $\omega_{k|l}$ and $\omega_{l|k}$
reduces to $\omega_{kl} = 
2\Phi(\norm{\bmu_l-\bmu_k}/2\sigma)$. The $k$-means
formulation of \eqref{eq:kmeans} can be viewed more
generally~\citep{maitraetal12} and extends beyond the case of  
Gaussian-distributed groups, so we develop nonparametric
methods for estimating the overlap measure. 
\subsubsection{Pairwise overlap between two $k$-means groups}
\label{2overlap.kmns.groups}
The pairwise overlap \eqref{eq:pairwiseoverlap} between
two groups can generally be calculated from $H_\Psi(\cdot)$  as
\begin{equation}
  \omega_{l|k} = \PP\left(  \norm{\bX-{\bmu}_l} <  \norm{\bX-{\bmu}_k}\mid \bX \in \bmC_k \right ) = 
1 - \PP\left(\Psi_k < \Psi_{l(k)}\right)
\label{eq:pairwiseoverlapNonParametriclk} 
\end{equation}
where $\Psi_k$ represents the normed residual obtained from the $k$th 
group, and $\Psi_{l(k)}$ represents the normed {\em pseudo-residual}
which we define as the norm of the remainder that is
obtained by subtracting the $l$th cluster mean $\bmu_l$ from an observation
$\bX\in\bmC_k$. Let $H_\Psi(y)$ be the RIG kernel-estimated smooth CDF
obtained using the bandwidth selected as per \eqref{eq:bw.gamma}.
Then, $\PP(\Psi_k < y)$ can be estimated using 
$\hat H_\Psi(y;\hat b)$ (where $\Psi$ in the subscript of $\hat H(\cdot;\cdot)$
denotes that the estimated CDF uses the normed residuals). However, the
calculation of $\PP\left(\Psi_k < 
  \Psi_{l(k)}\right)$ is not as straightforward. So we estimate
  $\PP\left(\Psi_k <   \Psi_{l(k)}\right)$ using a na\"ive average estimator
  \begin{equation}
  \hat{\PP}\left(\Psi_k < \Psi_{l(k)}\right) = \frac 1{n_k^\circ}
  \sum_{i=1}^n \hat\zeta_{ik}^\circ \hat H_\Psi(\norm{\bX_i -
    \hat\bmu^\circ_l};\hat b), 
\label{eq:plugin}
\end{equation}
  where $n_k^\circ =  \sum_{i=1}^n  \hat\zeta_{ik}^\circ
  $. The na\"ive estimator \eqref{eq:plugin} can be considered as an
  empirical estimator of $\bE[\hat H_\Psi
    (\norm{\bX_i-\hat\bmu_l};\hat b)
    \mid \hat\bmu_k,\hat\bmu_l, 
    \bX_i\in\mbox{ $k$th spherically-dispersed group }]$. Similar
estimates of $\omega_{k\mid l}$, and therefore $\omega_{kl}$, can be
obtained. We call this estimated overlap 
$\hat\omega_{kl}\equiv 
\hat\omega_{lk}$.
\subsubsection{Pairwise overlap between two composite groups}
\label{2overlap.comp.groups}
As described in \eqref{eq:1cluster}, a composite group is one that can
be further decomposed into sub-populations. We now extend the
definition of the pairwise overlap for such groups. 

Let $\omega_{\bmC_l\mid\bmC_k}$ be defined as in 
\eqref{eq:pairwiseoverlap} but for composite groups. That is, we use 
$\omega_{\bmC_l\mid\bmC_k}$  rather than $\omega_{l\mid k}$ in order
to specify that the overlap measure is between composite
clusters $\bmC_l$ and $\bmC_k$. Now
$\omega_{\bmC_l\mid\bmC_k} = 1-\PP[ \min_{r \in  {\bmC}_k}
\norm{\bX-\bmu_r} <  \min_{j \in \bmC_l} 
  \norm{\bX-\bmu_j} \mid \bX\in \bmC_k]$. Suppose now that
  $\bmC^\circ_{s\subset k}$ is the $s$th spherical sub-cluster of
  $\bmC_k$ with mean $\bmu_s^\circ$, $s = 1, 2, \ldots, \abs{\bmC_k}$, with
  $\abs{\bmC_k}$ being the number of spherical sub-clusters in
$\bmC_k$. We assume that if $\bX \in \bmC_{k}$, then 
  $\argmin_{r\in\{1,2,\ldots,\abs{\bmC_k}\}} \norm{\bX - \bmu_r^\circ } =
  s\subset k$ implies that $\bX$ is in the subgroup given by
  $\bmC_{s\subset k}$. Under this assumption, the density of $\bX$ is
  defined through its ($s$th) sub-cluster and so
\begin{equation}
\PP\left( \min_{r\in \bmC_k} 
\norm{\bX-\bmu_r} \leq y \mid \bX \in \bmC_k\right ) =  1 -
\PP\left(\min_{r\in\bmC_k} \Psi_r > y \right)
= 1 - \left[ 1 -
\PP\left( \Psi_r \leq y \right)\right]^{\abs{\bmC_k}}
\end{equation}
where $\Psi_r$ is a normed residual (obtained, for instance, from the
$k$-means solution) for the $r$th spherically-dispersed subgroup  in the $k$th
cluster. We use the RIG kernel distribution estimator to obtain
$\PP\left( \Psi_r < y \right)$. From \eqref{eq:pairwiseoverlap}, and
using the same ideas as in \eqref{eq:plugin} we get the na\"ive
estimator 
\begin{equation}
  \hat\omega_{\bmC_l\mid\bmC_k} = 
  \left[1 - \frac1{n_c} \sum_{i = 1}^{n_c}\hat\bzeta_{ic} \hat
    H_\Psi(\min_{r\in \bmC_l} \norm{\bX_i - \bmu_r};\hat b)
\right]^{\abs{\bmC_k}}
\label{eq:composite.plugin}
\end{equation}
and similarly for $\hat\omega_{\bmC_k\mid\bmC_l}$, from where we 
calculate 
$\hat\omega_{\bmC_l
  \bmC_k} \equiv \hat\omega_{\bmC_k \bmC_l} =\hat\omega_{\bmC_l\mid\bmC_k}
+\hat\omega_{\bmC_k\mid\bmC_l} $. Our definitions of $\bmC_k$s and
$\hat\omega_{\bmC_k\bmC_l}$ are consistent in the sense that if
$\bmC_k=\{k\}$ and $\bmC_l=\{l\}$ are both $k$-means groups, then
$\hat\omega_{\bmC_k\bmC_l}=\hat\omega_{kl}$. We use this equivalence
in the description of our KNOB-SynC algorithm in
Section~\ref{sec:algorithm} below.
\subsubsection{Summarizing overlap in a partitioning}
Our development so far has provided us with pairwise overlap measures
for $k$-means-type (Section~\ref{2overlap.kmns.groups}) and composite
(Section~\ref{2overlap.comp.groups}) groups. For a $K$-groups (whether
of the composite or $k$-means type) partitioning, we get $K\choose2$
pairwise overlap measures. Summarizing the pairwise overlap measures is
important to provide a sense of clustering complexity so
\citet{maitraandmelnykov10} originally proposed regulating $\check\omega$ 
(maximum of all pairwise overlaps) and $\bar\omega$ (average of
all $K\choose2$ pairwise overlaps) and demonstrated~(see
Figures 2 and 3 of~\citep{maitraandmelnykov10}) the ability to
summarize a wide range of cluster geometries. However, because
 specifying two measures simultaneously is cumbersome, later
versions of the {\sc CARP}~\citep{melnykovandmaitra11} and {\sc
  MixSim}~\citep{melnykovetal12} software packages borrowed ideas from
\citet{maitra10} to obtain the {\em generalized overlap}
$\ddot\omega=(\check\lambda_{\bOmega}-1)/(K-1)$ 
where $\check\lambda_{\bOmega}$ is the largest eigenvalue of the
(symmetric) matrix $\bOmega$ of pairwise  
overlaps $\omega_{l,k}$ ($\omega_{\bmC_k\bmC_l}$ for composite
groups) and with diagonal entries that are all 1. 
$\ddot\omega$ lies in [0,1] with zero
indicating perfect separation between all group densities and 1
indicating indistinguishability between any of them. In this paper, we 
obtain the estimated generalized overlap $\ddot{\hat\omega}$ using the
estimated matrix $\hat\bOmega$ with off-diagonal entries given by the
kernel-estimated pairwise overlaps $\hat\omega_{l,k}$ or
$\hat\omega_{\bmC_k\bmC_l}$, depending on whether we have simple
$k$-means-type or composite  groups. 

\subsection{The KNOB-SynC Algorithm}\label{sec:algorithm}
Having provided theoretical development for the machinery that we will
use, we now describe our multi-phased KNOB-SynC algorithm:
\begin{enumerate}
\item \label{kmns.phase} {\em The $k$-means phase:} This phase finds 
  the optimal partition  of the dataset in terms of
  homogeneous spherically-dispersed groups and has the following steps:
  \begin{enumerate}
    \item \label{kmns.init} For each $K\in \{1,2,\ldots, K_{\max}\}$,
      obtain $K$-means partitions initialized each of $nKp$ times with $K$
      distinct seeds randomly chosen from the dataset and run to
      termination. The best -- in terms of the value of the objective
      function (WSS) at termination -- of each set of $nKp$ runs is our
      putative optimal $K$-means partition for that $K\in
      \{1,2,\ldots, K_{\max}\}$. We use $K_{\max} =  \max\{\sqrt{n}, 50\} $. 
     \item When $n$ is small relative to $p$ (operationally, 
       $n<p^2$), use \citet{krzanowskiandlai88}'s KL criterion to
       decide on the optimal $K$. Otherwise, for larger $n$, we use the jump statistic    \citep{sugarandjames03} on the optimal $K$-means partitions
      ($K\in \{1,2,\ldots, K_{\max}\}$)    obtained in
      Step~\ref{kmns.init}  to determine the optimal $K$ (denoted by
      $\hat K$). In
      calculating the jump statistic, we have used $y = p/2$, which has become
      the default in most applications. We refer to
      \citet{sugarandjames03} for more detailed discussion on this
      choice of $y$. The 
      corresponding $\hat K$-means solution is the optimal
      homogeneous spherically-dispersed partition of
      the dataset.  This concludes the $k$-means phase of the
      algorithm. 
    \end{enumerate}
  \item \label{ol.phase} {\em The initial overlap calculation phase:}
    This phase starts with the output of Step~\ref{kmns.phase}.
    That is, we start with a structural definition of
    the dataset in terms of $\hat K$ optimal homogeneous
    spherically-dispersed groups. Our objective here is to
    calculate the overlap between each of these groups using nonparametric
    kernel estimation methods. We proceed as follows:
    \begin{enumerate}
    \item For each observation $\bX_i, i =1,2,\ldots,n$, compute its
      normed residual 
      $\hat\Psi_i = \sqrt{\hat{\bepsilon}'_i\hat{\bepsilon}_i}$ where
      $\hat\bepsilon_i$ is  defined as in~\eqref{residual}. Also, obtain the
      normed pseudo-residual $\hat\Psi_{i;l(k)} = \norm{\bX_i - \hat\bmu_l}$ for
      $\bX_i\in\bmC_k$, and $l\neq k\in \{1,2,\ldots, \hat K\}$. 
    \item Using the set of normed residuals
      $\{\hat\Psi_i;i=1,2,\ldots,n\}$, obtain its RIG-kernel-estimated
      CDF  using  \eqref{eq:ske} with bandwidth determined as per
      \eqref{eq:bw.gamma}.
    \item For any two groups $k\neq l \in \{1,2,\ldots,\hat K\}$,
      estimate the pairwise overlap $\hat  \omega_{lk} = \hat
      \omega_{l|k} + \hat \omega_{k|l}$, where 
      $\hat \omega_{l|k}$ and $\hat \omega_{k|l}$ are calculated
      using \eqref{eq:pairwiseoverlapNonParametriclk} and
      \eqref{eq:plugin}. We obtain the estimated overlap matrix
      $\hat\bOmega$ (with diagonal elements all equal to unity). For
      clarity, denote this overlap matrix as
      $\hat\bOmega^{(1)}$ and pairwise overlaps as
      $\hat\omega_{kl}^{(1)}\equiv \hat\omega_{\bmC_k\bmC_l}^{(1)}$.
 
    \item From the overlap matrix $\hat{\bOmega}^{(1)}$,
      calculate the generalized overlap $\genhatomega$. Call it
      $\genhatomega^{(1)}$. 
    \end {enumerate}
  \item\label{merge.phase} {\em The merging phase:}
    The merging phase is triggered only if some of the overlap
    measures between overlapping clusters 
    are more than the others (operationally, if
    $4\genhatomega^{(1)}\not\geq \check{\hat\omega}$ 
    where
    $\check{\hat\omega}$  is the maximum  of     the estimated
    pairwise overlaps) or if  $\genhatomega$ is
    not negligible, that is, if $\genhatomega^{(1)}\not\approx 0$
    (operationally  $\genhatomega^{(1)}\geq 10^{-5}$). 
    In that case, this phase merges groups, provides pairwise overlap
    measures between newly-formed composite groups, the updated overlap 
    matrix and the generalized overlap, continuing for as long as the
    generalized overlap keeps decreasing (by at least
    $10^{-5}$) or is not
    negligible. Specifically, this phase iteratively proceeds for
    $\ell=1,2,\ldots$ with  the following steps:
  \begin{enumerate}
  \item \label{merge.step1}
    Merge the groups with the maximum overlap and  every pair of
    groups that have individual pairwise overlaps 
    substantially larger than the generalized overlap
    $\genhatomega^{(\ell)}$.  That is, 
    merge every pair of groups $\mC_k$,  $\mC_l$, $k \neq l$ such that
    $\hat{\omega}_{lk}^{(\ell)} \equiv \check{\hat\omega}^{(\ell)}$
    or  $\hat{\omega}_{lk}^{(\ell)} > \kappa \genhatomega^{(\ell)}$,
    for some $\kappa$ as described in the comments section below.  
    Call the new merged group
    $\bmC_{\min(k,l)}$ and decrease the index labels of the groups with
    indices greater than     $\max(k,l)$. Decrement $\hat K$ by 1
    for every merged pair.
  \item Using \eqref{eq:composite.plugin}, update the pairwise overlap
    measures that have changed as a result of the merges in
    Step~\ref{merge.step1}. Call the updated measures
    $\hat\omega_{\bmC_k\bmC_l}^{(\ell+1)}$. Obtain the updated overlap
    matrix (call it $\hat\bOmega^{(\ell+1)}$) and the updated
    generalized overlap $\genhatomega^{(\ell+1)}$. Set $\ell\leftarrow \ell+1$. 
  \item The merging phase terminates if $\genhatomega^{(\ell)}>
    \genhatomega^{(\ell-1)}$, $\genhatomega^{(\ell)} \approx
    0$, or $\genhatomega^{(\ell)} \approx\check{\hat\omega}^{(\ell)}$.  The
    terminating $\hat K$ is the $\hat C$ of
    \eqref{eq:gen.cluster}.  
  \end{enumerate}
\item {\em Final clustering solution:} The grouping 
  $\{\bmC_1,\bmC_2,\ldots, \bmC_{\hat C}\}$ at the end of the merging
  phase is the final partition of the dataset. This gives us a
  total   of $\hat C$ general-shaped groups in the dataset. 
\end{enumerate}
\paragraph{Comments:} We provide some additional remarks on  KNOB-SynC
and relate it to other algorithms for finding general-shaped clusters
and settings:
\begin{enumerate}
\item The $k$-means phase finds regular-structured (more specifically,
  homogeneous spherical) groups and, in this regard, is similar to the
  initial stages of K-mH
  \citep{petersonetal17} and EAC \citep{fredandjain05}. However, 
    EAC repeats $k$-means with fixed $K$ several times and is built
    upon the premise that each $k$-means run does not end up with the
    same clustering, especially when we do not have underlying homogeneous
    spherically-dispersed groups. On the other hand, K-mH uses a
    separability index built on Gaussian assumptions for each cluster
    and has a large number of user-specified parameters. KNOB-SynC
    uses nonparametric CDF estimation with a plugin 
    bandwidth selector and a na\"ive average estimator to
    calculate the     overlap between  
    spherically-dispersed groups and a na\"ive estimator for the
    overlap between composite groups. Our methodology has one parameter
    ($\kappa$) that is chosen in a completely data-driven
    framework. No parameter requires fine-tuning by the
    practitioner.  Also, the number of     general-structured groups
    is decided upon  
    termination that is objectively declared whenever the generalized overlap
    vanishes or does not go down further.
  \item As with MMC, DEMP or DEMP+, the use of cluster
    distributions in the overlap calculations simplifies and keeps
    practical computations even for large datasets. In contrast, 
    EAC, DBSCAN,    DBSCAN$^*$, DP 
    and K-mH require memory-intensive cross-tabulation of the entire
    dataset across multiple clusterings because  $n\times n$
    frequency tables need to be calculated and/or stored.
  \item KNOB-SynC uses a na\"ive estimator to update the overlap
    between composite groups, unlike DEMP+ which uses Monte Carlo
    simulations and is     slower. Further, DEMP+ uses the maximum
    overlap that is very 
    sensitive to individual pairwise overlap measures while KNOB-SynC
    uses the generalized overlap measure~\citep{maitra10} that
    provides a nonlinear summary of all the individual pairwise overlaps. 
  \item Unlike DEMP or DEMP+, the stopping criterion of KNOB-SynC
    is data-driven, thus allowing for the possibility of obtaining
    well-separated and less well-separated partitionings as supported
    by the data.  Our algorithm also has the potential, unlike MMC,
    DEMP or DEMP+, to merge multiple pairs of groups in a step.
  \item KNOB-SynC uses nonparametric CDF estimation but does so in
    univariate space by exploiting the inherent spherically-dispersed 
    structure (ellipsoidal in the case of clustering with the
    Mahalanobis distance) of the sub-clusters. Therefore, it has
    greater immunity against the curse of dimensionality that bedevils
    multivariate density estimation that is used in algorithms such as
    DBSCAN$^*$ and DP. 
    \item The parameter $\kappa$ determines the types of composite groups
    that  are formed. For larger values of $\kappa$, we have
    groups formed by merging a few pairs at each iteration while 
    smaller values  $\kappa$ prefer many simultaneous mergers. (For $\kappa\rightarrow\infty$, no merging is 
    possible.) In
    the first case, we  expect to have stringy groups while in
    the second case, we   find clusters that are irregular-shaped but
    less stringy.  A data-driven approach to choosing $\kappa$, that 
    we adopt, runs the     algorithm with different values of
    $\kappa=1,2,3,4,5,\infty$ and uses the final partitioning with the 
    smallest terminating $\ddot{\hat\omega}$     as the optimal
    clustering.
  \item Unlike other syncytial clustering algorithms like DEMP, DEMP+
    or K-mH, KNOB-SynC allows for the possibility of     multiple
    pairs of groups to be merged at an iteration.  
  \item Our initial stage uses $k$-means for speed and efficiency that
    also allows us to explore larger candidate values of $K$. However,
    the approach could very well have been  used with  clustering
    algorithms obtained using, say, the {\em generalized Mahalanobis 
    distance}. 
    The overlap calculations are then easily modified. To see this,
    suppose that the generalized Mahalanobis distance between two points $\bx$ and
    $\by$ is given by $d_{\bGamma}(\bx,\by) =
    (\bx-\by)'\bGamma^{-}(\bx-\by)$, where $\bGamma$ is any
    appropriate nonnegative-definite matrix with (say, Moore-Penrose)
    inverse given by $\bGamma^-$. (A positive definite $\bGamma$
    leads to the usual Mahalanobis distance.) Under the generalized 
    Mahalanobis distance framework, \eqref{eq:pairwiseoverlap} reduces to
\begin{equation}
\begin{split}
  \omega_{l|k} & =  \PP\left( \norm{(\bGamma^-)^{1/2}(\bX-{\bmu}_l)} <
  \norm{(\bGamma^-)^{1/2}(\bX-{\bmu}_k)}\mid \bX \in \mbox{ $k$th    group } \right )\\
               & = 1- \PP\left( \norm{(\bGamma^-)^{1/2}(\bX-{\bmu}_k)}
  < \norm{(\bGamma^-)^{1/2}(\bX-{\bmu}_l)}\mid \bX \in \mbox{  $k$th   group }\right ),
\end{split}
\label{eq:overlapMahalanobis} 
\end{equation}
which means that the problem reduces to the Euclidean case if we use
what we here refer to as the normed Mahalanobis-free residuals (and
pseudo-residuals). Operationally, this is equivalent to obtaining
$\hat\bvarepsilon_i =
(\bGamma^-)^{1/2}(\bX-\sum_{i=1}^K\hat\zeta_{ik}^\circ\hat{\bmu}_k^\circ)$,
and replacing the $\hat\epsilon_i$ with $\hat\bvarepsilon_i$ in the
calculation of \eqref{normed.residual} and proceeding as before. 
This framework also includes the case when we  scale each
variable before 
clustering, as happens when the features are on vastly different
scales, or when we use principal components (PCs) as our
clustering variables -- we illustrate these scenarios in
Sections~\ref{sec:app.image}
and~\ref{sec:app.GRB}.   

  \item Our algorithm accommodates clustering scenarios in the presence
    of scatter as provided, for instance, by the output of the $k$-clips
    algorithm of~\citet{maitraandramler09}. Scatter observations are
    those that are unlike any other and may be considered as
    individual groups in their own right. KNOB-SynC incorporates these
    scatter observations as individual clusters in addition to the
    groups found from the output of the $k$-clips algorithm and
    proceeds with the overlap calculation and merging phases as
    described earlier in this section. We illustrate this scenario in
    Section~\ref{sec:att.image}. 
  \item Datasets often have incomplete records with missing observations in
    some features. The 
    $k_m$-means algorithm~\citep{lithioandmaitra18} provides a
    $k$-means type algorithm for Euclidean distance clustering in this
    setting. Then, instead of $k$-means,     KNOB-SynC can incorporate
    results from $k_m$-means in the first stage. For the 
    incompletely-observed records, we calculate the 
    rescaled normed residual in the presence of 
    missing information by removing the missing value from their
    calculation and re-weighting it appropriately. Specifically, we
    calculate the $i$th rescaled normed residual as
    \begin{equation}
      \hat\varPsi_i= \frac{p}{p_i}\sum_{l=1}^{p_i}(X_{i{j_l}} - \hat\mu_{k{j_l}})^2 
      \label{normed.missing.residual}
    \end{equation}
    where $X_{i{j_1}},X_{i{j_2}},\ldots,X_{i{j_{p_i}}}$ represent the
    $p_i$ available features for the $i$th record that has been
    assigned to the $k$th spherically-dispersed sub-group with
    estimated mean $\hat\bmu_k$.  
    Similar arguments allow for the calculation of the rescaled normed
    pseudo-residual $\hat\varPsi_{il(k)}$.
    The use of the nonparametric CDF estimator in KNOB-SynC provides
    us with the flexibility to calculate the initial overlap estimates 
    from these scaled-up normed residuals. The merging phase and
    termination criteria of our algorithm remain
    unchanged. Section~\ref{sec:app.sdss} illustrates
    KNOB-SynC on a dataset with incomplete records.  
 \item The use of nonparametric methods in the overlap calculations
    means that a large number of methods may be possible to use in the initial
    partitional phase. The method can also potentially be modified to apply to
    other kinds of datasets. For instance, the initial clustering
    can be done for categorical datasets using
    $k$-modes~\citep{huang97a,huang98,chaturvedietal01,dormanandmaitra20} and
    then the Generalized or Gaussianized     Distributional
    Transform~\citep{ruschendorf13,daietal19} 
    and copula model~\citep{nelsen06}  can potentially be applied to  each
    cluster to obtain
    numerical-valued residuals for use with our overlap estimation and
    calculations.
  \end{enumerate}
Having proposed our KNOB-SynC algorithm, we now illustrate and
evaluate its performance in relation to a host of competing methods.

%% file: experiments.tex
\section{Performance Evaluations}\label{sec:eval}
We first illustrate the performance of KNOB-SynC on the 2D
\texttt{Aggregation} dataset of \citet{gionis07} and then follow 
with more detailed performance evaluations on a large number of
datasets usually used to evaluate competing algorithms in the
literature. These datasets range from two to many 
dimensions. We compare our methods with a wide range of
suitors. These rival methods are the syncytial clustering techniques of
K-mH~\citep{petersonetal17} using author-supplied {\sf R} code,
MMC~\citep{baudryetal10} as implemented in  
the {\sf R package} {\sc MixModCombi}, DEMP~\citep{hennig10} 
using the {\sf R} package {\sc fpc} and DEMP+~\citep{melnykov16}.
We also evaluate performance with
EAC~\citep{fredandjain05} and GSL-NN~\citep{stuetzleandnugent10} using
publicly available author-supplied code. We also apply two common
connectivity-based techniques  
of spectral and kernel $k$-means clustering. Both these methods need the
number of groups to proceed: for spectral clustering we decide this 
number to be the one with  the highest gap in successive eigenvalues
of the similarity matrix~\citep{vonluxburg07}. For kernel $k$-means,
we set $K$ to be the true value: we recognize that our evaluation of kernel
$k$-means potentially provides this method with an unfair advantage, however, we
proceed in this fashion in order to understand the best case scenario
of this competing method.
Finally, we also compare our method's performance relative to
DBSCAN$^*$ as implemented 
in the  {\sf R} package {\sc dbscan} ~\citep{hashlerandpiekenbrock18}, 
DP clustering~\citep{rodriguezandlaio14} as implemented
in the {\sf R} package {\sc densityClust}~\citep{pendersenetal17},
PGMM~\citep{mcnicholasandmurphy08} using the {\sf R} package {\sc 
  pgmm}~\citep{mcnicholasetal18}, MSAL using the {\sf R} package {\sc
  MixSAL}~\citep{franczacketal18} and MGHD using the {\sf R} package
{\sc MixGHD}~\citep{tortoraetal19}. Many of these algorithms have
multiple parameters that need to be set -- in our experiments, we use
the default settings where guidance for choosing these parameters is not
explicitly available. Also, the DBSCAN$^*$ and DP clustering algorithms
 identify scatter/outliers that by definition are those
observations that are unlike any other in the dataset, so we follow 
\citet{maitraandramler09} in considering them as individual singleton
clusters in our performance assessments. Performance for each method is
evaluated by \citet{hubertandarabie85}'s
adjusted Rand index ${\mR}$ measured between the true partition and
the estimated final partitioning. In general, $\mR\leq1$: values
closer to 1 indicating greater similarity between 
partitionings and good clustering performance. The index takes values
farther from 1 as performance becomes poorer and is expected to take
a value of zero for a random assignmment. The
index $\mR$ can  take arbitrarily negative values, but as very
helpfully pointed out to us by a
reviewer, 
the probability of observing $\mR < -1$ is relatively
small~(see~\citep{steinley04} for further discussion on the characteristics of
  this index). 
\subsection{Illustrative Example: the \texttt{Aggregation} dataset}
The 2D \texttt{Aggregation} dataset \citep{gionis07} has $n=788$
observations from $C=7$ groups of different characteristics. Figure~\ref{fig:merging.step.1}  
\begin{figure*}[h]
  \vspace{-0.1in}
  \centering
  \mbox{
    \hspace{-0.1in}
    \subfloat[$k$-means phase: $\hat K=14
      $]{\includegraphics[width=0.3333\textwidth]{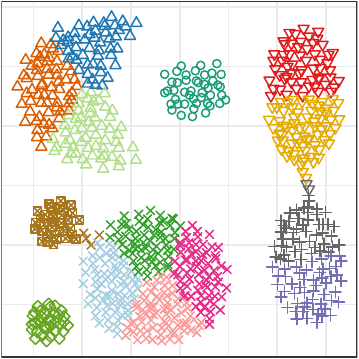}\label{fig:kmeans.step.1}}
    \subfloat[First merging phase: $\hat C = 8
      $]{\includegraphics[width=0.333\textwidth]{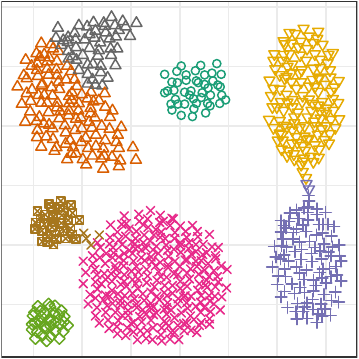}\label{fig:kmeans.step.3}}
    \subfloat[Final solution: $\hat C = 7
      $]{\includegraphics[width=0.333\textwidth]{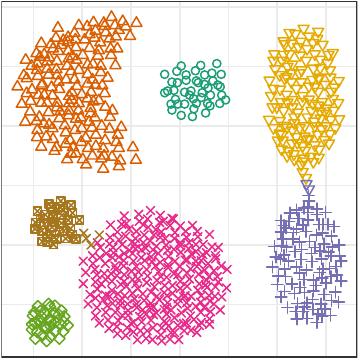}\label{fig:kmeans.step.5}}}
\mbox{
  \hspace{-0.1in}
  \subfloat[$\genhatomega=0.0013$]{\includegraphics[width=0.3\textwidth]{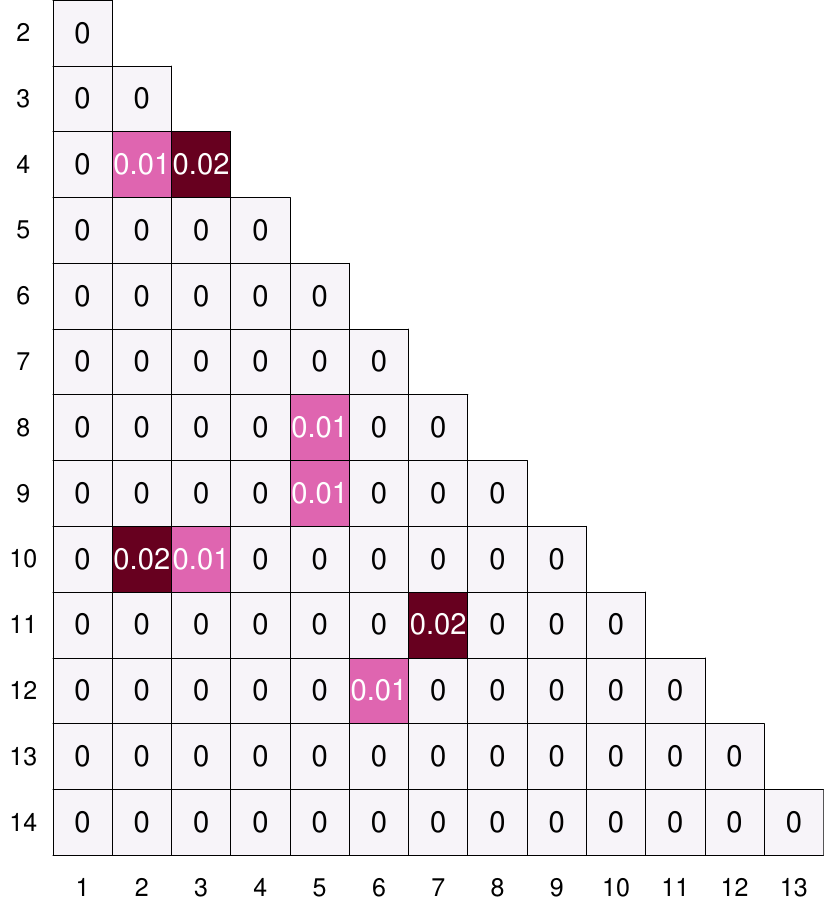}\label{fig:overlap.step.2}}
  \subfloat[$\genhatomega=0.0008$]{\includegraphics[width=0.3\textwidth]{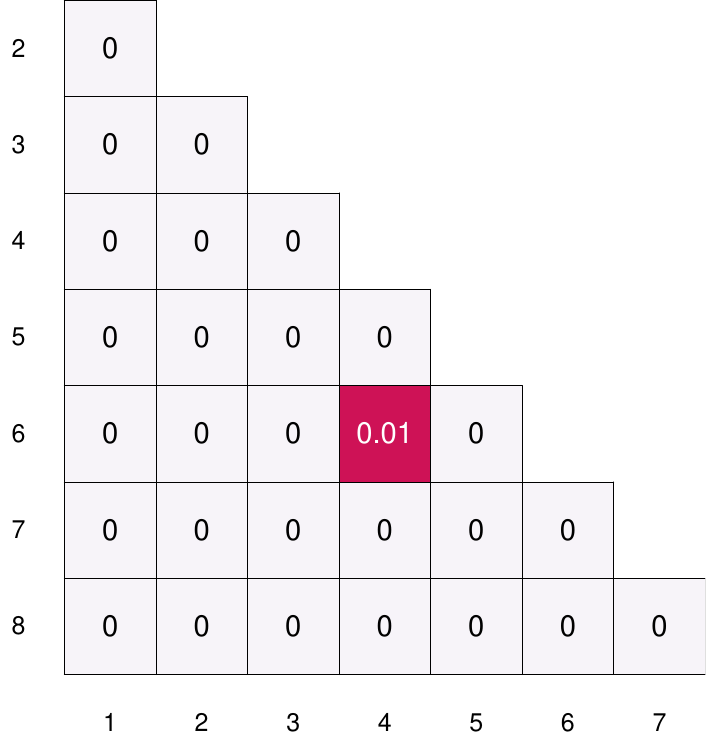}\label{fig:kmeans.step.4}}
\subfloat[$\genhatomega=\!10^{-6}$]{\includegraphics[width=0.4\textwidth]{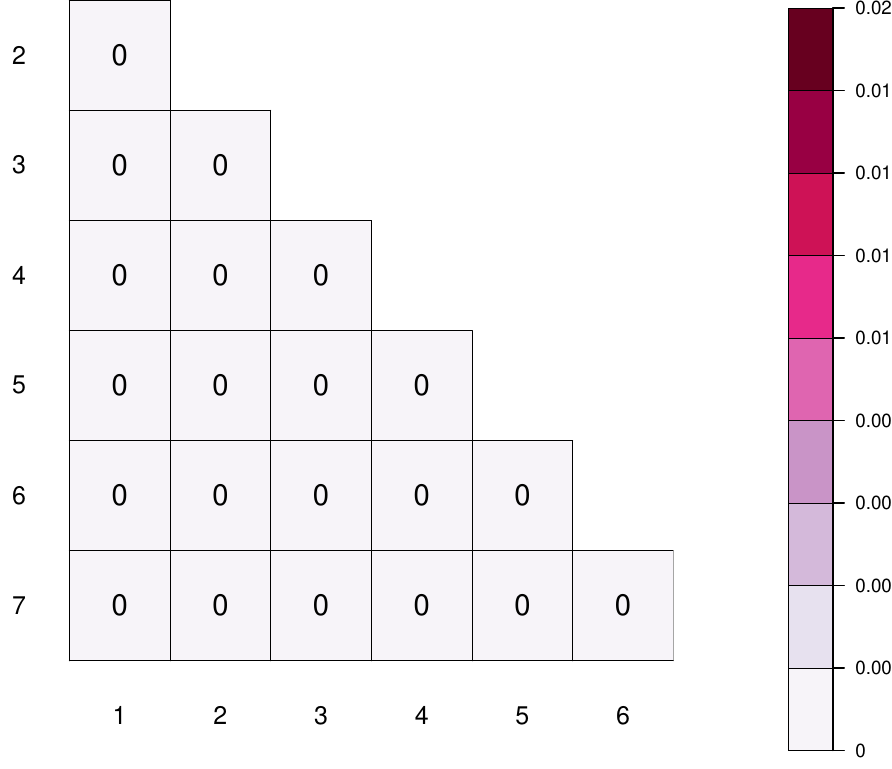}\label{fig:kmeans.step.6}}}
\caption{Illustrating all three stages of the KNOB-SynC algorithm
  on the {\tt Aggregation} dataset: Results of (a) the $k$-means
  phase, (b) the first merging phase and (c) the second (and final)
  merging  phase of the algorithm. In these and all such subsequent
  figures, character denotes true class membership while color indicates
estimated class membership. (d)--(f) Estimated pairwise nonparametric 
overlap values corresponding to the partitions in (a),
(b) and (c).}\label{fig:merging.step.1}
\end{figure*}
displays the results of the different phases and iterations of
KNOB-SynC. We display the stages of KNOB-SynC for $\kappa=1$ which is
when we have the lowest terminating $\genhatomega$
(from among $\kappa = 1,2,3,\infty$) for this example.   The $k$-means
phase of our algorithm identifies 14 clusters with partitioning as in
Figure~\ref{fig:kmeans.step.1} and estimates the initial overlap
matrix $\hat\Omega$ to be as in Figure~\ref{fig:overlap.step.2}. The
first merging 
phase yields the partitioning in Figure~\ref{fig:kmeans.step.3} with
the updated 
$\hat\Omega$ of Figure~\ref{fig:kmeans.step.4}. The next merging phase only
combines one pair of groups and is terminal, resulting in the final
partitioning of the dataset as in Figure~\ref{fig:kmeans.step.5}. The
overlap matrix (Figure~\ref{fig:kmeans.step.6}) indicates
well-separated clusters, with only six mislabeled observations
relative to the true, and a $\mR$ of 0.98 between the true and
estimated classifications. 

\begin{figure*}[h]
  \vspace{-0.05in}
  \hspace{-0.2in}\includegraphics[width=\textwidth]{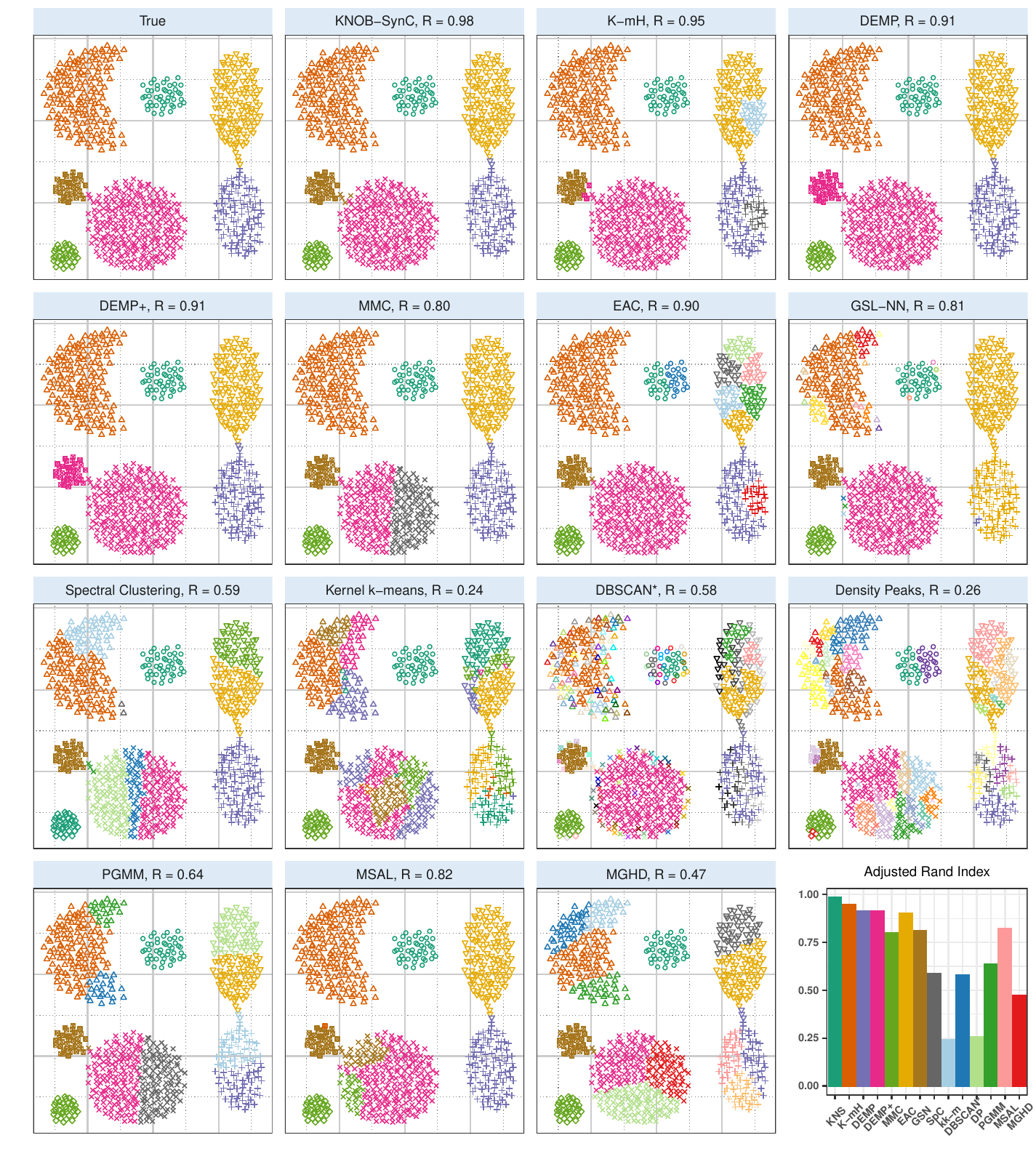}
  \vspace{-0.1in}
  \caption{Clustering performance of each algorithm on the {\tt Aggregation} dataset.}
  \label{fig:aggregation2}
  \end{figure*}
The competing methods (Figure~\ref{fig:aggregation2}) all perform
marginally to substantially worse. K-mH is the second best performer
($\mR=0.95$) finding $\hat C=9$ groups but breaking the top right 
cluster into two and also grouping a few other stray
observations. Both DEMP and DEMP+ yield the same result
($\mR=0.91,\hat C = 6$), but MMC ($\mR=0.8$, $\hat C = 8$) has trouble
with the largest group, splitting it into two sub-groups. EAC breaks the top
central and large groups on the right into many clusters, resulting in
$\hat C = 14$ but $\mR = 0.9$. Thus, in spite of identifying a large
number of groups, EAC is able to capture a fair bit of the complex group
structure  of this dataset. GSL-NN can not distinguish between the
groups on the right but also finds many other small groups elsewhere,
ending with $\hat  C =12 $ groups and $\mR = 0.81$. The performance of
spectral clustering is worse: it finds $\hat C = 12$ groups and has a $\mR =
0.59$ with the true classification. Despite being provided 
with the true $C=7$, kernel $k$-means with $\mR=0.24$ is the worst
performer in this example, with DP ($\mR = 0.26$, $\hat C = 26$)
only marginally better. DBSCAN* at $\mR = 0.58$ and $\hat C = 277$,
correctly finds the large circular group and one of the smallest
groups, but the observations in the top left group are almost all 
classified as outliers/scatter. Among the MBC methods for general-shaped
clusters, MSAL at $\mR=0.82$ is the best performer, finding  $\hat C
= 7$ groups but having trouble with the larger group at the bottom. 
PGMM finds $\hat C = 12$ groups ($\mR = 0.64$) with the larger
groups split further, while the worst-performing MBC method is MGHD ($\mR =
0.47$, $\hat C=15)$. 
\subsection{Additional 2D Experiments}
\subsubsection{Experimental framework}
Figure~\ref{fig:2ddata} displays the 12 additional 2D datasets used to
\begin{figure*}[h]
  \includegraphics[width = \textwidth]{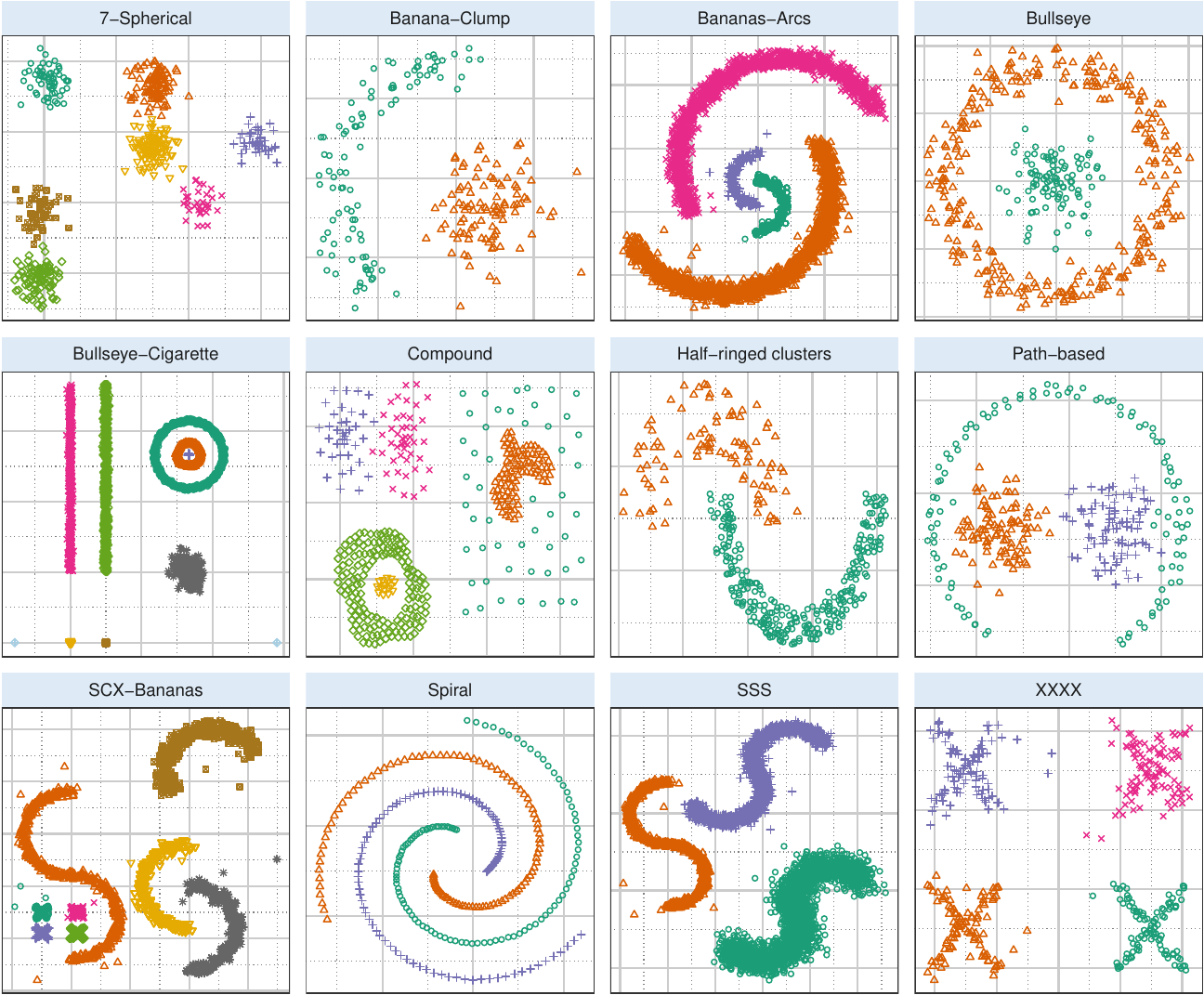}
  \caption{Shape datasets used in the two-dimensional performance evaluations.}
  \label{fig:2ddata}
\end{figure*}
evaluate performance of KNOB-SynC and its competitors. Barring the
first example, all these datasets have been used by other authors to
demonstrate and evaluate performance of their methods. The groups in
these datasets have structure ranging from the regular (for example, the 7-spherically-dispersed Gaussian clusters dataset that is 
modeled on a similar example in~\citep{maitra09} and where
sophisticated methods like KNOB-SynC are superfluous and unnecessary) to 
widely-varying complexity. The {\tt Banana Arcs} dataset has $n=4515$ observations
clumped in four banana-shaped structures 
arced around each other. The {\tt Banana-clump} and {\tt Bullseye}
datasets are from \citet{stuetzleandnugent10} -- the former has 200
observations with one 
spherical group and another arced around it on the left like a
banana, while the latter has 400 observations grouped, as its
name implies, as a bullseye.  The more complex-structured {\tt
  Bullseye-Cigarette} dataset \citep{petersonetal17} has three 
concentric-ringed groups, two elongated groups above two spherical
groups on the left, and another group that is actually a superset of
two overlapping spherical groups ($n=3025$ and $C=8$). The
{\tt Compound} dataset~\citep{zahn71} is very complex-structured
with $n=399$ observations in $C=6$ groups that are not just
varied in shape, but a group that sits atop another on the right. The
{\tt Half-ringed clusters} dataset~\citep{jainandlaw05} has 373 observations in two
arc-shaped clusters, one of which is dense and the other being very
sparsely-populated. The {\tt Path-based} dataset~\citep{changandyeung08} has
300 observations in three groups, two of which are regular-shaped and
surrounded by a widely arcing third group.  The {\tt   Spiral}
dataset~\citep{changandyeung08} has 312 observations in three spiral groups 
that are very difficult for standard clustering algorithms to recover accurately. The
{\tt SSS} dataset has 5015 observations in 
three S-shaped groups of varying density and orientations while the
{\tt XXXX} dataset has $n=415$ observations
distributed in four cross-shaped structures.
\subsubsection{Results}
    Figure~\ref{fig:AR.2D} and Table~\ref{tab:AR.2d} summarize the
performance of all methods on the 2D experimental
datasets.  Detailed displays of different methods 
on individual datasets are in Appendix~\ref{appendix-experiments}. The
summaries indicate across-the-board good performance of KNOB-SynC
with it always being a top performer. In its worst case,
KNOB-SynC gets a $\mR=0.55$ (on the  
{\tt Path-based} dataset) where it terminates
early~(Figure~\ref{fig:pathbased}) but here also  it is the
                                fourth-best performer, 
\begin{figure*}[h]
  \vspace{-0.1in}
  \centering
  \mbox{
    \hspace{-0.2in}
    \subfloat[Performance by dataset]{\includegraphics[width=0.5\textwidth]{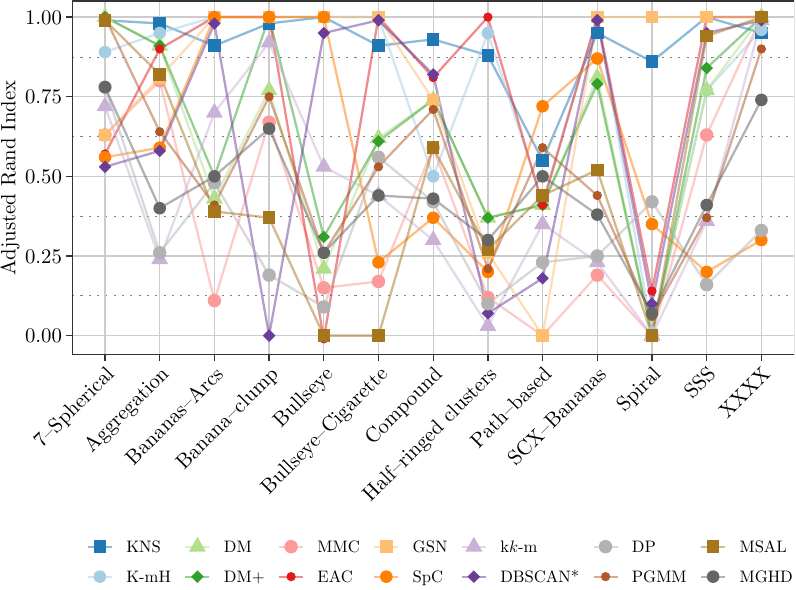}\label{fig:2d-lines}}
    \subfloat[Performance by method]{\includegraphics[width=0.5\textwidth]{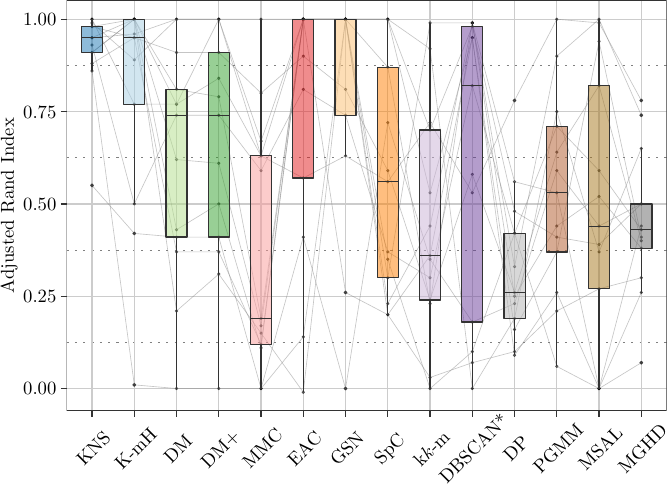}\label{fig:2d-boxplot}}
  }
    \caption{Performance of KNOB-SynC (abbreviation: KNS), K-mH,
      DEMP (DM), DEMP+ (DM+), MMC, EAC, GSL-NN (GSN), spectral
      clustering (SpC), kernel $k$-means (k$k$-m), DBSCAN$^*$, DP,
      PGMM, MSAL and MGHD on 2D datasets.} 
    \label{fig:AR.2D}
  \end{figure*}
behind spectral clustering
($\mR=0.72$), MGHD ($\mR =0.60$) and PGMM ($\mR =0.59$). 
The competing syncytial clustering methods do well in some
cases, but not in others where other  methods
perform better. Among the syncytial clustering methods, K-mH performs
better than  DEMP, DEMP+  and MMC whose  performance
can sometimes be poor ({\em e.g.}, on the 
{\tt Bullseye}, {\tt Half-ringed clusters} and {\tt Spiral} datasets
-- vide Figures
\ref{fig:bullseye}, \ref{fig:jain} and \ref{fig:spiral}). It is on these datasets that the other methods 
(EAC, GSL-NN, and spectral clustering) do better. The performance of
kernel-$k$-means even with known true number of groups is varied,
being very good sometimes ({\em e.g.}, in the {\tt Bananas-clump} dataset of
Figure~\ref{fig:banana-clump}) but very poor in other cases ({\em
  e.g.}, as seen before in the {\tt Aggregation} dataset) where almost
every  other method does well. The three general MBC methods perform
similarly but PGMM does a bit better than MSAL and MGHD. DBSCAN$^*$ and,
especially DP, generally perform poorly -- however, DBSCAN$^*$ performs well in
some cases. We quantify performance of each method against 
its competitors in terms of its average deviation from the best
performer. Specifically, for each dataset, we compute the deviation
(or difference) 
in $\mR$ of a method from that of the best performer for that
dataset. The average deviation ($\bar{\mathcal D}$) of a method over all datasets is an overall indicator
of its performance. Table~\ref{tab:AR.2d} provides
$\bar{\mathcal D}$ and the standard deviation or SD ($\mathcal D_\sigma$) of
the differences.  On the average, KNOB-SynC is the
best performer (and with the lowest $\mathcal D_\sigma$) followed
by GSL-NN, K-mH and EAC. This conclusion of KNOB-SynC's superior
overall performance is also supported by
Figure~\ref{fig:AR.2D}\subref{fig:2d-boxplot}. We surmise that KNOB-SynC does 
well across the different datasets because of its ability
by construction to merge many or few components at a time, with the 
exact choice of merges and termination objectively selected and
determined by the distinctiveness of the resulting partitioning as per
$\genhatomega$.   
\subsection{Higher-Dimensional Datasets}
\label{sec:expt.HD}
We also study the performance of KNOB-SynC and its competitors on
higher-dimensional datasets. These datasets are modest- to
higher-dimensional, with between 173 to 10993 records. 
For the higher-dimensional datasets ({\em i.e.}, with non-redundant
dimension greater than 10), we find that all methods other than 
GSL-NN generally perform better when used on the first few ($m$) kernel
principal components (KPCs) rather than on the raw data. For these
methods and these datasets, we use the 
first $m$ KPCs of each dataset with $m$ chosen as the first time after which 
increases in the eigenvalues corresponding to the successive KPCs are
below 0.5\%. GSL-NN is implemented on the original datasets. Further,
KNOB-SynC, K-mH and EAC are built on $k$-means whose results
depend on the scale of the features. So, for these methods, we
scaled each feature by the SD prior to
analysis unless the features were all collected on a similar scale, as
with the {\em E. coli} example of Section~\ref{sec:ecoli} or the
log-transformed GRB dataset
of Section~\ref{sec:app.GRB}. (Following
the usual rule-of-thumb in 
multivariate statistics, we assume that features are on similar scales
if the most variable feature has SD no more than four
times that of the feature with lowest variability.)
Each dataset and the performance of each method is first described
individually. A comprehensive  summary of the performance of each
method on each dataset follows in Section~\ref{sec:HD}.  

\subsubsection{Simplex-7 Gaussian Clusters}
This dataset, from \citet{stuetzleandnugent10}, is of Gaussian
realizations of size 50, 60, 70, 80, 90, 100 and 110 each
from seven clusters with means set at the vertices of the
seven-dimensional unit simplex and homogeneous spherical dispersions
with common SD of 0.25 in each dimension. Like the
{\tt 7-Spherical} dataset, this dataset exemplifies a case where standard
methods such as $k$-means or Gaussian MBC
should be  adequate. Therefore it is a test of whether our algorithm
and its competitors are able to refrain from
identifying spurious complexity. All methods, except for EAC, DBSCAN*
and DP, identify seven groups and have good clustering
performance. In particular, the syncytial methods have very good
performance ($\mR=0.97$); other methods have $\mR\in[0.92,0.98]$. EAC still performs well ($\mR =0.94$) but finds $\hat C
= 6$ groups. DBSCAN*  is the worst performer ($\mR = 0.03$) on this
dataset, finding many outliers ($\hat C = 508$). DP's
performance is middling at $\mR = 0.58$ and with many outliers ($\hat C = 79$).
\subsubsection{E. coli Protein Localization}
\label{sec:ecoli}
The {\em E. coli} dataset, publicly  available from the University of
California Irvine's 
Machine Learning Repository (UCIMLR)~\citep{newmanetal98}, concerns
identification of 
protein localization sites for the {\em E. coli}
bacteria~\citep{nakaiandkinehasa91}. There are eight protein 
localization sites: cytoplasm, inner membrane without signal sequence,  periplasm, inner membrane with an uncleavable signal sequence, 
outer membrane, outer membrane lipoprotein, 
inner membrane lipoprotein, and  inner membrane with a
cleavable signal sequence. Identifying these sites is an important
early step for finding remedies~\citep{nakaiandkinehasa91}. 
Each protein sequence has a number of numerical attributes -- see
\citet{hortonandnakai96} for a listing and
their detailed description. Two attributes are
binary, but 326 of the 336 sequences have common values for these
attributes. We restrict our investigation to these sequences and drop
the two binary attributes from our list of variables. These 326 sequences have
no representation from the inner membrane or outer membrane
lipoproteins. Additionally, we also drop two sequences because they
are the lone representatives from the inner membrane with cleavable
sequence site~\citep{maitra02}. 
\begin{table}[ht]
\centering
\caption{Confusion matrix of the KNOB-SynC groups against the true
  {\em E. coli} localization sites.}
\label{tab:knob.sync.Ecoli-unscaled-5-kappa-5}
\begin{tabular}{lrrrrrrrrr}
  \hline
 & 1 & 2 & 3 & 4 & 5 & 6 & 7 & 8 & 9 \\
  \hline
cytoplasm &   138 & 0 &  0 & 4 &   0 &   0 &   0 &   0 &    1 \\
  inner membrane, no signal sequence &  7 & 62  & 0 &  0 &   0 &   3 &   3 &   1 &   0 \\
  inner membrane, uncleavable signal sequence & 1 & 32 &  0 & 0 &   0 &   0 &   0 &   1 &   0  \\
  outer membrane &   0 &   0 & 17 &  2 &   0 &   0 &   0 &   0 &     0 \\
  periplasm &   3 &   1 &  2 & 38 &   8 &   0 &   0 &   0 &     0 \\
   \hline
\end{tabular}
\end{table}
Therefore we have $n=324$ observations from $C=5$ true classes. KNOB-SynC identifies $\hat C = 9$
groups. Table~\ref{tab:knob.sync.Ecoli-unscaled-5-kappa-5} presents
the confusion matrix containing the number of times a protein from a
localization site is assigned to each KNOB-SynC group. Ignoring stray
assignments, the sites are fairly well-defined in the first four
groups, with $\mR=0.72$. Uncleavable signal sequences 
from the inner membrane site are difficult to distinguish from those
that are also from there but have no signal
sequence. Sequences from the other sites are better-clarified. 
Among the alternative methods, EAC does slightly better ($\mR=0.77)$
but identifies 10 groups. The remaining methods all do slightly  to
substantially worse. DEMP, DEMP+ and K-mH each identify four groups
but with $\mR\in [0.63,0.70]$. K-mH finds only two groups ($\mR=0.41$)
while the rest find more groups but disagree more strongly with the
true localizations. DBSCAN$^{*}$ finds a large number of groups ($\hat C= 174$) with relatively poor performance ($\mR = 0.31$). 
DP is marginally better ($\mR=0.39,\hat C=12$) while PGMM ($\mR=0.48$)
and MGHD ($\mR=0.6$) improves on DP, each finding 8 groups. (MSAL did not
converge to a solution.) Overall, EAC and KNOB-SynC are the top two
performers, with DEMP and DEMP+ close behind.   
\subsubsection{Standard Wine Recognition} \label{sec:app.wine}
The standard wine recognition
dataset~\citep{forinaetal88,aeberhardetal92}, also available from the
UCIMLR contains $p=13$ measurements on $n=178$ 
wine samples that are obtained from its chemical analysis. There are
59, 71 and 48 wines of the Barolo, Grignolino  and Barbera cultivars, so $C=3$. Because $p>10$ here, we use $m=17$ KPCs. KNOB-SynC is the best  
performer, finding $\hat C = 3$ groups with a
  clustering   performance of $\mR = 0.92$. The first group contains
  all the 59 wines from the Barola cultivar and 2 Grignolino
  wines. The second group contains 66 wines, all exclusively from the
  Grignolino cultivar. The third group has 2 Grignolino  and 48
   Barbera wines. Thus, there is very good definition
  among the KNOB-SynC groups.  
On the other hand, only MMC ($\mR = 0.67; \hat C = 5$), K-mH ($\mR =
0.62,\hat C = 6$), PGMM ($\mR = 0.66,\hat C = 5)$ and
EAC ($\mR = 0.60;\hat C = 9$) perform modestly while the
others are substantially worse with DBSCAN$^*$, in particular,
classifying all observations as outliers, resulting in $\hat C = 178$ and
$\mR=0$.
\subsubsection{Extended Wine Recognition} \label{sec:app.wine2}
A reviewer very helpfully pointed out that the dataset used in
Section~\ref{sec:app.wine} is actually a reduced variant, and a fuller
version of the dataset with 27 variables is
available in the R package {\sc pgmm}. We used $m=26$ KPCs in our
experimental evaluations on  this larger dataset. DEMP and
DEMP+ ($\mR = 1, \hat C = 3$) show perfect classification while MGHD
($\mR = 0.95, \hat C = 3$) and MMC ($\mR = 0.91, \hat C = 4$) also
perform well.  
KNOB-SynC is a top performer and the best among the distribution-free
methods, finding $\hat C = 3$ groups and with a clustering 
performance of $\mR = 0.93$. The first group here has 58  Barola
and 2   Grignolino wines. The second group contains 68 wines  from the
  Grignolino cultivar and the one Barolo wine that was not placed in
  the first group. The third group has the 48 Barbera wines and the
  one remaining   Grignolino wine. Similar to
  the 13-dimensional case, we get good definition among the
  groups. Other methods not discussed here do moderately to
  substantially worse. 
\subsubsection{Olive Oils}
\label{sec:olives}
The olive oils dataset~\citep{forinaandtiscornia82,forinaetal83} has
measurements on 8 chemical components for 572 samples of olive oil
taken from 9 different areas in Italy that are from three regions:
Sardinia and Northern and Southern Italy. This is an interesting
dataset with sub-classes (areas) inside classes (regions). Indeed,
\citet{petersonetal17} were able to identify, with one
misclassification, sub-groups within the regions but not the areas
($\mR=0.67, \hat C=11$; we however get 
$\mR=0.56, \hat C=8$  using the authors' supplied code) -- they surmised
that it may be more possible to identify characteristics of olive oils
based on regions defined by physical geography rather than
areas demarcated by political geography. We therefore analyze
performance on this dataset both in terms of how  regions and
areas are recovered. KNOB-SynC identifies $\hat C=4$ 
regions~(Table~\ref{tab:knob.sync.Olive-oil.kappa.1}) with oils from
the Sardinian and Northern regions correctly classified into the first
two groups. The Southern region oils are split into our two remaining
groups, one containing all but 2 of the 25 North Apulian samples and 6
of the 36 Sicilian samples, and the other group containing all the
Southern oils.  
\begin{table}[ht]
  \centering
  \caption{Confusion matrix of the KNOB-SynC grouping of the Olive
    Oils dataset.} 
\label{tab:knob.sync.Olive-oil.kappa.1}
\begin{tabular}{llrrrr}
  \hline
Region & Area & 1 & 2 & 3 & 4 \\ 
\hline
Sardinia & Coast-Sardinia &  33 &   0 &   0 &   0 \\
& Inland-Sardinia &  65 &   0 &   0 &   0 \\ 
North&  East-Liguria &   0 &   0 &  50 &   0 \\ 
& West-Liguria &   0 &   0 &  50 &   0 \\
&  Umbria &   0 &   0 &  51 &   0 \\ 
South&  Calabria &   0 &  56 &   0 &   0 \\
& North-Apulia &   0 &   2 &   0 &  23 \\
&  South-Apulia &   0 & 206 &   0 &   0 \\
&  Sicily &   0 &  30 &   0 &   6 \\ 
  \hline
\end{tabular}
\end{table}
In terms of clustering
performance, KNOB-SynC gets $\mR=0.55$ when compared to the true areal
grouping but $\mR=0.87$ when compared to the true regional grouping.
For this dataset DEMP ($\mR = 0.85;\hat C = 7$), DEMP+ ($\mR =
0.82;\hat C = 12$) and MSAL ($\mR = 0.7; \hat C = 9$) are the top
performers with respect to the true areal grouping. The remaining
methods all have middling performance. 
When compared with 
the true regional grouping, KNOB-SynC is by far the best
performer. Overall, the clustering performance of KNOB-SynC for the
regional grouping 
marginally trumps the performance of DEMP for the areal grouping and
so may be considered to be more accurate in uncovering the group
structure in the dataset.

\subsubsection{Image Segmentation}
 \label{sec:app.image}
 The image segmentation dataset, also available from the UCIMLR, is on 19 
 attributes of the scene in each $3\times 3$ 
 image manually classified to be from BRICKFACE,
 CEMENT, FOLIAGE, GRASS, PATH, SKY and WINDOW. (Thus, $C=7$.) We
 combine the training 
 and test datasets to obtain 330 instances of each scene, so
 $n=2310$. 
 There is a lot of redundancy in the attributes so we reduce the
 dataset to 8 PCs that together explain at least 99.9\%
 of the total variance in the dataset. The PCs are obtained from  the
correlation matrix because the 19 attributes have vastly different
scales. The KNOB-SynC solution finds $\hat C = 12$ clusters, with
 $\mR=0.55$. The confusion matrix
 (Table~\ref{tab:knob.sync.Image-scaledpc-8-kappa-3}) indicates that
the SKY images are perfectly identified while GRASS and, to a lesser
 extent,  PATH and CEMENT, are fairly  well-identified. On the other
 hand, the partitioning  struggles to distinguish between BRICKFACE,
 FOLIAGE  and WINDOW.  
\begin{table}[h!]
  \centering
  \caption{Confusion matrix of the KNOB-SynC grouping against the true for
    the Image segmentation dataset.}
\label{tab:knob.sync.Image-scaledpc-8-kappa-3}
  \begin{tabular}{lrrrrrrrrrrrr}
    \hline
    & 1 & 2 & 3 & 4 & 5 & 6 & 7 & 8 & 9 & 10 & 11 & 12 \\ 
    \hline
    BRICKFACE & 330 & 0 & 0 &  0 & 0 & 0 & 0 &   0 &   0 &   0 &   0 &   0 \\ 
    CEMENT &  42 &  257 & 0 & 4  &   0 &   27 &  0 &   0 &   0 &   0 &   0 &   0 \\ 
    FOLIAGE & 300 &  5 & 0 & 0 & 0 & 5 &  2  &   7 &   3 &   1 &   3 &   4 \\ 
    GRASS &   1 &  0 &   327 &   0 & 0 & 2 & 0 &   0 &   0 &   0 &   0 &   0 \\ 
    PATH &   0 & 0 &   0 &   269 &   0 &  61 & 0& 0 &   0 &   0 &   0 &   0 \\ 
    SKY &   0 &   0 & 0 &   0 &   330 &   0 &   0 &   0 &   0 &   0 &   0 &   0 \\ 
    WINDOW & 309 &  13 &   0 &  0 &   0 &   8 & 0& 0 &   0 &   0 &   0 &   0 \\ 
   \hline
  \end{tabular}
\end{table}
Among other methods, only EAC ($\mR = 0.59, \hat C =40$) and 
MGHD ($\mR=0.56, \hat C = 7$) are modestly to marginally better than KNOB-SynC.
Inspection of the EAC grouping indicates many small groups but
also  difficulty in separating FOLIAGE and WINDOW, placing them
together in one group. Further, BRICKFACE is split 
into five groups, four of which are predominantly of this kind, but
the fifth group is unable to distinguish 146 observations of BRICKFACE
from 32, 52 and 62 observations of CEMENT, FOLIAGE and WINDOW,
respectively. The other methods all perform moderately to
substantially worse (Table~\ref{tab:AR.HD}) with PGMM, MSAL and
DBSCAN$^*$ unable to find  clustering solutions. 
\subsubsection{Yeast Protein Localization}
The yeast protein localization dataset~\citep{nakai96}, also obtained
from the UCIMLR, was used by \citet{melnykov16} to illustrate the
application of DEMP+. This dataset is on the
localization of the proteins in yeast into one of 
$C=10$ sites and  has two attributes (presence of ``HDEL''
substring and peroxisomal targeting signal in the C-term) that are 
essentially binary and trinary. 
Following \citet{melnykov16}, we drop these variables and use the other $p=6$ variables, namely signal sequence recognition scores based on
(a) McGeoch's and (b) von Heijne's methods, (c) ALOM membrane spanning region
prediction score, and discriminant analysis scores of the amino acid
content of (d) N-terminal region (20 residues long) of mitochondrial and
non-mitochondrial proteins and (e)  vacuolar and extracellular
proteins and  (f) discriminant scores of nuclear localization signals
of nuclear and non-nuclear proteins. For this dataset, all methods
perform poorly. KNOB-SynC ($\mR = 0.226;\hat C=7$) is the best
performer -- the other clustering methods essentially randomly
allocate observations. Surprisingly, DEMP+ ($\hat C=6$, $\mR =-0.01$)
performs very poorly. (\citep{melnykov16} only used the first five of our
variables to illustrate the DEMP+ method:  
  we find no appreciable improvement even then, with $\hat C=7$ and $\mR 
  =-0.009$. Personal queries to the author did not successfully resolve this
  discrepancy.) It appears therefore  that the yeast
protein localization  dataset may be difficult to accurately partition
 in a completely  unsupervised framework.
  \subsubsection{Acute Lymphoblastic Leukemia}\label{sec:sim.ALL}
The Acute Lymphoblastic Leukemia (ALL) training dataset of \citet{yeohetal02}
was used by \citet{stuetzleandnugent10} to illustrate GSL-NN in a
high-dimensional small sample size framework. We use the
standardized dataset in  \citet{stuetzleandnugent10} that measured the oligonucleotide 
expression levels of the 1000 highest-varying genes in 215 patients
suffering from one of seven leukemia subtypes, namely, T-ALL,
E2A-PBX1, BCR-ABL, TEL-AML1, MLL rearrangement, Hyperploid $> 50$
chromosomes, or an unknown category labeled OTHER. Some subtypes have
very few cases: for instance, only 9, 14 and 18 patients are of type
BCR-ABL, MLL and E2A-PBX1, respectively. For this
dataset, we use $m=42$ KPCs for all methods but GSL-NN. The $k$-means
stage of KNOB-SynC identifies six groups, none of which are merged in
the merging phase, resulting in the best partitioning among all competing
methods. Table~\ref{tab:knob.sync.ALL-Nugent-kPCA-42}  
\begin{table}[ht]
  \centering
    \caption{Confusion matrix of the KNOB-SynC grouping against the
      true leukemia subtypes for the ALL dataset.}\label{tab:knob.sync.ALL-Nugent-kPCA-42}
\begin{tabular}{lrrrrrr}
  \hline
 & 1 & 2 & 3 & 4 & 5 & 6 \\ 
  \hline
  Hyperdiploid $\!>\!$ 50 &   0 &   0 &   2 &  35 &   5 &   0 \\
  E2A-PBX1 &   0 &  17 &   0 &   0 &   1 &   0 \\
  BCR-ABL &   0 &   0 &   2 &   1 &   6 &   0 \\
  TEL-AML1 &   4 &   2 &   8 &   2 &  36 &   0 \\ 
  MLL &   0 &   3 &  10 &   0 &   1 &   0 \\ 
  T-ALL &   0 &   0 &   1 &   0 &   0 &  27 \\ 
  OTHER &  51 &   0 &   0 &   0 &   1 &   0 \\ 
   \hline
\end{tabular}
\end{table}
presents the confusion matrix containing the number of cases a patient of a 
leukemia subtype was assigned to a KNOB-SynC group. We see that 
most leukemia subtypes are distinctively identified in the KNOB-SynC
solution. The alternative methods 
perform mildly to substantially worse with PGMM, spectral clustering,
GSL-NN and DP having clustering solutions ($\mR=0.61,0.55,0.54,0.53$)
that are the next best after KNOB-SynC.  Other methods generally do
poorly, with MSAL and MGHD unable to find solutions while DBSCAN$^*$
classifies all  observations as outliers. 
\subsubsection{Zipcode images}
The zipcode images~\citep{stuetzleandnugent10} dataset consists of $n=2000$
$16\times 16$ images of handwritten Hindu-Arabic numerals and is our
second higher-dimensional example. As in the ALL dataset of
Section~\ref{sec:sim.ALL}, we normalize the observations to have zero
mean and unit variance so that the Euclidean distance between any two normalized
images is negatively and linearly related to the correlation between
their pixels. We extract and use the first $m=33$ KPCs for all
algorithms but GSL-NN. KNOB-SynC identifies 9 groups and has the best
clustering performance ($\mR=0.76$). DP is the second best ($\mR=0.58$) performer but finds $\hat
C=53$ groups (including singletons) followed by MGHD ($\mR=0.56,\hat C=6$), 
K-mH ($\mR=0.55, \hat  C=22$), GSL-NN and spectral clustering (both
with $\mR=0.54$ but $\hat C = 7$ and 23). The other  methods all
perform moderately to substantially worse. 
\begin{figure}[h] 
  \centering
  \includegraphics[width = \textwidth]{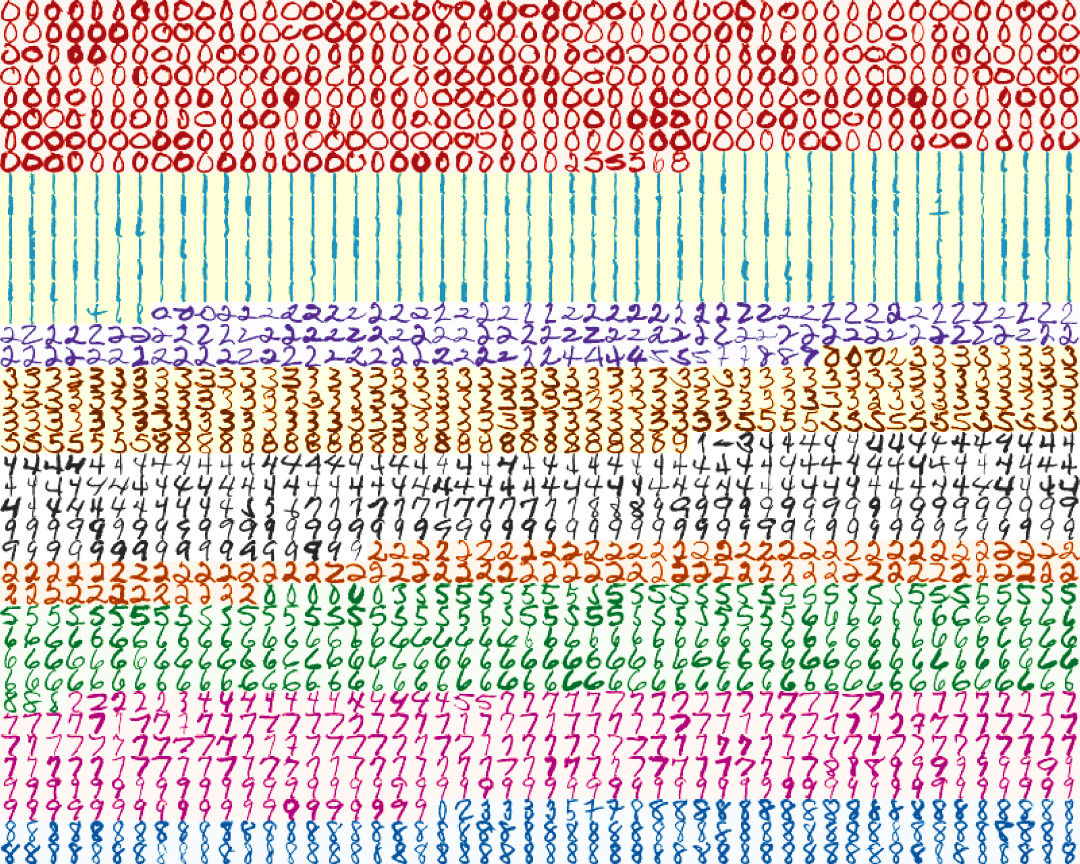}
  \caption{KNOB-SynC groups, with colormap indicating group, of the Zipcode dataset.}
  \label{fig:zip}
\end{figure}
Figure~\ref{fig:zip} displays the 9 KNOB-SynC groups. While misclassifications
abound in almost all groups, there is good agreement with 0, 1, 2, the
leaner 8s and, (to a lesser extent) 3 and 6, largely correctly
identified. The digit 2 is placed in two groups, of the leaner and the
rounded versions. The group where 3 predominates also has some 5s and 8s 
but the categorization makes visual sense. Another group is composed
largely of 4s, 7s and 9s but that placement also appears visually
explainable. Clearer and straighter  7s and 9s are placed
in a separate group. Our partitioning finds it harder to distinguish
between 5 and 6 but here also the commonality of the strokes in the
digits assigned to this group explains this categorization. Thus we
see that KNOB-SynC  is not only the best performer for this dataset
but  also provides interpretable results. We comment that our
application of all methods to this dataset has been entirely
unsupervised: methodologies that also account for spatial context and
pixel neighborhood may further improve the grouping
but are outside the purview of this paper.
\subsubsection{Handwritten Pen-digits}\label{sec:experiment.pendigits}
The Handwritten Pen-digits
dataset~\citep{alimoglu96,alimogluandalpaydin96} available at  the 
\begin{table}[ht]
\vspace{-0.05in}
  \centering
  \caption{Confusion matrix for the Handwritten Pen-digits dataset.} 
  \label{tab:knob.sync.pendigits}
    {\footnotesize
  \begin{tabular}{rrrrrrrrrrrrrrrr}
    \hline
    & 0 & 1 & 2 & 3 & 4 & 5 & 6 & 7 & 8 & 9 & 11 & 12 & 13 & 14 & 15 \\ 
    \hline
    0 & 1099 &   1 &   0 &   0 &  19 &   0 &  21 &   0 &   0 &   2 &   0 &   1 &   0 &   0 &   0 \\ 
    1 &   0 & 657 & 358 &  34 &   1 &   0 &   2 &   2 &   0 &  89 &   0 &   0 &   0 &   0 &   0 \\ 
    2 &   0 &   2 & 1141 &   0 &   0 &   0 &   0 &   0 &   0 &   1 &   0 &   0 &   0 &   0 &   0 \\ 
    3 &   0 &   4 &   2 & 1046 &   1 &   0 &   0 &   0 &   0 &   2 &   0 &   0 &   0 &   0 &   0 \\ 
    4 &   0 &   5 &   1 &   2 & 1118 &   0 &   1 &   0 &   0 &  17 &   0 &   0 &   0 &   0 &   0 \\ 
    5 &   0 &   1 &   0 & 252 &   0 & 625 &   0 &   0 &   2 & 175 &   0 &   0 &   0 &   0 &   0 \\ 
    6 &   0 &   0 &   1 &   0 &   0 &   1 & 1054 &   0 &   0 &   0 &   0 &   0 &   0 &   0 &   0 \\ 
    7 &   0 & 144 &   5 &   2 &   0 &   0 &   0 & 914 &   0 &   0 &   0 &   0 &  77 &   0 &   0 \\ 
    8 &   4 &   0 &   0 &   3 &   0 &   1 &   0 &   1 & 461 &   0 & 139 & 321 &  48 &  24 &  53 \\ 
    9 &  24 &   9 &   0 &  72 &   3 &   0 &   0 &   0 &   1 & 714 &   0 &   0 &   0 & 232 &   0 \\ 
    \hline
  \end{tabular}
  }
\end{table}
UCIMLR is a larger dataset that has 16 attributes from 250 
handwritten samples of 30 writers. (There are $n=10992$ records
because eight samples are unavailable.)
We use $m=18$ KPCs in our analysis \citep[used the first 7 PCs and got
$\mR=0.64$ and $\hat C=24$]{petersonetal17}. 
KNOB-SynC finds $\hat C=15$ groups and is the best performer
($\mR=0.723$). It separates the digits 0, 2, 3, 4, 6 and, to a lesser extent,
 7 fairly well but identifying 1, 5 and 9  is a bit more challenging~(Table
\ref{tab:knob.sync.pendigits}). It also identifies multiple types of
8. MGHD finds the correct number and is the next-best performer
$({\cal R} = 0.67)$. The other methods perform moderately to
substantially worse with MSAL unable to find a clustering solution.
\subsubsection{Summary of Performance}\label{sec:HD}
\begin{figure*}[h]
  \vspace{-0.1in}
  \centering
  \mbox{
    \hspace{-0.1in}
    \subfloat[Performance by dataset]{\includegraphics[width=0.5\textwidth]{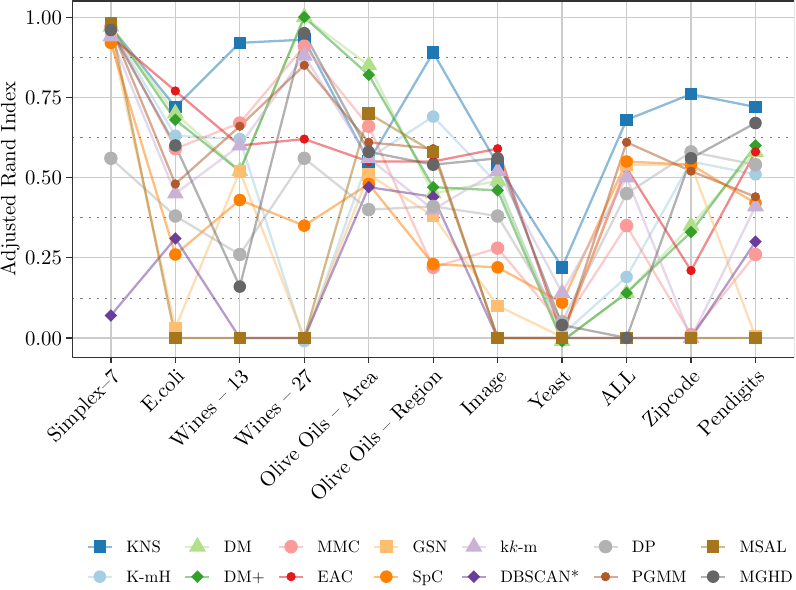}\label{fig:HD-lines}}
    \subfloat[Performance by method]{\includegraphics[width=0.5\textwidth]{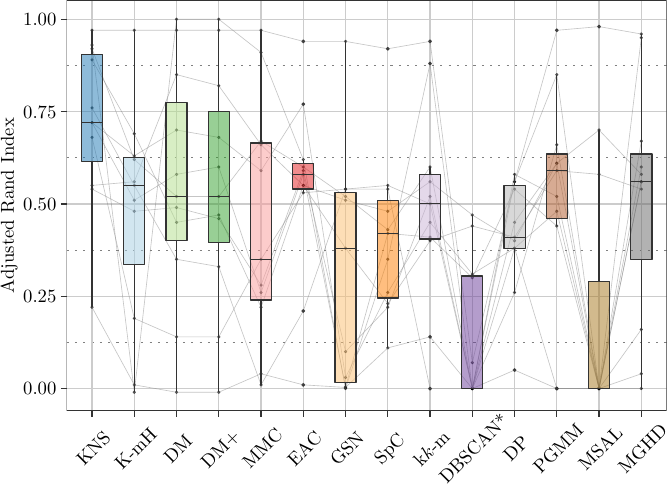}\label{fig:HD-boxplot}}
  }
  \caption{Overall performance of all competing methods on all 
    higher-dimensional datasets. Abbreviations are as in Figure~\ref{fig:AR.2D}.
  } 
    \label{fig:AR.HD}
  \end{figure*}
Figure~\ref{fig:AR.HD} and Table~\ref{tab:AR.HD} summarize performance
of all methods on the higher-dimensional experiments. As in the 2D case, KNOB-SynC is almost always among the top
performers for high-dimensional datasets. Indeed, KNOB-SynC has 
the lowest average difference in $\mR$ from that of the best-performing method  over all datasets (Table~\ref{tab:AR.HD}). The
other methods generally perform worse, with EAC, PGMM and
kernel-$k$-means (with true number of groups) among the better
ones. Thus, the results of our experiments on real and synthetic
datasets indicate good performance of KNOB-SynC relative to its competitors.
\subsection{Extensions of KNOB-SynC}\label{sec:extensions}
As indicated in Section~\ref{sec:methodology}, the development of our
syncytial clustering methodology is based on the nonparametric
estimation of the CDF of the residuals and so can be
applied to other scenarios. We explore performance of our methodology
in two such settings. 
\subsubsection{KNOB-SynC in the presence of scatter}
\label{sec:att.image}
\citet{maitraandramler09} provided the $k$-clips algorithm for
$k$-means clustering in the presence of scatter, or observations that
are unlike any other in the dataset. Our KNOB-SynC
methodology and software readily incorporates $k$-clips results
by replacing the $k$-means phase with that algorithm, and proceeding by
including the scatter points as individual singleton clusters. We
illustrate our methodology on the first 100 images of the Olivetti
faces database~\citep{samariaandharter94} that were used by
\citet{rodriguezandlaio14} to illustrate their DP algorithm. The 100
images under our consideration are of 10 faces each of 10 individuals
taken at different angles and under different light conditions. Therefore, each individual can be considered to 
be a group with members that are that person's 10 images. Each
$112\times92$ image has a total of 10,304 pixels so we use the first
37 KPCs.  While this
application does not have any true scatter points, we use this
application to illustrate KNOB-SynC with $k$-clips because it was used
by \citet{rodriguezandlaio14} to showcase DP  that finds scatter
(outliers, in their parlance) in addition to clusters.  

The $k$-clips algorithm with the default Bayesian Information
Criterion (BIC)~\citep{schwarz78} finds
only two well-defined homogeneous spherical clusters and 68 scatter
points. We use the trace of the within-sums-of-squares-and-products
matrix, rather than its determinant~\citep{maitraandramler09}, in our
objective function in order to satisfy the condition of homogenous
spherical clusters around which our base KNOB-SynC algorithm is built. Thus, we have a total of 
70 initial groups. KNOB-SynC's 
merging phase ends with 9 large groups, 5 small groups and 1 scatter
observation (so $\hat C=16$) and $\mR=0.902$.
\begin{table}[h]
  \caption{Confusion matrix of the KNOB-SynC results for the Olivetti faces dataset.}
  \label{tab:faces}
  \centering
\begin{tabular}{crrrrrrrrrrrrrrrr}
  \hline
  &\multicolumn{16}{c}{Assigned Groups} \\
 Individual & 1 & 2 & 3 & 4 & 5 & 6 & 7 & 8 & 9 & 10 & 11 & 12 & 13 & 14 & 15 & 16 \\ 
  \hline
1 &   9 &   0 &   0 &   0 &   0 &   0 &   0 &   0 &   0 &   0 &   0 &   0 &   1 &   0 &   0 &   0 \\ 
  2 &   0 &  10 &   0 &   0 &   0 &   0 &   0 &   0 &   0 &   0 &   0 &   0 &   0 &   0 &   0 &   0 \\ 
  3 &   0 &   0 &   8 &   0 &   0 &   0 &   0 &   0 &   0 &   0 &   2 &   0 &   0 &   0 &   0 &   0 \\ 
  4 &   0 &   0 &   0 &  10 &   0 &   0 &   0 &   0 &   0 &   0 &   0 &   0 &   0 &   0 &   0 &   0 \\ 
  5 &   0 &   0 &   0 &   0 &   3 &   0 &   0 &   0 &   0 &   0 &   0 &   2 &   0 &   2 &   1 &   2 \\ 
  6 &   0 &   0 &   0 &   0 &   0 &  10 &   0 &   0 &   0 &   0 &   0 &   0 &   0 &   0 &   0 &   0 \\ 
  7 &   0 &   0 &   0 &   0 &   0 &   0 &  10 &   0 &   0 &   0 &   0 &   0 &   0 &   0 &   0 &   0 \\ 
  8 &   0 &   0 &   0 &   0 &   0 &   0 &   0 &  10 &   0 &   0 &   0 &   0 &   0 &   0 &   0 &   0 \\ 
  9 &   0 &   0 &   0 &   0 &   0 &   0 &   0 &   0 &  10 &   0 &   0 &   0 &   0 &   0 &   0 &   0 \\ 
  10 &   0 &   0 &   0 &   0 &   0 &   0 &   0 &   0 &   0 &   9 &   0 &   0 &   0 &   0 &   1 &   0 \\ 
   \hline
\end{tabular}
\end{table}
The results are displayed in Table~\ref{tab:faces} and
Figure~\ref{fig:faces} -- for comparison, the latter also displays the
results reported in \citet{rodriguezandlaio14} which found 9 clusters
and 62 scatter points, resulting in $\hat C=71$ and $\mR=0.22$. (The
figure displays images assigned to a group by means of a distinctive
sequential palette. Because there are not enough colors to also
identify each scatter point with its individual sequential palette,
we use an individual randomized nominal palette for each scatter
assignment.)
\begin{figure*}[h]
  \vspace{-0.1in}
  \mbox{
    \hspace{-0.02\textwidth}
    \subfloat[KNOB-SynC: $\mR=0.90$]{\includegraphics[width=0.49\textwidth]{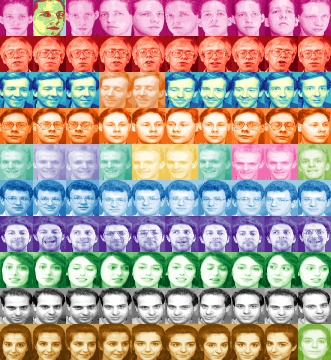}\label{fig:10p-kns}}
    \hspace{0.02\textwidth}
    \subfloat[DP: $\mR=0.22$]{\includegraphics[width=0.49\textwidth]{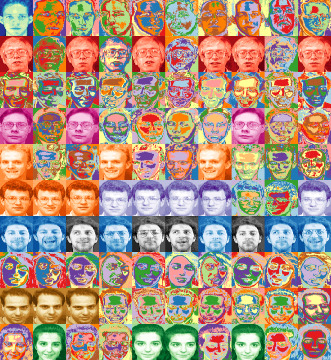}\label{fig:10p-sci}}
  }
  \caption{Clusters of the first 100 images in the Olivetti database
    obtained by (a) KNOB-SynC and (b) DP as reported in
    \citet{rodriguezandlaio14}. Each group is represented by its own
    distinctive sequential palette. Scatter observations ({\em i.e.}
    singleton groups) are represented by individual randomized nominal
    palettes.} 
    \label{fig:faces}
  \end{figure*}
  KNOB-SynC identifies images from six individuals (Persons 2, 4, 6, 7, 8 and
9) perfectly and the first, third and the tenth individuals nearly
so. The fifth individual is 
characterized into 5 smaller groups that includes the case where one
image is grouped together with the one misclassified image of the
tenth person. The performance of our algorithm overwhelms that
reported in \citet{rodriguezandlaio14}. We note that we used the first
37 KPCs with our KNOB-SynC algorithm while \citet{rodriguezandlaio14}
used the original images with similarity metric as in
\citet{sampatetal09}. Using DP ($\hat C = 83,\mR=0.06$) or DBSCAN$^*$
($\hat C = 100, \mR=0$) with Euclidean similarity on the 37 KPCs
gave us worse results.  
\subsubsection{KNOB-SynC with incomplete records}
\label{sec:app.sdss}
We now illustrate a scenario where KNOB-SynC is applied to a dataset with
incomplete records. In this example, we replace the $k$-means phase
with ~\citet{lithioandmaitra18}'s $k_m$-means algorithm that modifies
$k$-means to account for incomplete records. The authors also develop
a modified jump statistic to select the number of groups. The 
$k_m$-means results are input into the merging phase of KNOB-SynC and 
the algorithm proceeds as usual. 

We illustrate our methodology on a subset~\citep{wagstaff04} of the
Sloan Digital Sky Survey (SDSS) dataset that measures five features
(brightness, in psfCounts, size in petroRads, texture, and two measures
of shape ($M\_e1$ and $M\_e2$ that we refer to as Shape1 and Shape2 in
our analysis) on 1220 galaxies and 287 stars. Thus the true $C=2$ and
$n=1507$. The dataset has some missing values for the shape measures
of 42 galaxies.

The $k_m$-means algorithm with the modified jump statistic
of~\citet{lithioandmaitra18} finds $K_0 = 46$ homogeneous
spherically-dispersed groups. 
The initial overlap calculations of Step~\ref{ol.phase} of our
algorithm yield $\genhatomega = 0.0297$ and $\hat{\check\omega}
=0.381$. The merging phase is triggered, and terminates  with $\hat C = 4$
groups. 
\begin{figure}[h]
  \centering
  \mbox{
     \subfloat{
      \includegraphics[width=0.3\textwidth]{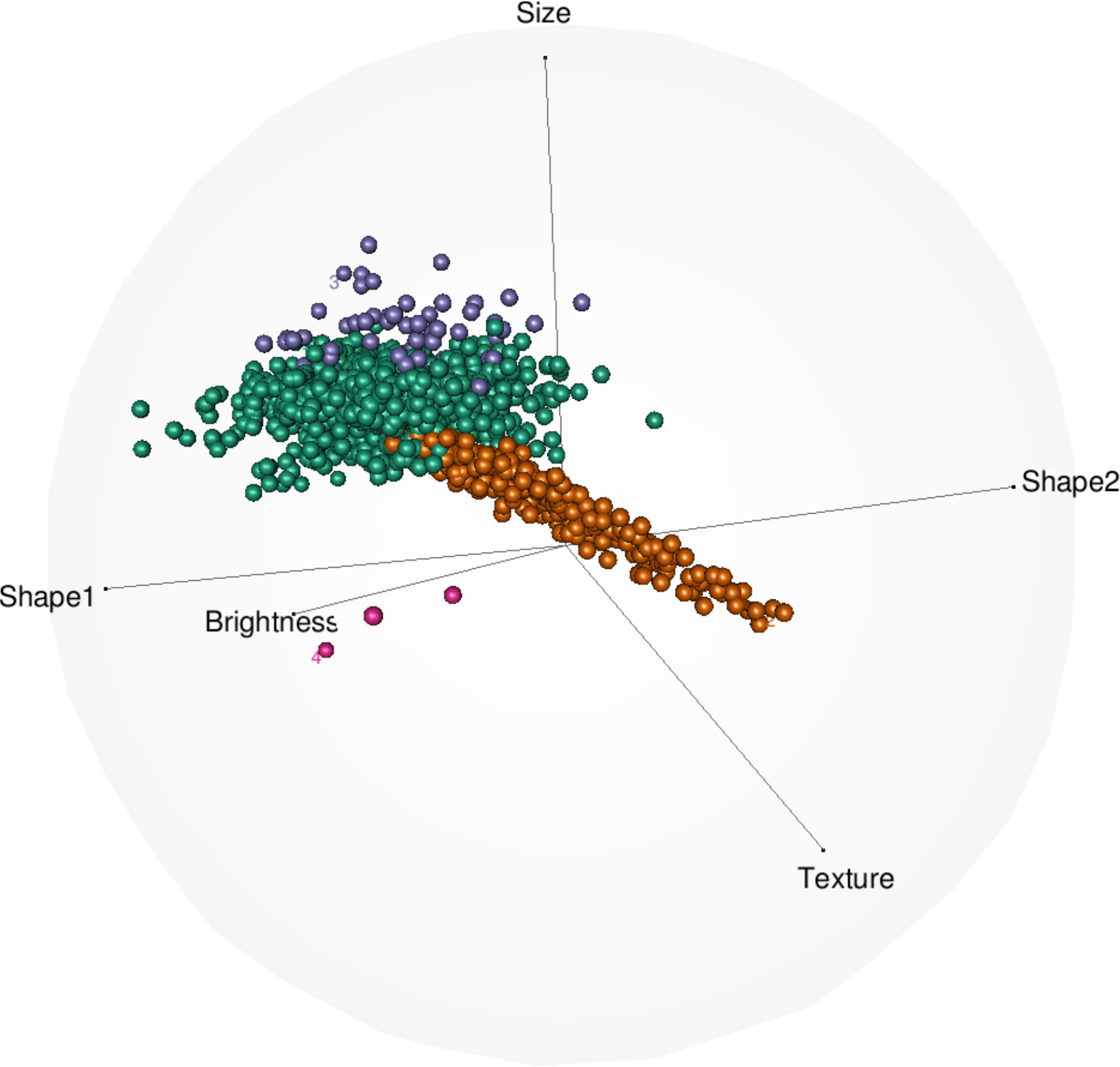}}
    \subfloat{
      \includegraphics[width=0.3\textwidth]{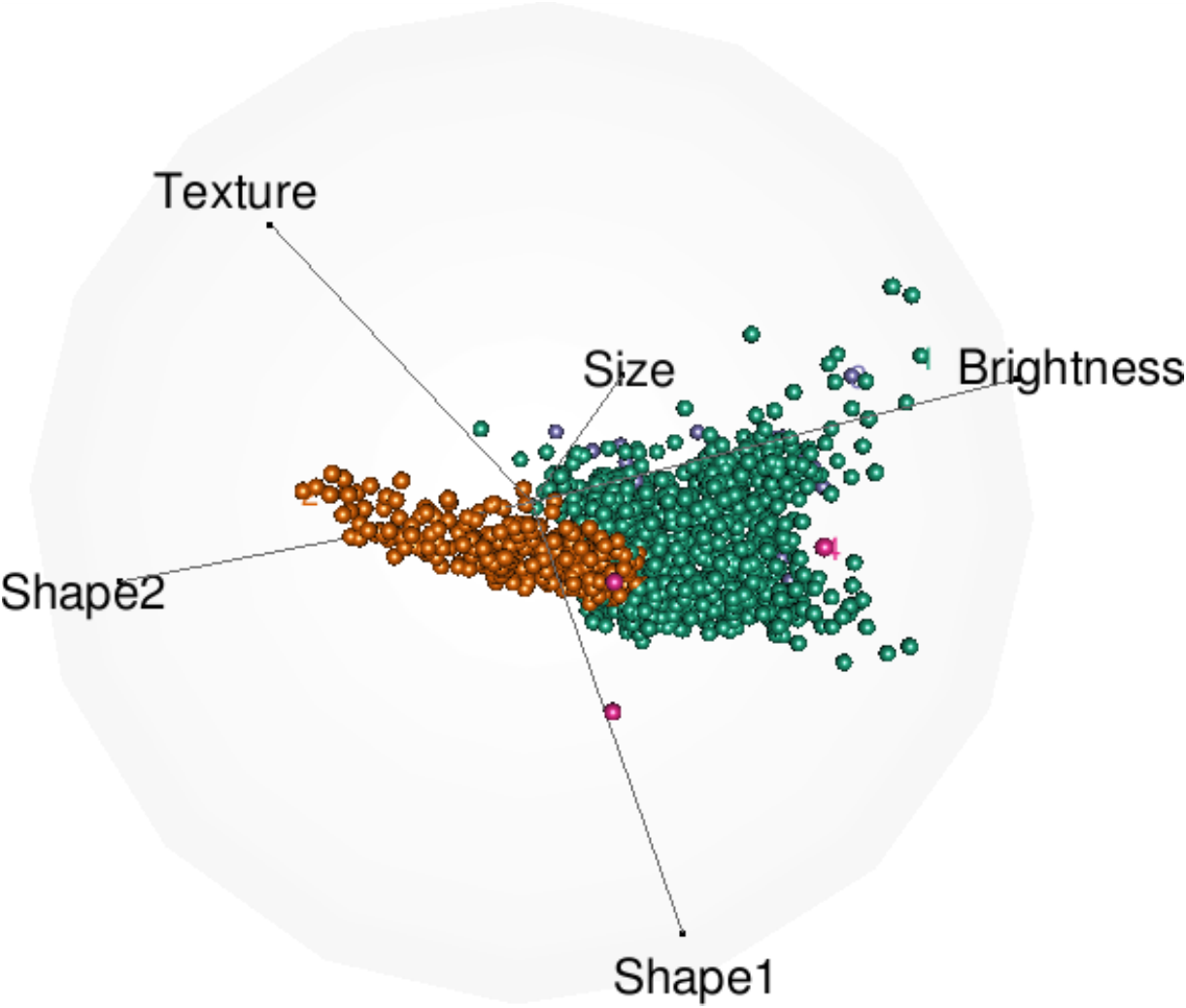}}
      \subfloat{
        \includegraphics[width=0.3\textwidth]{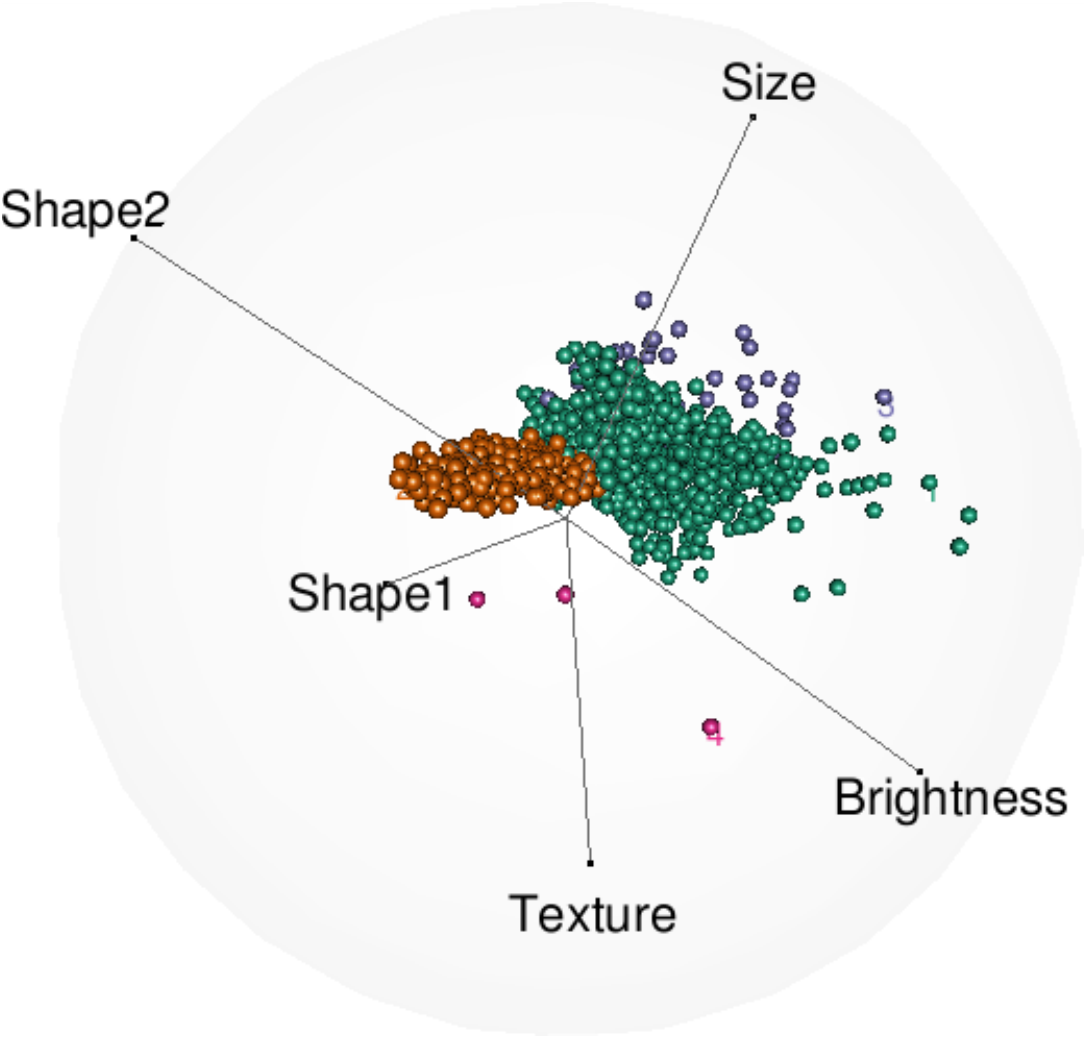}}
  }
  \subfloat
      {\begin{tabular}{r|rrrr}\hline
        & \multicolumn{4}{c}{KNOB-SynC Groups} \\
        & 1 & 2 & 3 & 4 \\ 
        \hline
        Galaxies & 1159 &   2 &  56 &   3 \\ 
        Stars &   0 & 287 &   0 &   0 \\ 
        \hline
      \end{tabular}\label{tab:sdss.confusion.matrix}
    }
      \caption{(Top) Three views of 3D radial visualization displays of
        the KNOB-SynC groups found in the SDSS dataset. Only the completely
        observed records are displayed in the figures. (Bottom)
        Confusion matrix between the true classifications of Galaxies
        and Stars with the KNOB-SynC grouping that yielded $\mR=0.86$.}
    \label{fig:sdss.radviz}
\end{figure}
Figure \ref{fig:sdss.radviz} provides a 3D radial
visualization~\citep{daietal19} of the clustering results and a
confusion matrix of the obtained grouping vis-a-vis the true
classification. 
We see that KNOB-SynC groups all the 287 stars together, but also includes
2 galaxies. The remaining galaxies are all partitioned into groups of
1159, 56 and 3 observations. The large galaxy group and the group with
stars are all well-separated from the ones in the smaller galaxy
groups. The second-largest KNOB-SynC galaxy group has 
larger-sized galaxies while the three galaxies in the last group have
larger Shape1 and brightness. This illustration demonstrates
KNOB-SynC's ability to identify general-shaped clusters even in the
presence of 
incomplete records. We note that some of the competing methods 
such as K-mH or EAC may be modified to incorporate $k_m$-means
results but such modifications to both the methodology and software is
outside the scope of this paper.

Our experimental evaluations comprehensively demonstrate that our
KNOB-SynC algorithm works very well in finding general-shaped
clusters. Indeed, our methodology can also incorporate scenarios that
allow for scatter or incomplete records in the dataset. 

%% file: application.tex
\section{Real-world applications} \label{sec:application}
In this section we apply KNOB-SynC to first  find the different
kinds of Gamma Ray Bursts (GRBs) in an astronomy catalog and second,
to identify activation detected in fMRI experiments. The ground truth
is unknown in both these applications, so we compare our results with
other available evidence in the literature.
\subsection{Determining the distinct kinds of Gamma Ray Bursts}\label{sec:app.GRB}
There is tremendous interest in understanding the source and nature of 
Gamma Ray Bursts (GRBs) that are the brightest electromagnetic events known
 to occur in  space~\citep{chattopadhyayetal07,piran05}. Many 
 researchers \citep{mazetsetal81,norrisetal84,dezalayetal92} 
have hypothesized that GRBs are of several kinds, but the exact number
and descriptive properties of these groups is an area of active research and
investigation. Most analyses have traditionally focused on univariate
and bivariate statistical and descriptive methods for classification
and found  two groups but other authors
\citep{mukherjeeetal98,chattopadhyayetal07} have found three different
kinds of GRBs when using more variables in the clustering. Recent careful 
analyses~\citep{chattopadhyayandmaitra17,chattopadhyayandmaitra18} has
conclusively established five ellipsoidally-shaped groups in the GRB
dataset obtained from the BATSE 4Br catalog. Indeed,
\citet{chattopadhyayandmaitra18} established that all nine 
fields of the BATSE 4Br catalog have important clustering information
using methods developed in \citet{rafteryanddean06}. These nine
fields are the two duration variables (time by which 50\% and 90\%
of the flux arrive), the four time-integrated fluences  in the 20-50,
50-100, 100-300, and $>300$ keV spectral channels, and the (three)
measurements on peak fluxes in time bins of 64, 256 and 1024
milliseconds. The authors used multivariate $t$-mixtures MBC on the logarithm of
the measurements, and BIC 
for model selection, to  arrive at their result of five 
ellipsoidally-shaped groups.

GRB datasets have typically been analyzed after using a $\log_{10}$
transformation to remove skewness in the dataset. This summary
transformation is somewhat arbitrary so~\citet{berryandmaitra19} used
their Transformation-infused $K$-means (TiK-means) algorithm 
to alternately transform features and cluster skewed datasets. A
modification~\citep{berryandmaitra19} of the jump statistic that 
accounts for the use 
of transformations in the algorithm found five
groups that were characterized as long-intermediate-intermediate,
short-faint-intermediate,  short-faint-soft, long-bright-hard and
long-intermediate-hard in terms of their duration ($T_{90}$), total
fluence ($F_{total}$) and spectral hardness ($H_{321}$) which are the
summaries used to characterize GRB groups~\citep{mukherjeeetal98}.
\begin{figure*}[h]
  \centering
  \mbox{\subfloat[KNOB-SynC: $\hat C = 5$]{\includegraphics[width=.5\textwidth]{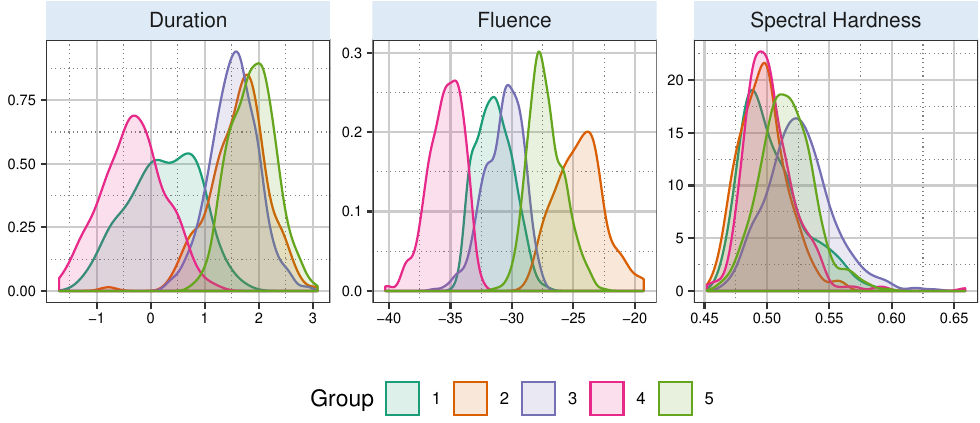}\label{fig:grb.c5.kns.pc}}
  \subfloat[Tik-means: $\hat C = 5$]{\includegraphics[width=.5\textwidth]{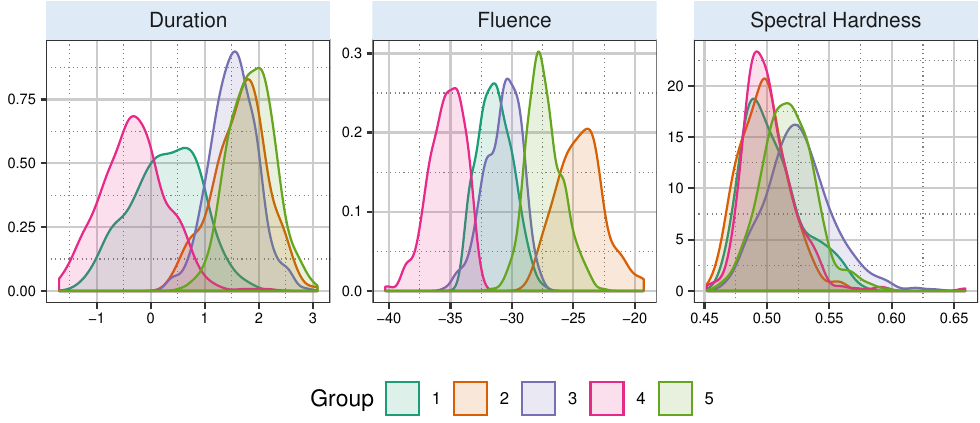}\label{fig:grb.c5.kns.tik}}}
  \caption{Summary of duration ($T_{90}$), total fluence
     ($F_{total}$) and spectral hardness ($H_{321}$) for the groups  obtained using (a)
     KNOB-SynC with the generalized Mahalanobis distance and (b) TiK-means.}\label{fig:grb.kns.c5} 
\end{figure*}

The fields of the BATSE 4Br catalog are heavily correlated in the
log-scale. This has led many researchers to argue for and summarily
ignore all 
but a few variables in their analysis. Here we explore performance
using KNOB-SynC on the first three PCs (accounting for
96.27\% of the total variance) of the nine scaled log-transformed
variables which is equivalent to using KNOB-SynC with the generalized
Mahalanobis  distance. The $k$-means phase applied on the 3 PCs finds 5 groups.
KNOB-SynC does not enter the merging stage at all
since the maximum and generalized overlaps are the same. Comparison of 
our results~(Figure \ref{fig:grb.kns.c5} and
Table~\ref{tab:grb.kns.c5}) with those of~\citet{berryandmaitra19}
shows fairly good agreement in the confusion matrix ($\mR=0.868$). 
The first and the fourth groups have short
burst durations ($T_{90}$) and soft spectral hardness ($H_{321}$)
although the first group has fainter total fluence ($F_{total}$). Groups 2
and 3 have long durations and bright fluences but different spectral
hardness. Group 5 has long duration GRBs but with intermediate fluence and
spectral hardness. The true number and kinds of GRB groups is not
known but our results show that the KNOB-SynC solution yields groups that are 
distinct, interpretable and in line with the newer results obtained by
TiK-means ~\citep{berryandmaitra19} or
MBC~\citep{chattopadhyayandmaitra17,chattopadhyayandmaitra18}. 
\begin{table}[ht]
  \centering
  \caption{(a) Summary of duration ($T_{90}$), fluence ($F_{total}$) and
    spectral hardness ($H_{321}$) for each of the five groups
    obtained using KNOB-SynC with the generalized Mahalanobis distance
    and TiK-means. (b) Confusion matrix of the TiK-means and KNOB-SynC
    solutions ($\mR = 0.868$).}\label{tab:grb.kns.c5}
  \subfloat[]{
\begin{tabular}{rcccccl}
  \hline
  \hline
  & $k$ & $n_k$ & $T_{90}$ & $F_{total}$ & $H_{321}$ & $T_{90}-F_{total}-H_{321}$\\
  \hline
\multirow{5}{*}{\rotatebox[origin=c]{90}{KNOB-SynC}}  & 1 & 207 & $0.234 \pm
                                                           0.003$ &
                                                                    $-31.423\pm0.007$ & $0.505\pm 10^{-4}$ & short-intermediate-soft\\ 
  & 2 & 187 & $1.619\pm 0.003$ & $-24.492\pm 0.01$ & $0.497\pm 10^{-4}$ & long-bright-soft\\ 
& 3 & 459 & $1.543\pm 0.001$ & $-30.682\pm 0.003$ & $0.526\pm  10^{-4}$ & long-intermediate-hard\\ 
  & 4 & 318 & $-0.32\pm0.002$ & $-35.381\pm 0.004$ & $0.503\pm  10^{-4}$ & short-faint-soft\\  
  & 5 & 428 & $1.882\pm 0.001$ & $-27.235\pm 0.003$ & $0.516\pm 10^{-4}$ & long-bright-hard\\  
  \hline
  \hline
  \multirow{5}{*}{\rotatebox[origin=c]{90}{TiK-means}}   & 1 & 197 & $0.272\pm0.003$ & $-31.349\pm0.007$ & $0.505\pm  10^{-4}$ & short-intermediate-soft\\
  & 2 & 188 & $1.65\pm0.003$ & $-24.439\pm0.01$ & $0.498\pm 10^{-4}$ & long-bright-soft\\
  & 3 & 429 & $1.531\pm0.001$ & $-30.727\pm0.003$ & $0.526\pm  10^{-4}$ & long-intermediate-hard\\
  & 4 & 333 & $-0.304\pm0.002$ & $-35.302\pm0.004$ & $0.502\pm 10^{-4}$ & short-faint-soft\\
  & 5 & 452 & $1.867\pm0.001$ & $-27.362\pm0.003$ & $0.517\pm  10^{-4} $ & long-bright-hard\\ 
  \hline
  \hline
\end{tabular}
}

\subfloat[]{
\begin{tabular}{lrrrrrr}
  \hline
  & \multicolumn{6}{c}{KNOB-SynC}\\
  \cline{3-7}
  & $k$ & 1 & 2 & 3 & 4 & 5 \\ 
  \hline
 \multirow{5}{*}{\rotatebox[origin=c]{90}{Tik-means}} & 1 & 188 &   2 &   2 &   3 &   2 \\ 
&  2 &   0 & 181 &   0 &   0 &   7 \\ 
&  3 &   1 &   0 & 420 &   2 &   6 \\ 
&  4 &  13 &   0 &   7 & 313 &   0 \\ 
&  5 &   5 &   4 &  30 &   0 & 413 \\ 
   \hline
\end{tabular}
}
\end{table}

\subsection{Activation detection in a fMRI finger-tapping task experiment}\label{sec:app.fMRI2}
Our second application uses KNOB-SynC to identify activation in fMRI
experiments. One objective of fMRI is to determine cerebral regions that respond to a task or particular stimulus~\citep{bandettinietal93,belliveauetal91,kwongetal92,ogawaetal90}. 
A typical approach relates, after correction and pre-processing, the observed  Blood Oxygen Level Dependent (BOLD) 
time course sequence at each image voxel to the expected BOLD response~\citep{fristonetal94,glover99,lazar08} by fitting a general linear 
model~\citep{fristonetal95} and obtaining a test statistic (often a $t$-statistic) that tests for significance  at that voxel. Thresholding 
methods~\citep{formanetal95,Genoveseetal2002} are often used on these $t$-statistics to determine
activation. Attempts to use clustering algorithms have been made, but
~\citet{thirionetal14} found that despite the advantages of speed and
simplicity, $k$-means is not, in general, a good performer because it
fits ``data idiosyncracies'' and pathologies. We therefore explore if
KNOB-SynC can improve the $k$-means clustering solution on these datasets.

Our dataset for this experiment is from a right-hand finger-tapping
experiment of a right-hand-dominant male and was acquired over twelve
regularly-spaced sessions in a two-month span. We choose
only 5 of these sessions that were identified in~\citet{maitra10} as
the ones with the highest reliability. Because there is no known gold
standard, our comparison here will be of the five partitionings
detected in each replication with each other. 
Each dataset
was preprocessed and voxel-wise $Z$-scores were obtained that
quantified the test statistic under the hypothesis of no activation at
each voxel.  We refer to~\citet{maitraetal02} and ~\citet{maitra09b}
for imaging details. At each of the $n = 179364$ voxels, we compute
the $Z$-scores to test the hypothesis that 
the expected BOLD levels are significantly related to the right-hand
tapping at a voxel. These $Z$-scores for each replication are our
(one-dimensional) dataset. Because of the large size of the dataset,
most competing methods are impractical to apply, so we only use
KNOB-SynC here. (For computational reasons also, we do not estimate
$K_0$ in the $k$-means phase but set it at $K_0=50$.)

\begin{figure*}[h]
  \centering
  \hspace{-0.04\textwidth}
  \begin{minipage} {0.78\textwidth}
    \mbox{\subfloat[KNOB-SynC]{\includegraphics[width=0.45\textwidth]{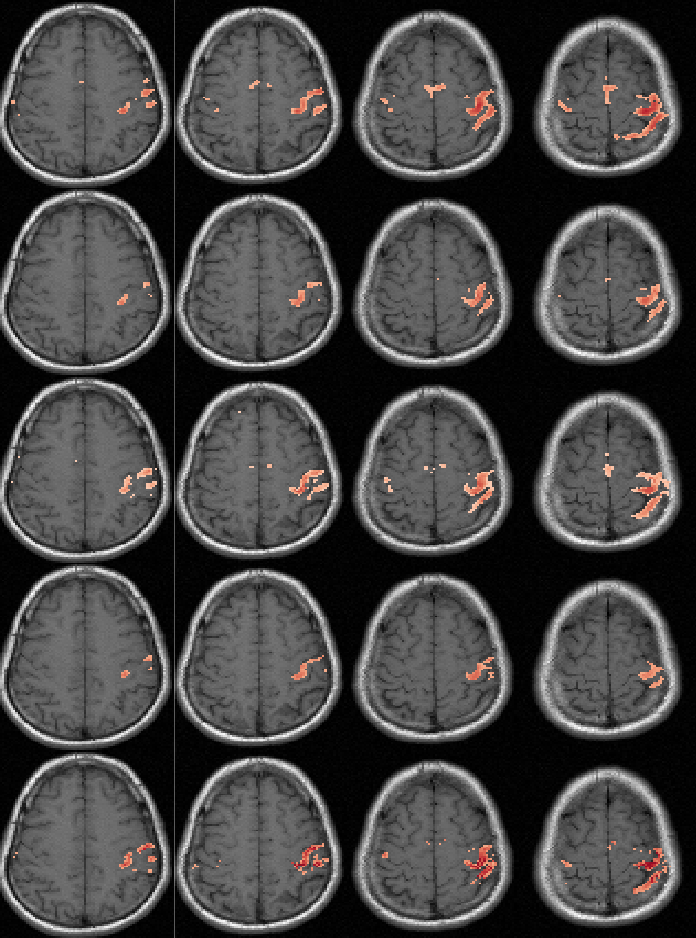}\label{fig:finger.tapping.KNS.C2}}
    \hspace{-0.005\textwidth}
  \subfloat[AR-FAST]{\includegraphics[width=0.45\textwidth]{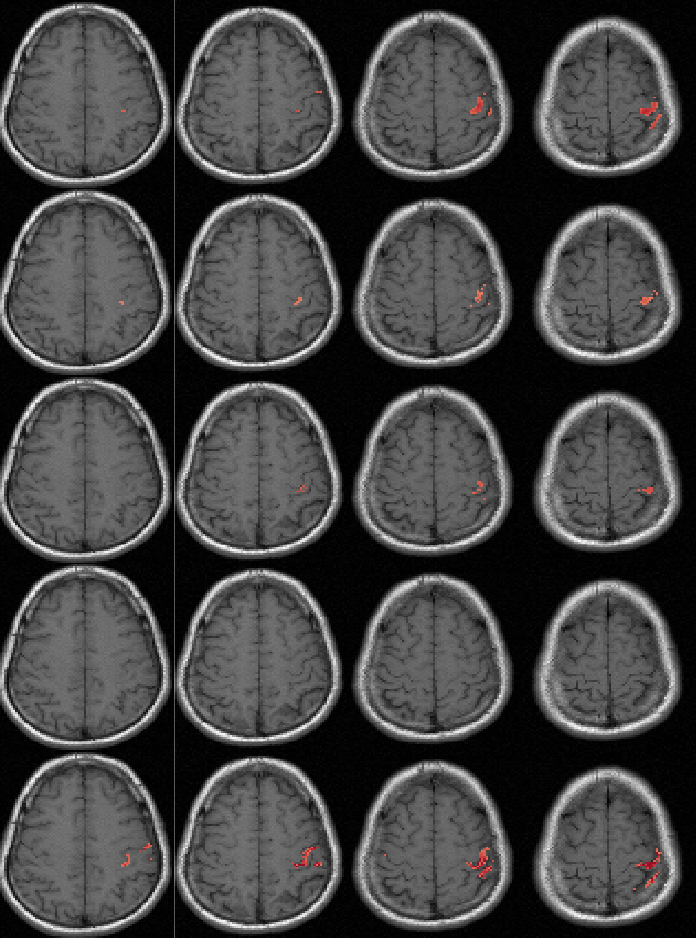}\label{fig:finger.tapping.AR}}
\hspace{-0.01\textwidth}
  \subfloat{\includegraphics[width=0.05\textwidth]{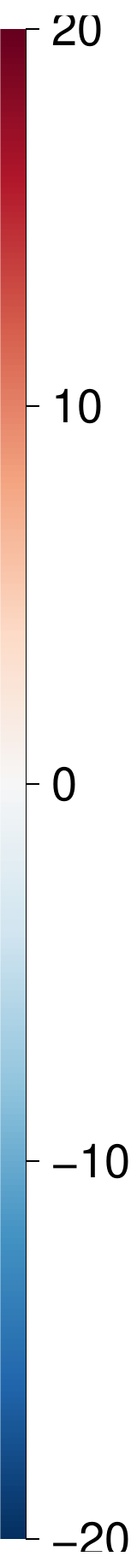}}
}
\addtocounter{subfigure}{-1}
\end{minipage}
\hspace{-0.04\textwidth}
\begin{minipage}{0.19\textwidth}
    \mbox{
      \subfloat[]{\includegraphics[width=\textwidth]{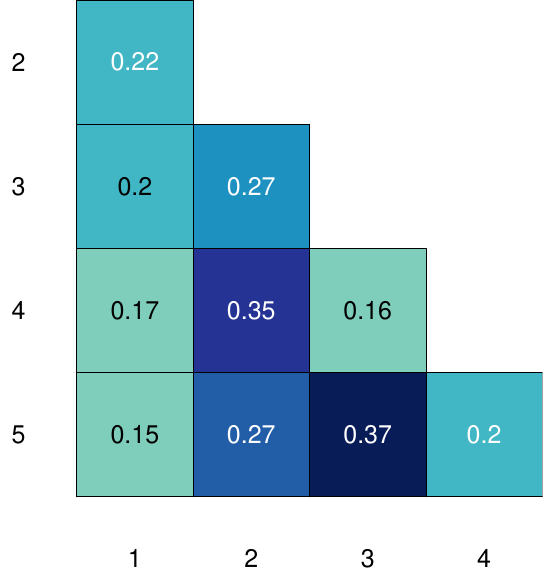}\label{fig:JI.kns}}}
    \mbox{
      \subfloat[]{\includegraphics[width=\textwidth]{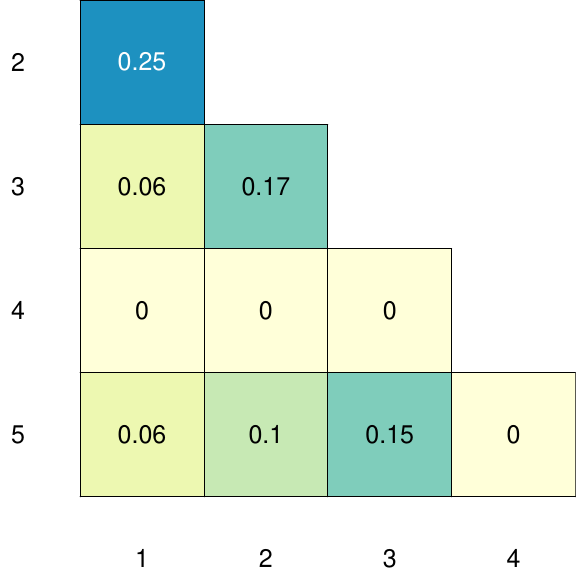}\label{fig:JI.ar}}
    }
  \end{minipage}
  \hspace{0pt}
  \begin{minipage}{0.05\textwidth}
    {\includegraphics[width=\textwidth]{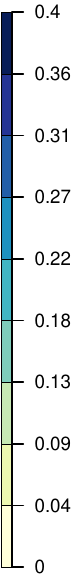}}
  \end{minipage}
  \caption{Observed $Z$-scores at the voxels in the smaller
    (activated) group for the right-hand finger-thumb opposition task
    experiments obtained using (a) KNOB-SynC and (b) AR-FAST. For each
    set of experiments, we display activation maps for the 18th, 19th, 20th
    and 21st slices in each column. The replications are represented by
    rows. Jaccard indices of activation between each pair of
    replicates using (c) KNOB-SynC and (d) AR-FAST algorithms.}\label{fig:finger.tapping} 
\end{figure*}

The 50 homogeneous $k$-means groups in  each of the five replicates
when supplied to the merging phase each terminated with 
$\hat C = 2$ syncytial groups. For the first replicate, the largest
group has 178307 (99.4\%) voxels -- this is essentially the region of
no activation. The other replicates have 178898 (99.7\%), 178129
(99.3\%), 179087 (99.8\%), and 178658 (99.6\%) voxels in this group.
Figure~\ref{fig:finger.tapping.KNS.C2}
displays the $Z$-scores at the activated voxels over four slices of
the brain. The displayed slices
comprise the ipsi- and contra-lateral pre-motor cortices (pre-M1), the
primary motor cortex (M1), the pre-supplementary motor cortex
(pre-SMA), and the supplementary motor cortex (SMA). We see broad agreement between the replications in each of the four
slices. We compare our KNOB-SynC results with the robust adaptive smoothed
thresholding (AR-FAST) algorithm of ~\citet{almodovarandmaitra19}
implemented by the {\tt R} package {\tt
  RFASTfMRI}~\citep{almodovarandmaitra19a} which shows far less agreement
among the 5 replicates in terms of detected activation in the four
slices. Specifically, KNOB-SynC (Figure
\ref{fig:finger.tapping.KNS.C2}) identifies activation in the left M1
and in the ipsi-lateral pre-M1 areas. There is some identified
activation in the contra-lateral pre-M1, pre-SMA and SMA voxels. On
the other hand, AR-FAST (Figure \ref{fig:finger.tapping.AR}) finds less
activation in the left M1 and in the ipsi-lateral pre-M1
areas. Figure~\ref{fig:JI.kns} displays the \citet{jaccard1901} index
of the activation detected (using KNOB-SynC) 
between each pair of replications. Figure~\ref{fig:JI.ar} displays
similar Jaccard index calculations with regard to activation detected
using AR-FAST. The Jaccard indices are higher for KNOB-SynC-found
activation for each pair of replications and show greater
reproducibility. The summarized Jaccard index of 
\citet{maitra10} which provides  an overall measure of reproducibility
of activation detected across replicates is 0.238 for KNOB-SynC and
0.102 for AR-FAST which was shown~\citep{almodovarandmaitra19}
to be a top performer on this dataset. We comment that while the
overall Jaccard indices are low for both methods, the low value of
0.238 also reflects the challenge of 
activation detection in single-subject fMRI. Seen in this context,
KNOB-SynC does quite well. This example illustrates the potential of
KNOB-SynC to improve and refine clustering solutions making it possible, for
instance, to use $k$-means and to alleviate some of the concerns
raised in ~\citet{thirionetal14}.

%% file: discussion.tex
\section{Discussion}\label{discussion}
This paper has proposed a syncytial clustering algorithm
called KNOB-SynC that merges groups found by standard clustering
algorithms such as $k$-means, and does so in a data-driven and fully  
objective way. 
A R package called {\sc SynClustR} implements our method in the function
{\tt KNOBSynC} and the competing K-mH syncytial algorithm  in the function
{\tt kmH} and is publicly available at {\tt  https://github.com/ialmodovar/SynClustR}.  
Our method is distribution-free and can  apply to the results of many
standard clustering algorithms. We use the overlap measure of
\citet{maitraandmelnykov10} for merging and for decisions but use
kernel-based nonparametric methods to calculate this overlap.
Our algorithm has no parameters that require fine-tuning by the user
and, as pointed out by a reviewer, shows robust performance across
many datasets of many dimensions and with little to tremendous
complexity, when compared against a host of other methods.
Further, our methodology is general enough to extend to situations with
incomplete records or where clustering is done in the presence of 
scatter. Application of KNOB-SynC to data from the
BATSE 4Br catalog provides further evidence of five kinds of
GRBs. Our approach is also demonstrated to potentially make it 
possible to adapt $k$-means clustering for activation detection in
fMRI.


This paper also developed estimation methods of the CDF using the
asymmetric RIG kernel. We used the plugin-bandwidth selector that
minimizes the MISE as our bandwidth choice but it would be good to
develop and investigate more sophisticated approaches. Further, our
development in this paper 
provides an opportunity to develop nonparametric methods for
diagnostics in clustering. For instance, our developed kernel CDF
estimator could be used to determine uncertainties in $k$-means
classifications. A reviewer has also very kindly drawn our attention
to the fact that the 
construction of composite clusters that underlies the idea behind
syncytial clustering has also been used in the context of
semi-supervised clustering~\citep{smiejaandwiercioch17} where,
instead of estimated overlap as used in this paper, available class labels
are used in deciding to merge pairs of groups. We believe that such
an approach may also benefit from our methodology, especially when not
all classes have representation in the supervised portion of
the dataset. Thus, we see that although we have made an important
contribution, a number of issues remain that would benefit from
further  attention.

%% file: new-proof.tex
\section{}\label{sec:appendix}
\subsection{Proof of Lemma 3} \label{proof-lemma3}
\begin{proof}
We have
$\E[\hat{H}(y;b)] =  \E\left(\frac{1}{n}\sum^n_{i=1} G(y;Y_i,b)\right)
=  \E\left(G(y;Y_1,b)\right) = \int^\infty_{0} G(t;y,b) h(t) \mathrm{d} t 
 =  \int^\infty_0 \left(\int^y_0 K(t;w,b) \mathrm{d}w\right) h(t)
 \mathrm{d}t = 
\int^y_0 \int^\infty_0 K(t;w,b) h(t)\mathrm{d}t \mathrm{d}w = 
\int^y_0 \E[h(V_{1/{(w-b)},1/b})] \mathrm{d}w$,
where the random variable $V_{1/(w-b), 1/b} \sim
\mbox{RIG}[1/(w-b),1/b]$. The last equality holds because the inner
integral $\int^\infty_0 K(t;w,b) h(t)\mathrm{d}t = 
\E[h\{V_{1/(w-b),1/b}\}]$. Then, expanding $V_{1/(w-b),1/b}$ around its mean
$w$ and also using $\V\mbox{ar}\{V_{1/(w-b),1/b}\} = b(w+b)$
yields 
  \begin{equation}
\begin{split}
  \int^y_0  \E[h\{V_{1/(w-b),1/b}\}] \mathrm{d}t  
 & = \int^y_0 h(w) \mathrm{d}w + \frac{1}{2}\int^y_0 (bw +
 b^2)h''(w) \mathrm{d}w + o(b^2)\\
 & = H(y) + \frac b2\int_0^y wh''(w)\mathrm d w + o(b^2)\\
& = H(y) + \frac{by}2h'(y) - \frac b2\left[h(y) - h(0)\right] + o(b^2)\\
&  = H(y) + \frac b2\left[yh'(y) - h(y)\right] + o(b) \equiv H(y) + \mathcal O(b).
\end{split} 
\label{eq:kernel.mean}
 \end{equation}

For the variance, from the definition, 
$\V ar [\hat{H}(y;b)]  =  \V ar\left[ n^{-1}\sum^n_{i=1} G(y;Y_i,b) \right]
 =  \frac{1}{n} \V\mbox{ar} \left[ G(y;Y_1,b) \right]  =  \frac{1}{n} \E\left[ G^2(y;Y_1,b) \right] - \frac{1}{n} \left[\E(G(y;Y_1,b))\right]^2$.
The second term is easily obtained from \eqref{eq:kernel.mean}. It remains
to derive the second moment of the estimator, $\E\left[
  G^2(y;Y_1,b) \right] = \int^\infty_0 G^2(y;t, b) h(t) \mathrm{d} t $
which can be recast as
\begin{equation}
\begin{split}
\E[G^2(y;Y_1,b)] &= \int^\infty_0 G^2(y;t, b) h(t) \mathrm{d} t\\ &
= \int^\infty_0 G(y;t,b)\Phi\left(\sqrt{\frac tb} +\sqrt{\frac bt}\right)  \mathrm{d}t- \int^\infty_0
G(t;y,b)F(t;b)\mathrm{d}t\\
&= \int^\infty_0 G(y;t,b)\Phi\left(\sqrt{\frac tb} +\sqrt{\frac bt}\right)  \mathrm{d}t - \int^y_0\E[F(V_{1/(w-b),1/b};y,b)]\mathrm dw,
\end{split}
\label{eq:var.terms}
\end{equation}
where
$F(t;y,b) = \Phi\left(\sqrt{t/b} + \sqrt{b/t}  -y/\sqrt{tb}\right)h(t)$ using
a similar  random variable $V_{1/(w-b),1/b}$ and tactics as used in
the reductions leading to \eqref{eq:kernel.mean}. Since 
$t,b>0$, we have that $2\leq\sqrt{t/b}+\sqrt{b/t} < \infty$ and so
$\Phi(2) \leq \Phi\left(\sqrt{t/b} +\sqrt{b/t}\right) \leq
1$. Therefore, we have 
\begin{equation*}
  \Phi(2)\int^\infty_0 G(y;t,b)h(t)\mathrm{d} t \leq \int^\infty_0
  \Phi\left(\sqrt{\frac tb} +\sqrt{\frac bt}\right) G(y;t,b)h(t)\mathrm{d} t \leq \int^\infty_0 G(y;t,b)h(t)\mathrm{d} t.
\end{equation*}
But 
 $ \int^\infty_0
G(y;t,b)h(t)\mathrm{d} t \equiv \E[G(y;Y_1,b) = H(y) + b[yh(y)-h(y)]/2 +
o(b)$ and $\Phi(2) = 0.97725$ so that
\begin{equation}
   \int^\infty_0
  \Phi\left(\sqrt{\frac tb} +\sqrt{\frac bt}\right)
  G(y;t,b)h(t)\mathrm{d} t\approx  H(y) + \frac b2[yh(y)-h(y)] + o(b) 
  \label{eq:var.term1}
\end{equation}
For the second term in  \eqref{eq:var.terms}, expanding 
$V_{1/(w-b),1/b}$ around its mean $w$ yields
\begin{equation}
\begin{split}
 \int^y_0\E & [F(V_{1/(w-b),1/b};y,b)]\mathrm dw \\
& = \int_0^y
 F(w;y,b)\mathrm dw + \frac12\int_0^y\V\mbox{ar}(V_{1/(w-b),1/b};y,b)
 F''(w;y,b)\mathrm dw  + o(b)\\
& = \int_0^y  \Phi\left(\sqrt{w/b} + \sqrt{b/w}
  -y/\sqrt{wb}\right)h(w) \mathrm dw + \frac
b2\int_0^y(w+b)F''(w;y,b)\mathrm dw + o(b)\\
& = \Phi(\sqrt{b/y}) H(y) - \int_0^y\left\{ \frac{\mathrm d\mbox{ }}{\mathrm d
    w}  \Phi\left(\sqrt{w/b} + \sqrt{b/w}
    -y/\sqrt{wb}\right)\int h(w)\mathrm dw  \right\}\mathrm dw \\
&\qquad\qquad\qquad\qquad\qquad +  \frac
b2\int_0^y(w+b)F''(w;y,b)\mathrm dw + o(b)
\end{split}
\label{eq:var.term2}
\end{equation}
The derivative in the integrand is 
\begin{equation*}
\frac{\mathrm d\mbox{ }}{\mathrm d     w}  \Phi\left(\sqrt{w/b} + \sqrt{b/w}
  -y/\sqrt{wb}\right) =   \left(\frac1{\sqrt{bw}}+\frac
  {y-b}{w\sqrt{bw}}\right)\phi\left(\sqrt{\frac
    wb}-\frac{(y-b)}{\sqrt{bw}}\right) 
\end{equation*}
so that \eqref{eq:var.term2} equals $\Phi(\sqrt{b/y}) H(y)
+ o(\sqrt b)$. We now use a Taylor series expansion of
$\Phi(\sqrt{b/y})$ around 0 to get $\Phi(\sqrt{b/y}) = 1/2 +
\sqrt{b}/(2\sqrt{2\pi y}) + \mO(b)$. Inserting this result into
\eqref{eq:var.term2} and combining with \eqref{eq:var.term1} means
that \eqref{eq:var.terms} is  
$
\E[G^2(y;Y_1,b)] =  H(y)/2 - H(y) \sqrt{b}/(2\sqrt{2\pi y}) + o(\sqrt b)$
and the approximate expressions for the variance in Lemma
\ref{lemma:kernel.EVar} follow.  
\end{proof}

%% file: appendix-2d.tex
\section{Detailed Experimental Evaluations} \label{appendix-experiments}
Figures~\ref{fig:g7}-\ref{fig:xxxx} illustrate performance on 2D datasets obtained by
KNOB-SynC, K-mH, DEMP, DEMP+, MMC, EAC, GSL-NN, Spectral clustering,
kernel $k$-means, DBSCAN*, Density peaks, PGMM, MSAL and MGHD.  In all
cases, plotting character and color represent the true and 
estimated group indicator. Tables \ref{tab:AR.2d} and \ref{tab:AR.HD} display
 performance, in terms of ${\cal R}$ and estimated $\hat C$, on 2D and
 higher-dimensional  datasets, of KNOB-SynC (denoted as KNS in the
 table), K-mH, DEMP (DM), DEMP+ (DM+), MMC, EAC, GSL-NN (GSN),
 spectral clustering (SpC), kernel-$k$-means (k-$k$m), DBSCAN* (D*),
 density peaks (DP), PGMM (PGM), MSAL (MSL) and MGHD (MHD). For
 k-$k$m, $\hat C$ was set at the true $C$ and not estimated. For each
 method, the absolute deviation of $\mR$ for that method from the
 highest $\mR$ for each dataset was obtained: the average and SD of
 these absolute deviations are also reported for each method in the
 last row. 
 \begin{figure}[h]
   \vspace{-0.05in}
   \centering
  \includegraphics[width = .96\textwidth]{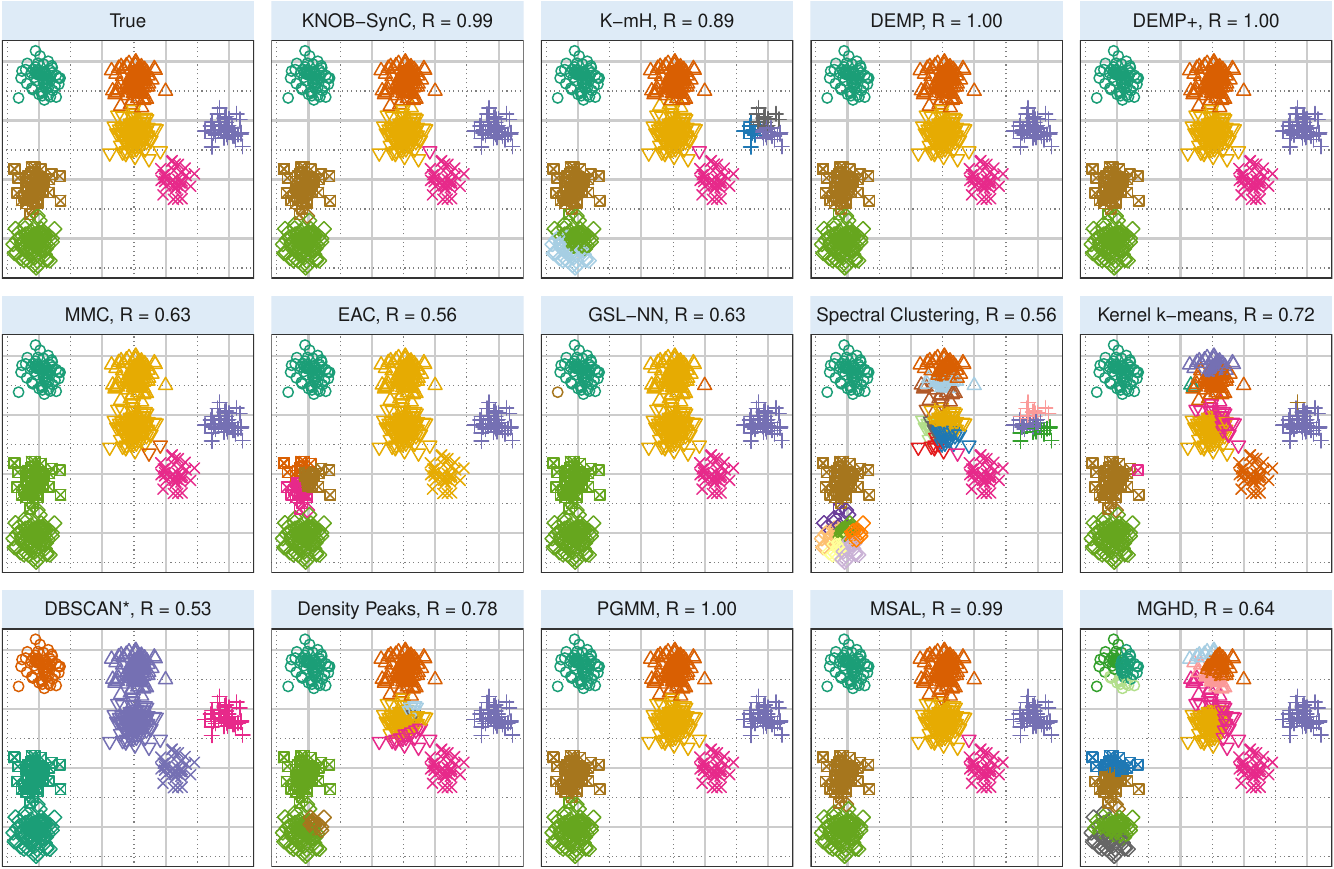}     
  \caption{The {\tt Spherical-7} example, and clusterings obtained using the 14 methods.}
  \label{fig:g7} 
\end{figure}
\begin{figure}
  \centering
  \includegraphics[width = .96\textwidth]{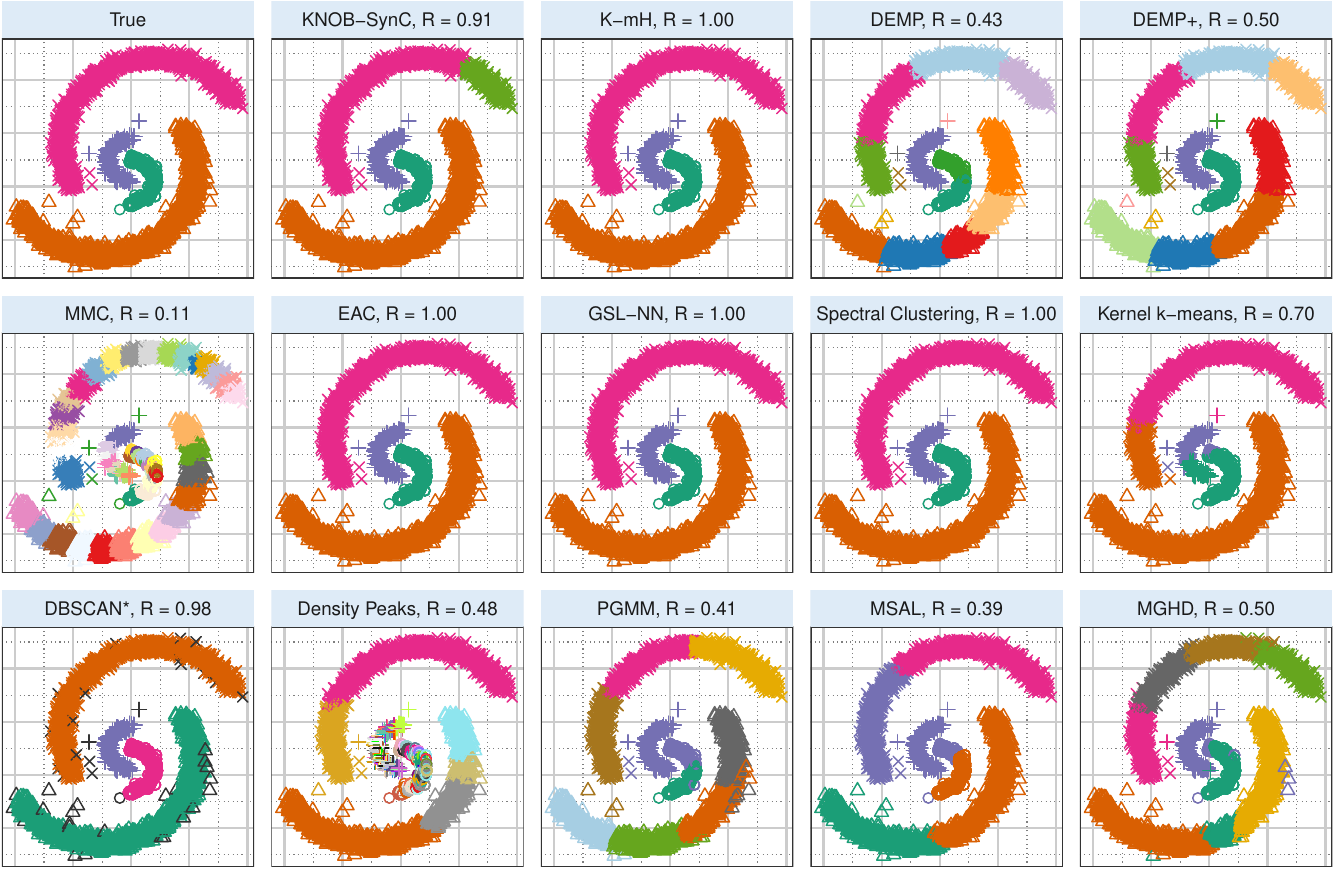}  
  \caption{The  {\tt Bananas-Arcs} example and groupings obtained using the 14 methods.}
  \label{fig:banana.arcs}
\end{figure}
\begin{figure}[h]
   \vspace{-0.05in}
  \centering
  \includegraphics[width = .96\textwidth]{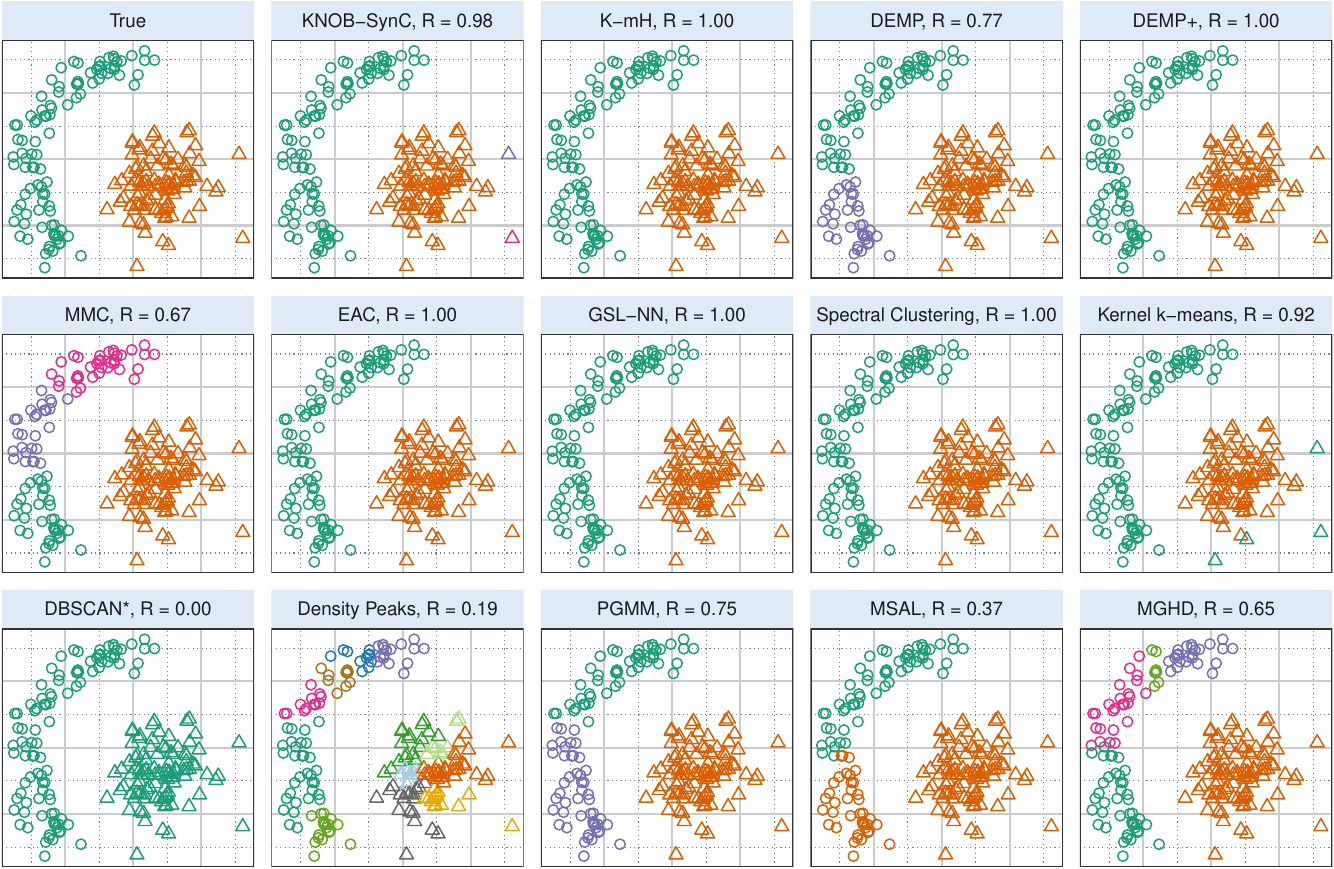}
  \caption{The {\tt Banana-clump} example and groups obtained using the 14 methods.}
  \label{fig:banana-clump}
\end{figure}
\begin{figure}
  \centering
  \includegraphics[width = .96\textwidth]{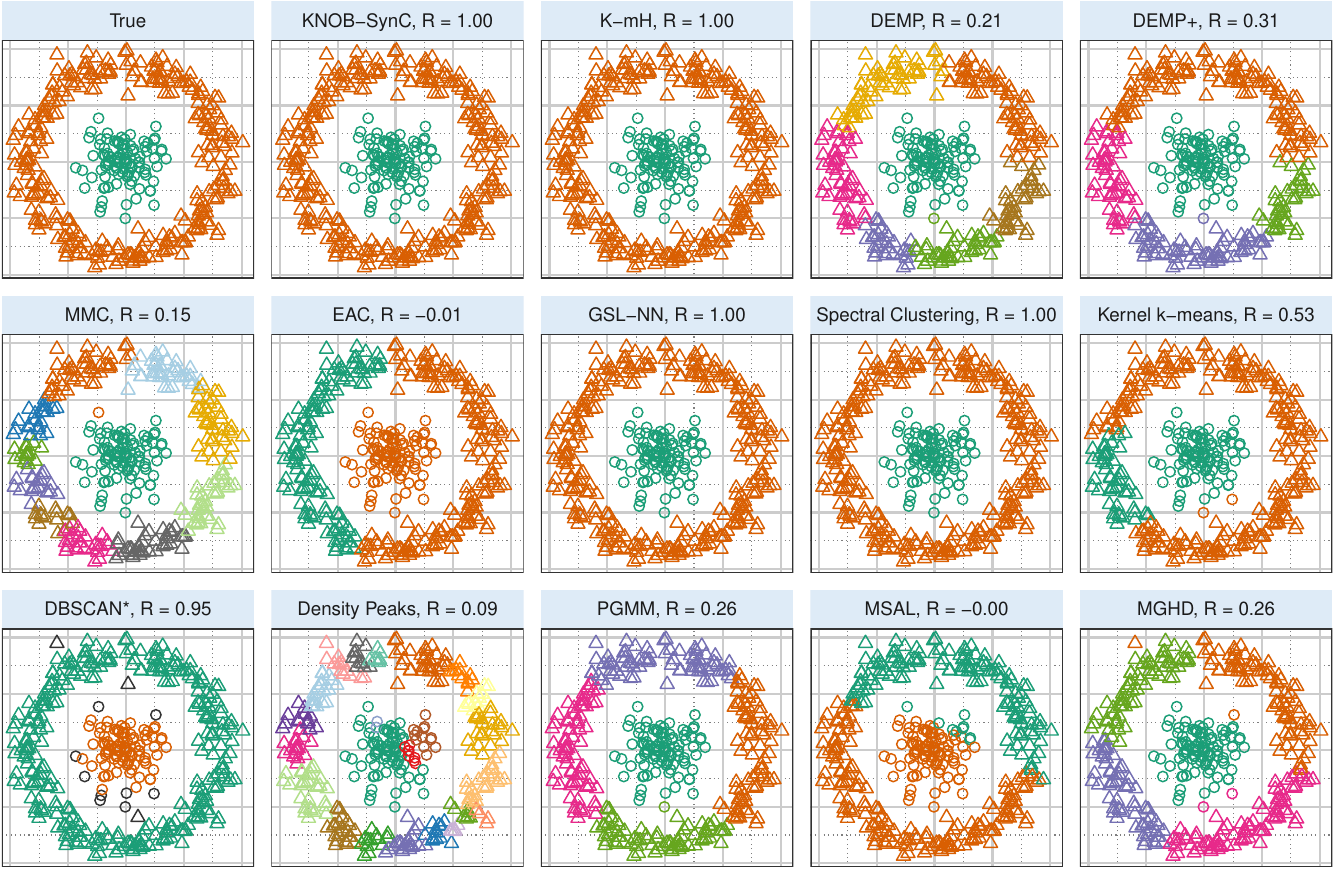}
  \caption{The {\tt Bullseye} example and clusters obtained using the 14 methods.}
  \label{fig:bullseye}
\end{figure}
\begin{figure}
   \vspace{-0.05in}
  \centering
 \includegraphics[width = .96\textwidth]{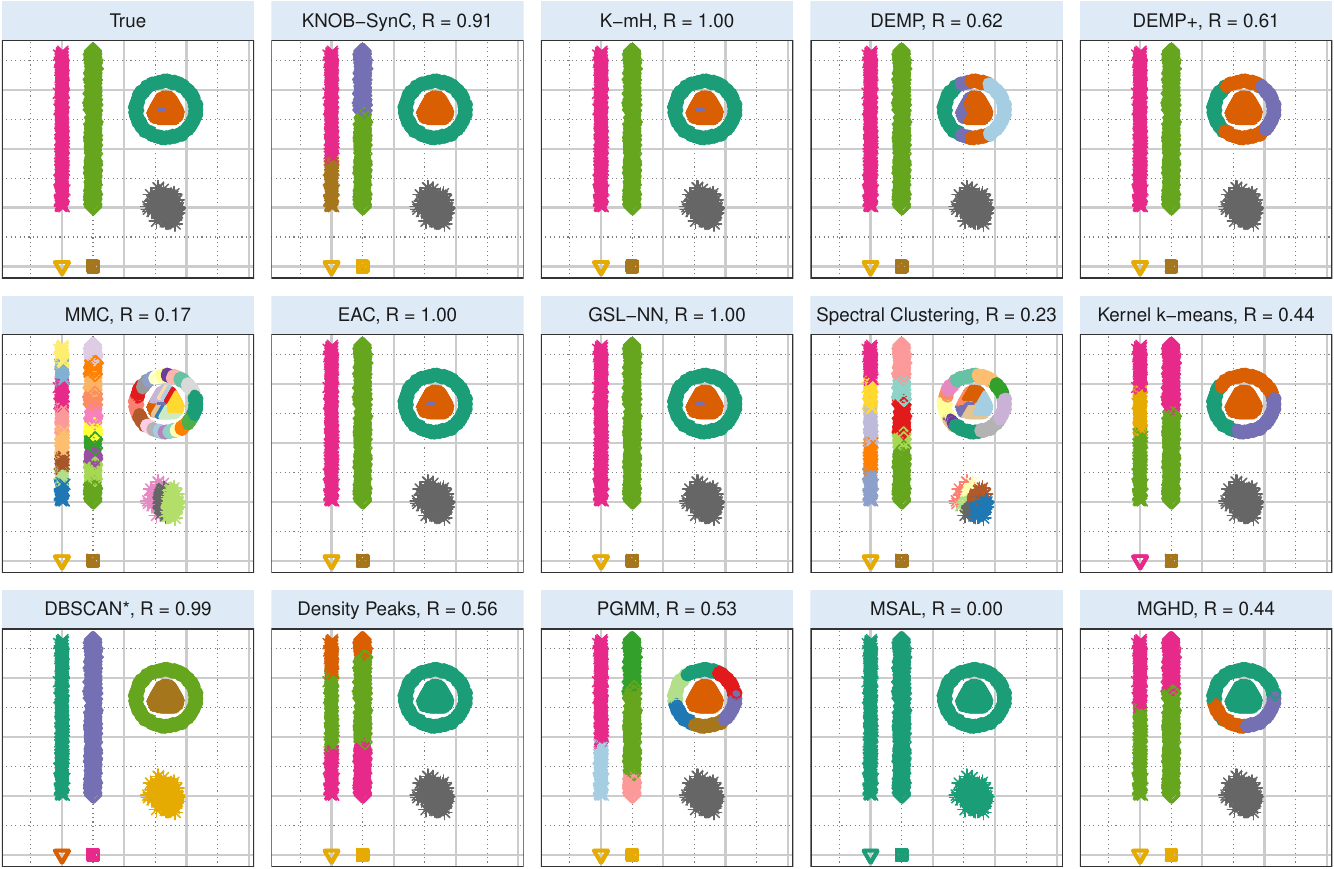}
 \caption{The {\tt Bullseye-Cigarette} example and groups obtained with the 14 methods.}
 \label{fig:bullseye-cig}
\end{figure}
\begin{figure}
  \centering
 \includegraphics[width = .96\textwidth]{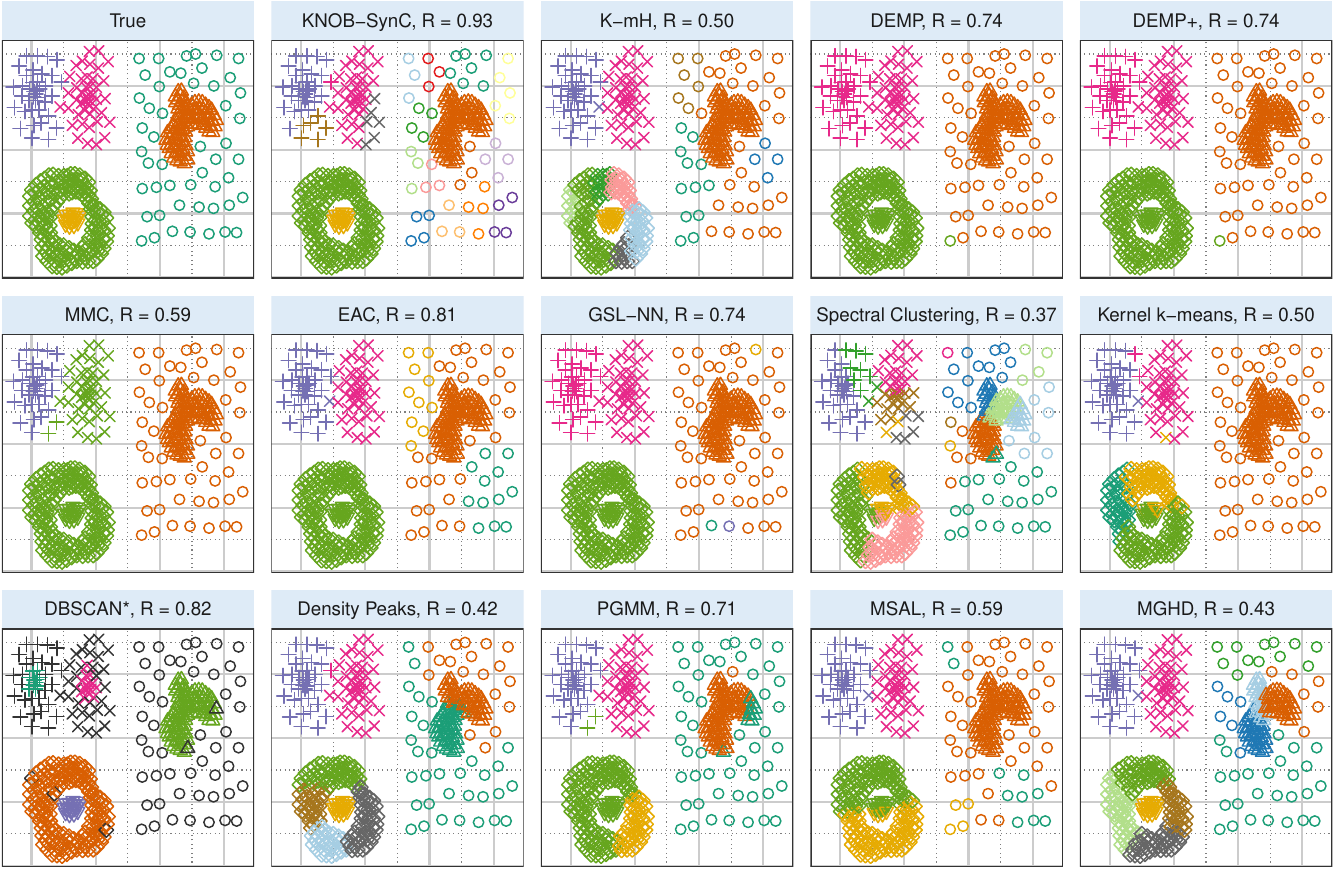}
 \caption{The {\tt Compound} example and partitionings obtained using the 14 methods.}
 \label{fig:compound}
\end{figure}
\begin{figure}
   \vspace{-0.05in}
  \centering
  \includegraphics[width = .95\textwidth]{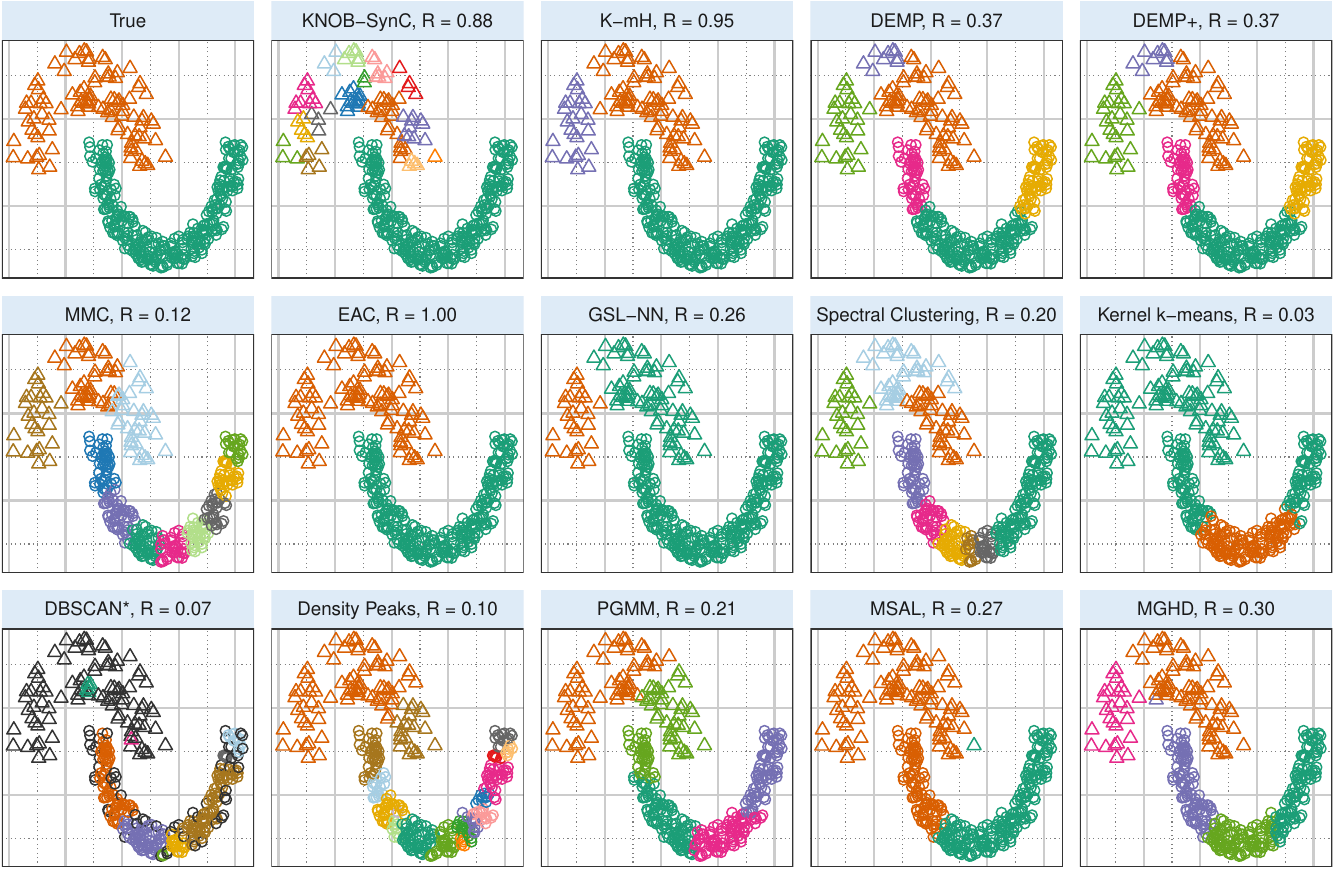}
  \caption{The {\tt Half-Ringed clusters} example and groups obtained with the 14 methods.}  
\label{fig:jain}
\end{figure}
\begin{figure}
  \centering
  \includegraphics[width = .95\textwidth]{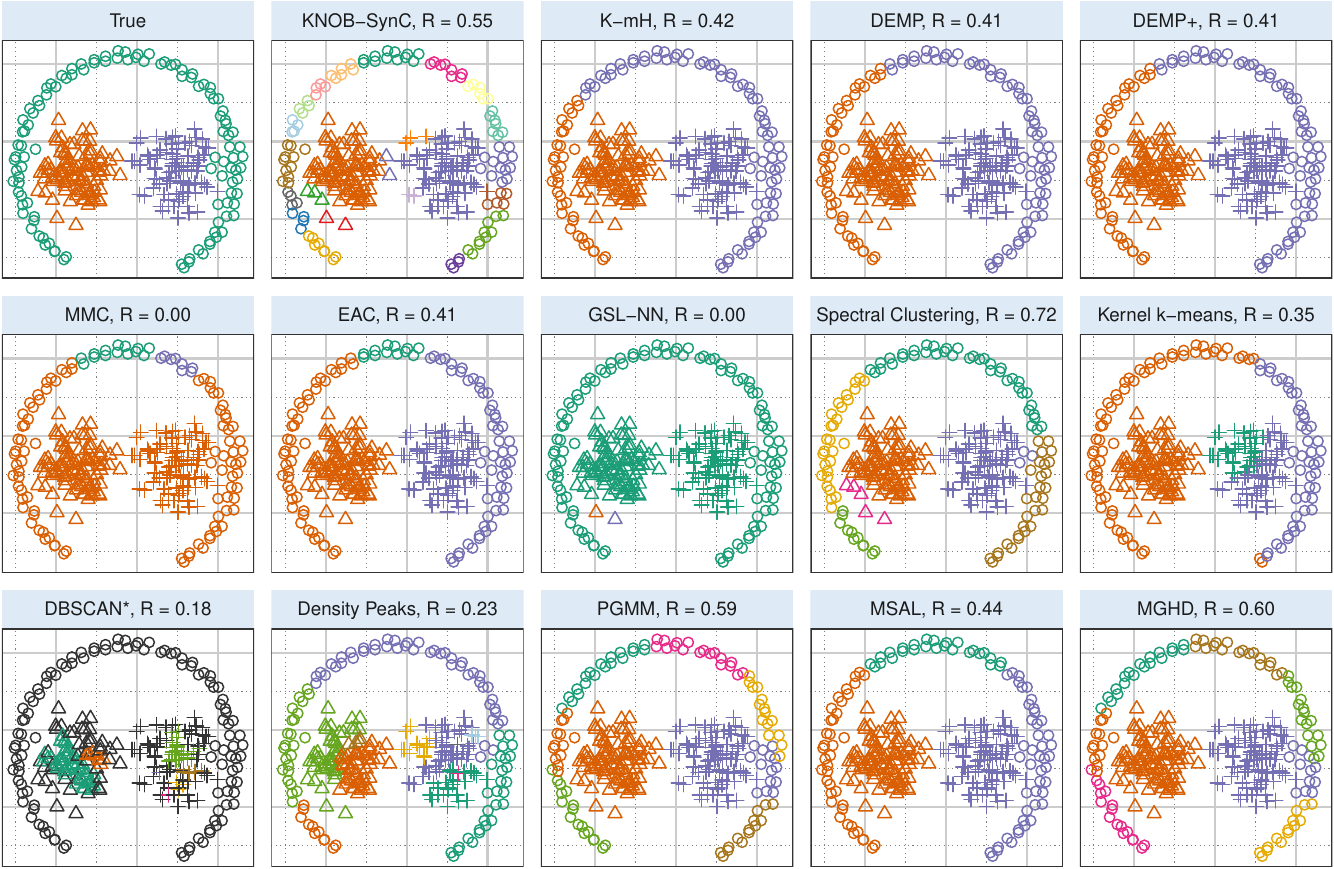}
  \caption{The {\tt Path-based} example and groups obtained with the 14 methods.}
  \label{fig:pathbased}
\end{figure}
\begin{figure}
   \vspace{-0.05in}
  \centering
  \includegraphics[width = .96\textwidth]{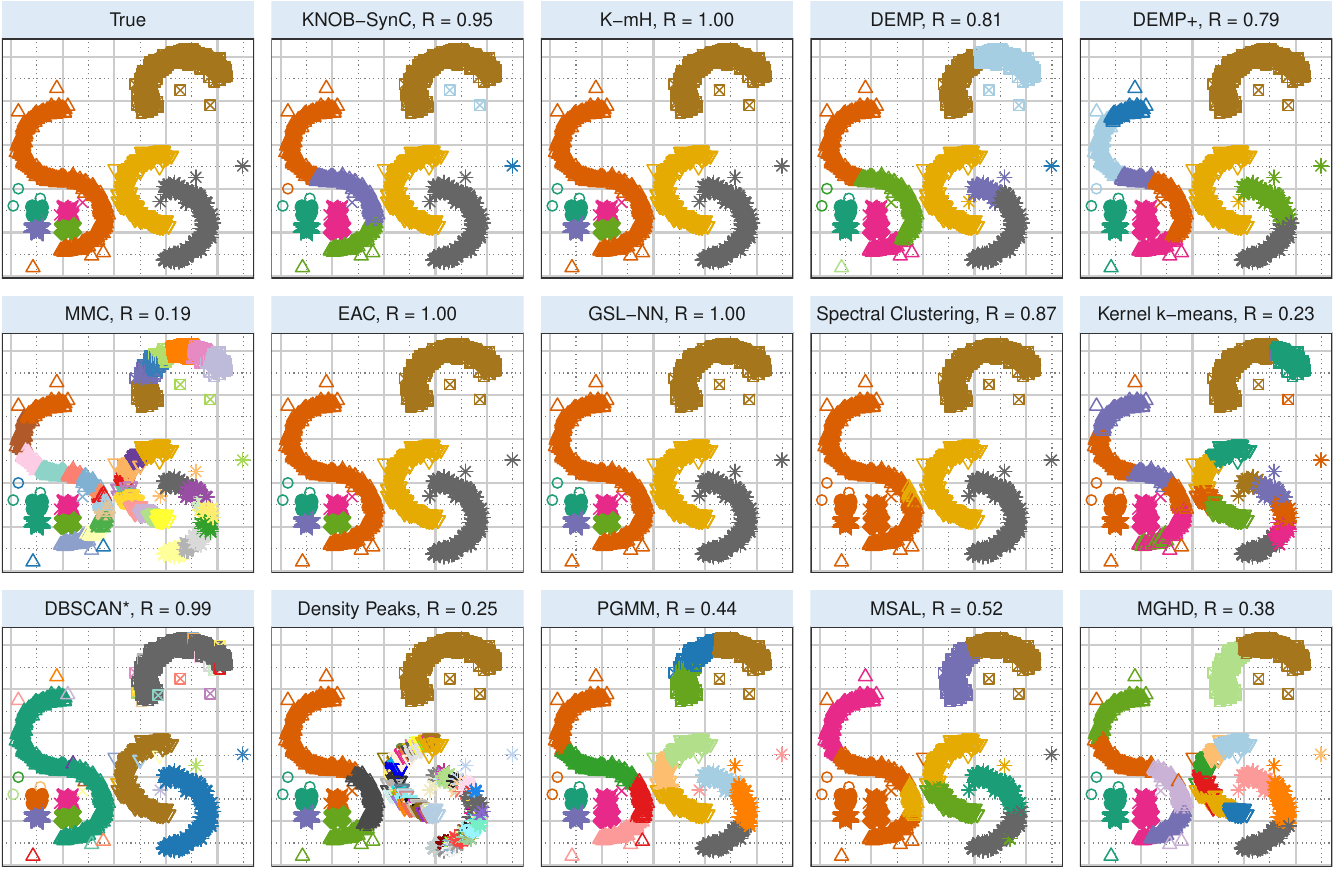}
  \caption{The {\tt SCX-Bananas} example and clusters obtained using the 14 methods.}
  \label{fig:scx.bananas}
 \end{figure}
\begin{figure}
  \centering
 \includegraphics[width = .96\textwidth]{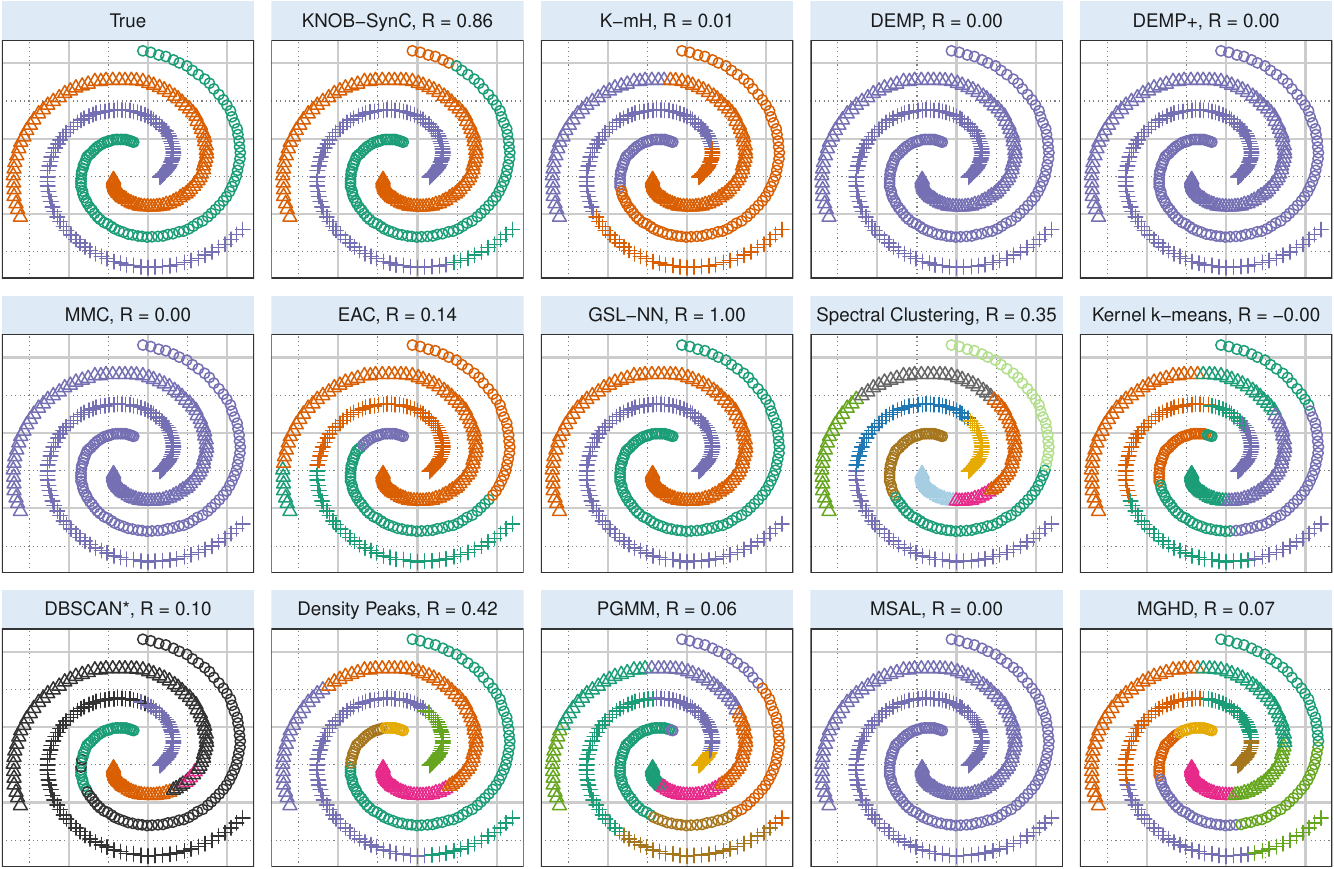}
 \caption{The {\tt Spiral} example and groupings obtained using the 14 competing methods.}
 \label{fig:spiral}
\end{figure}
\begin{figure}
   \vspace{-0.05in}
  \centering
\includegraphics[width = .96\textwidth]{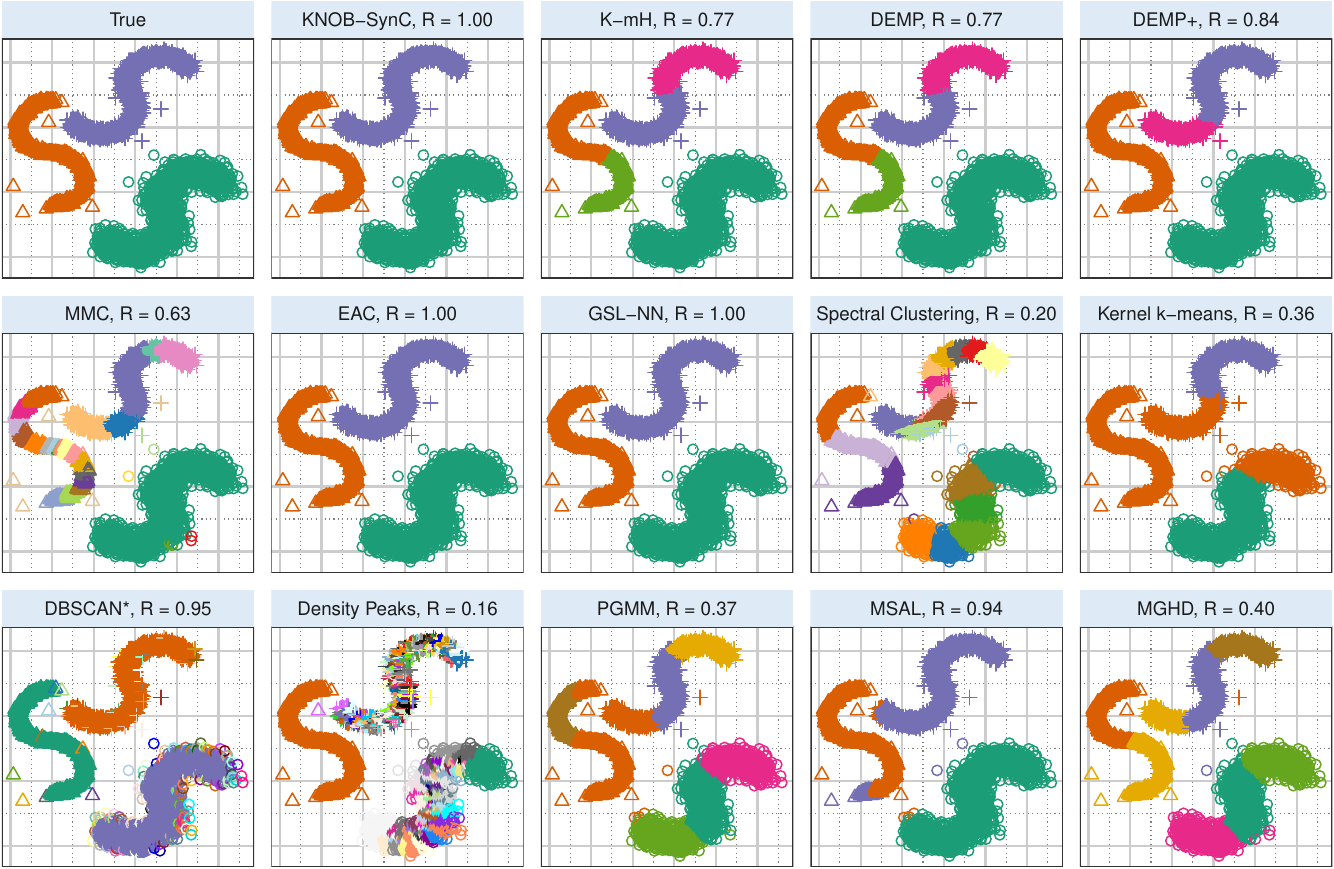}
\caption{The {\tt SSS} example and clusters obtained with the 14 competing methods.}
\label{fig:sss}
\end{figure}
\begin{figure}
  \centering
  \includegraphics[width = .96\textwidth]{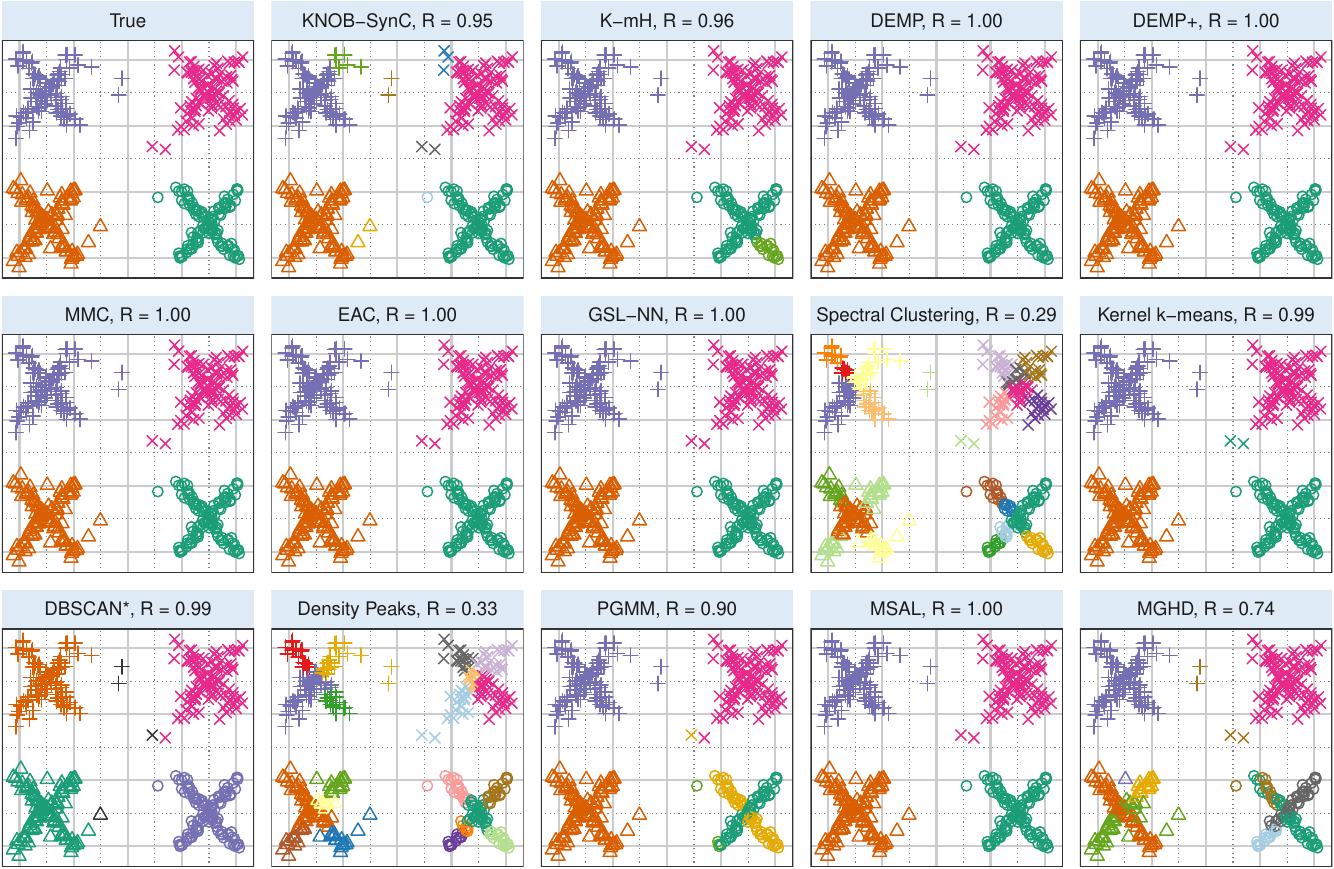}
  \caption{The {\tt XXXX} example and groupings obtained using the 14 competing methods.}
  \label{fig:xxxx}
\end{figure}

\begin{table}[h]
  \centering
  \caption{Performance, in terms of ${\cal R}$ and estimated $\hat C$, on 2D datasets, of competing methods. }  \label{tab:AR.2d}
  \begin{tabular}{lcccccccccccccc}
    \hline
    \hline
    Dataset & KNS & K-mH & DM & DM+ & MMC & EAC & GSN & SpC & k$k$-m & D* & DP & PGM & MSL & MHD\\
    $(n,p,K)$     &  & & & & &  &  & &  & &  &  &  & \\
    \hline
    7-Spherical  & 0.99 & 0.89 & 1 & 1 & 0.63 & 0.56 & 0.63 & 0.56 & 0.72 & 0.53 & 0.78 & 1  & 0.99 & 0.64\\
    $(500,2,7)$ & 7 & 10 & 7 & 7 & 5 & 7 & 5 & 20 & 7 & 4 & 9 & 7 & 7 & 13 \\
    \hline
     Aggregation  & 0.98 & 0.95 & 0.91 & 0.91 & 0.80 & 0.90 & 0.81 & 0.59 & 0.24 & 0.58 & 0.26  & 0.64 & 0.82 & 0.43 \\
   $(788,2,7)$  & 7 & 9 & 6 & 6 & 8 & 12 & 11 & 7 & 7 & 277 & 54 & 12 & 7 & 15 \\
     \hline
     Banana-arcs & 0.91 & 1  & 0.43 & 0.50 & 0.11 & 1 & 1 & 1 & 0.7 & 0.98 & 0.48  & 0.41 & 0.39 & 0.50\\
     $(4515,2,4)$  & 5 & 4 & 11 & 10 & 52 & 4 & 4 & 4 & 4 & 47 & 460 & 9 & 4 & 8 \\
     \hline
     Banana-clump & 0.98 & 1 & 0.77 & 1 & 0.67 & 1 & 1 & 1 & 0.92 & 0 & 0.19 &0.75  & 0.37 & 0.65\\
     $(200,2,2)$   & 4 & 2 & 3 & 2 & 4 & 2 & 2 & 4 & 2 & 1 & 12 & 3 & 2 & 5 \\
     \hline
     Bullseye & 1 & 1 & 0.21 & 0.31 & 0.15 & -0.01 & 1 & 1 & 0.53 & 0.95 & 0.09  & 0.26 & -0.00 & 0.26\\
     $(400,2,2)$      & 2 & 2 & 7 & 5 & 11 & 2 & 2 & 2 & 2 & 14 & 23 & 5 & 2 & 5 \\
     \hline
     Bullseye-Cig & 0.91 & 1 & 0.62 & 0.61 & 0.17 & 1 & 1 & 0.23 & 0.44 & 0.99 & 0.56  & 0.53 & 0 & 0.44\\
     $(3025,2,8)$      & 7 & 6 & 7 & 7 & 5 & 9 & 14 & 16 & 8 & 17 & 115 & 4 & 1 & 4 \\
     \hline
     Compound & 0.93 & 0.5 & 0.74 & 0.74 & 0.59 & 0.81 & 0.74 & 0.37 & 0.5 & 0.82 & 0.42  & 0.71 & 0.59 & 0.43\\
     $(399,2,6)$    & 19 & 13 & 5 & 5 & 5 & 6 & 6 & 13 & 6 & 117 & 9 & 6 & 6 & 12 \\
     \hline
     Half-ringed & 0.88 & 0.95 & 0.37 & 0.37 & 0.12 & 1 & 0.26 & 0.20 & 0.03 & 0.07 & 0.1  & 0.21 & 0.27 & 0.3\\
    $(373,2,2)$    & 16 & 3 & 6 & 6 & 11 & 2 & 2 & 8 & 2 & 155 & 17 & 5 & 2 & 5 \\
     \hline
     Path-based & 0.55 & 0.42 & 0.41 & 0.41 & 0 & 0.41 & 0 & 0.72 & 0.35 & 0.18 & 0.23  & 0.59 & 0.44 & 0.60\\
       $(300,2,3)$  & 21 & 2 & 2 & 2 & 3 & 15 & 3 & 3 & 3 & 217 & 9 & 7 & 3 & 7 \\
     \hline
     SCX-Bananas & 0.95 & 1 & 0.81 & 0.79 & 0.19 & 1 & 1 & 0.87 & 0.23 & 0.99 & 0.25  & 0.44 & 0.52 & 0.38\\
     $(3420,2,8)$  & 8 & 8 & 7 & 9 & 39 & 8 & 8 & 4 & 8 & 25 & 389 & 16 & 8 & 17 \\
     \hline
     Spiral & 0.86 & 0.01 & 0 & 0 & 0 & 0.14 & 1 & 0.35 & -0 & 0.1 & 0.42  & 0.06 & 0 & 0.09\\
     $(312,2,3)$ & 3 & 2 & 1 & 1 & 1 & 9 & 3 & 11 & 3 & 208 & 7 & 7 & 1 & 7 \\
     \hline
     SSS & 1 & 0.77 & 0.77 & 0.84 & 0.63 & 1 & 1  & 0.20 & 0.36 & 0.95 & 0.16  & 0.37 & 0.94 & 0.40\\
      $(5015,2,3)$ & 3 & 5 & 5 & 4 & 27 & 4 & 3 & 19 & 3 & 29 & 385 & 7 & 3 & 7 \\
     \hline
XXXX & 0.95 & 0.96 & 1 & 1 & 1 & 1 & 1 & 0.29 & 0.99 & 0.99 & 0.33  & 0.90 & 1 & 0.74\\
   $(415,2,4)$      & 10 & 5 & 4 & 4 & 4 & 4 & 4 & 20 & 4 & 8 & 20 & 6 & 4 & 9 \\
     \hline
    \hline
    $\bar{\cal D}$ & 0.06 & 0.17 & 0.35 & 0.32 & 0.58 & 0.22 & 0.17 & 0.40 & 0.53 & 0.35 & 0.64 & 0.44 & 0.48 & 0.52 \\
    ${\cal D}_\sigma$ & 0.06 & 0.28 & 0.31 & 0.31 & 0.33 & 0.35 & 0.27 & 0.34 & 0.31 & 0.39 & 0.20 & 0.29 & 0.38 & 0.21 \\
    \hline
    \hline
  \end{tabular}
\end{table}
\input{table-HD-experiments}

%% file: table-HD-experiments.tex
 \begin{table*}
  \centering
  \caption{Performance, in terms of ${\cal R}$ and estimated $\hat C$,
    on high-dimensional datasets. (In the table, $m$ displays the
    effective dimension of the dataset and is the number of
    coordinates, PCs or KPCs used as per the descriptions in Section~\ref{sec:expt.HD}.)}  \label{tab:AR.HD}
    \begin{tabular}{lcccccccccccccc}
    \hline
    \hline
    Dataset & KNS & K-mH & DM & DM+ & MMC & EAC & GSN & SpC & k$k$-m & DBSCAN* & DP & PGMM & MSAL & MGHD\\
    $(n,p,K,m)$     &  & & & & &  &  & &  & &  &  &  & \\
    \hline
    Simplex-7  & 0.97 & 0.97 & 0.97 & 0.97 & 0.97 & 0.94 & 0.94 & 0.92 & 0.94 & 0.03 & 0.58 & 0.97  & 0.98 & 0.96\\
   ($560,7,7,7$)  & 7 & 7 & 7 & 7 & 7 & 6 & 7 & 7 & 7 & 508 & 79 & 7 & 7 & 7 \\
    \hline
     E.coli  & 0.72 & 0.63 & 0.70 & 0.68 & 0.59 & 0.77 & 0.03 & 0.26 & 0.45 & 0.31 &0.38  & 0.48 & 0 & 0.60\\
     ($336,7,7,5 $)& 9 & 4 & 4 & 4 & 8 & 10 & 16 & 15 & 7 & 174 & 12 & 8 & 1 & 8 \\
     \hline
     Wines-13  & 0.92 & 0.62 & 0.52 & 0.5 & 0.67 & 0.6 & 0.38 & 0.43 & 0.6 & 0 & 0.26 & 0.66  & 0 & 0.16\\
     ($178,13, 3,17$)  & 3 & 11 & 7 & 7 & 7 & 8 & 7 & 23 & 7 & 178 & 15 & 5 &1 & 6 \\
     \hline
     Wines-27  & 0.93 & -0.01 & 1 & 1 & 0.91 & 0.62 & 0 & 0.35 & 0.88 & 0 & 0.56  & 0.85 & 0 & 0.95\\
     ($178,27, 3,26$) & 3 & 2 & 3 & 3 & 4 & 7 & 1 & 8 & 3 & 178 & 12 & 3 & 1 & 3 \\
     \hline
     Olive Oils-Area & 0.55 & 0.56 & 0.85 & 0.82 & 0.66 & 0.55 & 0.51 & 0.48 & 0.56 & 0.47 & 0.40  & 0.61 & 0.70 & 0.58\\
     ($572,8,9,8$)& 4 & 8 & 7 & 12 & 11 & 14 & 5 & 18 & 9 & 201 & 27 & 5 & 9 & 7 \\
     \hline
     Olive Oils-Region & 0.89 & 0.69 & 0.45 & 0.47 & 0.22 & 0.46 & 0.67 & 0.23 & 0.4 & 0.44 & 0.41 & 0.59  & 0.58 & 0.54\\
     ($572,8,3,8$) & 4 & 8 & 7 & 12  & 11  & 14 & 5& 18 & 9 & 201 & 27 & 5 & 9 & 7 \\
    \hline
    Image & 0.54 & 0.48 & 0.49 & 0.46 & 0.28 & 0.59 & 0.10 & 0.22 & 0.52 & 0 & 0.38  & 0 & 0 & 0.56\\
    ($2310,19, 7,8$)& 12 & 17 & 18 & 17 & 46 & 40 & 48 & 45 & 7 & 1 & 56 & 1 & 1 & 7 \\
     \hline
     Yeast & 0.22 & 0.01 & -0.01 & -0.01 & 0.04 & 0.004 & 0.003 & 0.11 & 0.14 & 0 & 0.03  & 0 & 0 & 0.04\\
     ($1484,8,10,6$) & 5 & 4 & 6 & 6 & 37 & 13 & 33 & 17 & 10 & 1 & 44 & 1 & 1 & 10 \\
     \hline
     ALL & 0.68 & 0.19 & 0.14 & 0.14 & 0.35 & 0.53 & 0.54 & 0.55 & 0.5 & 0 & 0.45  & 0.61 & 0 & 0\\
     ($215,1000, 7,42$)& 6 & 18 & 6 & 6 & 9 & 9 & 5 & 12 & 7 & 215 & 41  & 7 & 1 & 1\\
     \hline
     Zipcode & 0.76 & 0.55 & 0.35 & 0.33 & 0.01 & 0.21 & 0.54 & 0.54 & 0 & 0 & 0.58  & 0.52 & 0 & 0.56\\
     ($2000,256,10,33$)  & 9 & 22 & 36 & 33 & 1 & 45 & 7 & 23 & 10 & 2000 & 53  & 6 & 1 & 6\\
     \hline
     Pendigits & 0.72 & 0.51 & 0.58 & 0.6 & 0.26 & 0.58 & 0.004 & 0.42 & 0.3 & 0.3 & 0.54  & 0.44 & 0 & 0.67\\
     ($10992,16,10,18$)& 15 & 9 & 40 & 32 & 59 & 58 & 59 & 48 & 10 &7 & 9 & 6 & 1 & 10 \\
     \hline
    \hline
    $\bar{\cal D}$  & 0.04 & 0.29 & 0.21 & 0.22 & 0.31 & 0.22 & 0.44 & 0.35 & 0.27 & 0.62 & 0.35 & 0.24 & 0.56 & 0.25 \\ 
    ${\cal D}_\sigma$ & 0.09 & 0.27 & 0.20 & 0.20 & 0.23 & 0.17 & 0.29 & 0.21 & 0.21 & 0.26 & 0.16 & 0.15 & 0.33 & 0.26 \\ 
    \hline
    \hline
    \end{tabular}
\end{table*}

%% file: arxiv-merging.bbl
\begin{thebibliography}{100}
\providecommand{\url}[1]{#1}
\csname url@samestyle\endcsname
\providecommand{\newblock}{\relax}
\providecommand{\bibinfo}[2]{#2}
\providecommand{\BIBentrySTDinterwordspacing}{\spaceskip=0pt\relax}
\providecommand{\BIBentryALTinterwordstretchfactor}{4}
\providecommand{\BIBentryALTinterwordspacing}{\spaceskip=\fontdimen2\font plus
\BIBentryALTinterwordstretchfactor\fontdimen3\font minus
  \fontdimen4\font\relax}
\providecommand{\BIBforeignlanguage}[2]{{%
\expandafter\ifx\csname l@#1\endcsname\relax
\typeout{** WARNING: IEEEtran.bst: No hyphenation pattern has been}%
\typeout{** loaded for the language `#1'. Using the pattern for}%
\typeout{** the default language instead.}%
\else
\language=\csname l@#1\endcsname
\fi
#2}}
\providecommand{\BIBdecl}{\relax}
\BIBdecl

\bibitem{ramey85}
D.~B. Ramey, ``Nonparametric clustering techniques,'' in \emph{Encyclopedia of
  Statistical Science}.\hskip 1em plus 0.5em minus 0.4em\relax New York: Wiley,
  1985, vol.~6, pp. 318--319.

\bibitem{mclachlanandbasford88}
G.~J. McLachlan and K.~E. Basford, \emph{Mixture Models: Inference and
  Applications to Clustering}.\hskip 1em plus 0.5em minus 0.4em\relax New York:
  Marcel Dekker, 1988.

\bibitem{kaufmanandrousseuw90}
L.~Kaufman and P.~J. Rousseuw, \emph{Finding Groups in Data}.\hskip 1em plus
  0.5em minus 0.4em\relax New York: John Wiley \& Sons, 1990.

\bibitem{everittetal01}
B.~S. Everitt, S.~Landau, and M.~Leesem, \emph{Cluster Analysis (4th
  ed.)}.\hskip 1em plus 0.5em minus 0.4em\relax New York: Hodder Arnold, 2001.

\bibitem{melnykovandmaitra10}
V.~Melnykov and R.~Maitra, ``Finite mixture models and model-based
  clustering,'' \emph{Statistics Surveys}, vol.~4, pp. 80--116, 2010.

\bibitem{xuandwunsch09}
R.~Xu and D.~C. Wunsch, \emph{Clustering}.\hskip 1em plus 0.5em minus
  0.4em\relax NJ, Hoboken: John Wiley \& Sons, 2009.

\bibitem{bouveryonetal19}
C.~Bouveyron, G.~Celeux, B.~T. Murphy, and A.~E. Raftery, \emph{Model-Based
  Clustering and Classification for Data Science: With Applications in
  R}.\hskip 1em plus 0.5em minus 0.4em\relax Cambridge Series in Statistical
  and Probabilistic Mathematics, 2019.

\bibitem{michenerandsokal57}
C.~D. Michener and R.~R. Sokal, ``A quantitative approach to a problem in
  classification,'' \emph{Evolution}, vol.~11, pp. 130--162, 1957.

\bibitem{hinneburgandkeim99}
A.~Hinneburg and D.~Keim, ``Cluster discovery methods for large databases: from
  the past to the future,'' in \emph{Proceedings of the ACM SIGMOD
  International Conference on the Management of Data}, 1999.

\bibitem{maitra01}
R.~Maitra, ``Clustering massive datasets with applications to software metrics
  and tomography,'' \emph{Technometrics}, vol.~43, no.~3, pp. 336--346, 2001.

\bibitem{johnson67}
S.~Johnson, ``Hierarchical clustering schemes,'' \emph{Psychometrika}, vol.
  32:3, pp. 241--254, 1967.

\bibitem{jainanddubes88}
A.~Jain and R.~Dubes, \emph{Algorithms for clustering data}.\hskip 1em plus
  0.5em minus 0.4em\relax Englewood Cliffs, NJ: Prentice Hall, 1988.

\bibitem{forgy65}
E.~Forgy, ``Cluster analysis of multivariate data: efficiency vs.
  interpretability of classifications,'' \emph{Biometrics}, vol.~21, pp.
  768--780, 1965.

\bibitem{macqueen67}
J.~MacQueen, ``Some methods for classification and analysis of multivariate
  observations,'' \emph{Proceedings of the Fifth Berkeley Symposium}, vol.~1,
  pp. 281--297, 1967.

\bibitem{titteringtonetal85}
D.~Titterington, A.~Smith, and U.~Makov, \emph{Statistical Analysis of Finite
  Mixture Distributions}.\hskip 1em plus 0.5em minus 0.4em\relax Chichester,
  U.K.: John Wiley \& Sons, 1985.

\bibitem{mclachlanandpeel00}
G.~McLachlan and D.~Peel, \emph{Finite Mixture Models}.\hskip 1em plus 0.5em
  minus 0.4em\relax New York: John Wiley and Sons, Inc., 2000.

\bibitem{mcnicholas16}
P.~D. McNicholas, \emph{Mixture model-based classification}.\hskip 1em plus
  0.5em minus 0.4em\relax Chapman and Hall/CRC, 2016.

\bibitem{hartiganandwong79}
J.~A. Hartigan and M.~A. Wong, ``A $k$-means clustering algorithm,''
  \emph{Applied Statistics}, vol.~28, pp. 100--108, 1979.

\bibitem{lloyd82}
S.~Lloyd, ``Least squares quantization in {PCM},'' \emph{Information Theory,
  IEEE Transactions on}, vol.~28, no.~2, pp. 129--137, 1982.

\bibitem{fraleyandraftery02}
C.~Fraley and A.~E. Raftery, ``Model-based clustering, discriminant analysis,
  and density estimation,'' \emph{Journal of the American Statistical
  Association}, vol.~97, pp. 611--631, 2002.

\bibitem{dhillonetal04}
I.~Dhillon, Y.~Guan, and B.~Kulis, ``A unified view of kernel k-means, spectral
  clustering and graph cuts,'' University of Texas at Austin, Tech. Rep.
  TR-04-25, 2004.

\bibitem{fredandjain05}
A.~L. Fred and A.~K. Jain, ``Combining multiple clusterings using evidence
  accumulation,'' \emph{IEEE transactions on pattern analysis and machine
  intelligence}, vol.~27, no.~6, pp. 835--850, 2005.

\bibitem{vonluxburg07}
U.~von Luxburg, ``A tutorial on spectral clustering,'' \emph{Statistics and
  Computing}, vol.~17, no.~4, pp. 395--416, December 2007.

\bibitem{baudryetal10}
J.-P. Baudry, A.~E. Raftery, G.~Celeux, K.~Lo, and R.~Gottardo, ``Combining
  mixture components for clustering,'' \emph{Journal of Computational and
  Graphical Statistics}, vol.~19, no.~2, pp. 332 -- 353, 2010.

\bibitem{hennig10}
C.~Hennig, ``Methods for merging {G}aussian mixture components,''
  \emph{Advances in Data Analysis and Classification}, 2010.

\bibitem{melnykov16}
V.~Melnykov, ``Merging mixture components for clustering through pairwise
  overlap,'' \emph{Journal of Computational and Graphical Statistics}, vol.~25,
  no.~1, pp. 66--90, 2016.

\bibitem{petersonetal17}
A.~D. Peterson, A.~P. Ghosh, and R.~Maitra, ``Merging $k$-means with
  hierarchical clustering for identifying general-shaped groups,'' \emph{Stat},
  vol.~7, no.~1, p. e172, 2018.

\bibitem{ghahramaniandhinton97}
Z.~Ghahramani and G.~E. Hinton, ``The {EM} algorithm for factor analyzers,''
  University of Toronto, Toronto, Canada, Tech. Rep. CRG-TR-96-1, 1997.

\bibitem{mcnicholasandmurphy08}
P.~D. McNicholas and T.~B. Murphy, ``Parsimonious {G}aussian mixture models,''
  \emph{Statistics and Computing}, vol.~18, no.~3, pp. 285--296, 2008.

\bibitem{franczaketal13}
B.~C. Franczak, R.~P. Browne, and P.~D. McNicholas, ``Mixtures of shifted
  asymmetric {L}aplace distributions,'' \emph{IEEE Transactions on Pattern
  Analysis and Machine Intelligence}, vol.~36, no.~6, pp. 1149--1157, 2014.

\bibitem{browneandmcnicholas15}
R.~P. Browne and P.~D. McNicholas, ``A mixture of generalized hyperbolic
  distributions,'' \emph{Canadian Journal of Statistics}, vol.~43, no.~2, pp.
  176--198, 2015.

\bibitem{stuetzleandnugent10}
W.~Stuetzle and R.~Nugent, ``A generalized single linkage method for estimating
  the cluster tree of a density,'' \emph{Journal of Computational and Graphical
  Statistics}, 2010.

\bibitem{esteretal96}
M.~Ester, H.-P. Kriegel, J.~Sander, X.~Xu \emph{et~al.}, ``A density-based
  algorithm for discovering clusters in large spatial databases with noise.''
  in \emph{{KDD}-96}, vol.~96, no.~34, 1996, pp. 226--231.

\bibitem{campelloetal13}
R.~J. Campello, D.~Moulavi, and J.~Sander, ``Density-based clustering based on
  hierarchical density estimates,'' in \emph{Pacific-Asia conference on
  knowledge discovery and data mining}.\hskip 1em plus 0.5em minus 0.4em\relax
  Springer, 2013, pp. 160--172.

\bibitem{rodriguezandlaio14}
A.~Rodriguez and A.~Laio, ``Clustering by fast search and find of density
  peaks,'' \emph{Science}, vol. 344, no. 6191, pp. 1492--1496, 2014.

\bibitem{R}
\BIBentryALTinterwordspacing
{R Development Core Team}, \emph{R: A Language and Environment for Statistical
  Computing}, R Foundation for Statistical Computing, Vienna, Austria, 2018,
  {ISBN} 3-900051-07-0. [Online]. Available: \url{http://www.R-project.org}
\BIBentrySTDinterwordspacing

\bibitem{RmixmodCombi}
\BIBentryALTinterwordspacing
J.-P. Baudry and G.~Celeux, \emph{RmixmodCombi: Combining Mixture Components
  for Clustering}, 2014, r package version 1.0. [Online]. Available:
  \url{https://CRAN.R-project.org/package=RmixmodCombi}
\BIBentrySTDinterwordspacing

\bibitem{maitraandmelnykov10}
R.~Maitra and V.~Melnykov, ``Simulating data to study performance of finite
  mixture modeling and clustering algorithms,'' \emph{Journal of Computational
  and Graphical Statistics}, vol.~19, no.~2, pp. 354--376, 2010.

\bibitem{fraleyandraftery98}
C.~Fraley and A.~E. Raftery, ``How many clusters? which cluster method? answers
  via model-based cluster analysis,'' \emph{Computer Journal}, vol.~41, pp.
  578--588, 1998.

\bibitem{maitra10}
R.~Maitra, ``A re-defined and generalized percent-overlap-of-activation measure
  for studies of {fMRI} reproducibility and its use in identifying outlier
  activation maps,'' \emph{Neuroimage}, vol.~50, no.~1, pp. 124--135, 2010.

\bibitem{melnykovandmaitra11}
V.~Melnykov and R.~Maitra, ``{CARP}: Software for fishing out good clustering
  algorithms,'' \emph{Journal of Machine Learning Research}, vol.~12, pp. 69 --
  73, 2011.

\bibitem{melnykovetal12}
\BIBentryALTinterwordspacing
V.~Melnykov, W.-C. Chen, and R.~Maitra, ``{MixSim}: An {R} package for
  simulating data to study performance of clustering algorithms,''
  \emph{Journal of Statistical Software}, vol.~51, no.~12, pp. 1--25, 2012.
  [Online]. Available: \url{http://www.jstatsoft.org/v51/i12/}
\BIBentrySTDinterwordspacing

\bibitem{mahalanobis36}
P.~C. Mahalanobis, ``On the generalised distance in statistics,''
  \emph{Proceedings of the {N}ational {I}nstitute of {S}ciences of {I}ndia},
  vol.~2, no.~1, p. 49–55, 1936.

\bibitem{maitraandramler09}
R.~Maitra and I.~P. Ramler, ``Clustering in the presence of scatter,''
  \emph{Biometrics}, vol.~65, pp. 341 -- 352, 2009.

\bibitem{lithioandmaitra18}
A.~Lithio and R.~Maitra, ``An efficient k-means-type algorithm for clustering
  datasets with incomplete records,'' \emph{Statistical Analysis and Data
  Mining: The ASA Data Science Journal}, vol.~11, no.~6, pp. 296--311, 2018.

\bibitem{krzanowskiandlai88}
W.~J. Krzanowski and Y.~Lai, ``A criterion for determining the number of groups
  in a data set using sum-of-squares clustering,'' \emph{Biometrics}, pp.
  23--34, 1988.

\bibitem{sugarandjames03}
C.~A. Sugar and G.~M. James, ``Finding the number of clusters in a dataset,''
  \emph{Journal of the American Statistical Association}, vol.~98, no. 463,
  2003.

\bibitem{maitraetal12}
R.~Maitra, V.~Melnykov, and S.~Lahiri, ``Bootstrapping for significance of
  compact clusters in multi-dimensional datasets,'' \emph{Journal of the
  American Statistical Association}, vol. 107, no. 497, pp. 378--392, 2012.

\bibitem{silverman86}
B.~W. Silverman, \emph{Density Estimation for Statistics and Data
  Analysis}.\hskip 1em plus 0.5em minus 0.4em\relax London: Chapman \&
  Hall/CRC, 1986.

\bibitem{rosenblatt56}
M.~Rosenblatt, ``Remarks on some nonparametric estimates of a density
  function,'' \emph{The Annals of Mathematical Statistics}, vol.~27, no.~3, p.
  832, 1956.

\bibitem{parzen62}
E.~Parzen, ``On estimation of a probability density function and mode,''
  \emph{The Annals of Mathematical Statistics}, vol.~33, no.~3, p. 1065, 1962.

\bibitem{wandandjones95}
M.~P. Wand and M.~C. Jones, \emph{Kernel Smoothing}.\hskip 1em plus 0.5em minus
  0.4em\relax London: Chapman \& Hall/CRC, 1995.

\bibitem{epanechnikov69}
V.~A. Epanechnikov, ``Non-parametric estimation of a multivariate probability
  density,'' \emph{Theory of Probability and its Applications}, vol.~14, p.
  153–158, 1969.

\bibitem{azzalini81}
A.~Azzalini, ``A note on the estimation of a distribution function and
  quantiles by a kernel method,'' \emph{Biometrika}, vol.~68, no.~1, pp.
  326--328, 1981.

\bibitem{reiss81}
R.-D. Reiss, ``Nonparametric estimation of smooth distribution functions,''
  \emph{Scandinavian Journal of Statistics}, pp. 116--119, 1981.

\bibitem{bouezmarniandscaillet05}
T.~Bouezmarni and O.~Scaillet, ``Consistency of asymmetric kernel density
  estimators and smoothed histograms with application to income data,''
  \emph{Econometric Theory}, vol.~21, no.~02, pp. 390--412, 2005.

\bibitem{chen00}
S.~X. Chen, ``Probability density function estimation using gamma kernels,''
  \emph{Annals of the Institute of Statistical Mathematics}, vol.~52, no.~3,
  pp. 471--480, 2000.

\bibitem{jeonandkim13}
Y.~Jeon and J.~H.~T. Kim, ``A gamma kernel density estimation for insurance
  loss data,'' \emph{Insurance: Mathematics and Economics}, vol.~53, pp.
  569--579, 2013.

\bibitem{scaillet04}
O.~Scaillet, ``Density estimation using inverse and reciprocal inverse
  {G}aussian kernels,'' \emph{Nonparametric Statistics}, vol.~16, no. 1-2, pp.
  217--226, 2004.

\bibitem{huang97a}
Z.~Huang, ``Clustering large data sets with mixed numeric and categorical
  values,'' in \emph{Proceedings of the First Pacific Asia Knowledge Discovery
  and Data Mining Conference}.\hskip 1em plus 0.5em minus 0.4em\relax
  Singapore: World Scientific, 1997, p. 21–34.

\bibitem{huang98}
------, ``Extensions to the $k$-means algorithm for clustering large data sets
  with categorical values,'' \emph{Data Mining and Knowledge Discovery},
  vol.~2, p. 283–304, 1998.

\bibitem{chaturvedietal01}
A.~Chaturvedi, P.~E. Green, and J.~D. Caroll, ``{$K$}-modes clustering,''
  \emph{Journal of Classification}, vol.~18, pp. 35--55, 2001.

\bibitem{dormanandmaitra20}
K.~S. Dorman and R.~Maitra, ``An efficient $k$-modes algorithm for clustering
  categorical datasets,'' \emph{ArXiv e-prints:2006.03936}, 2020.

\bibitem{ruschendorf13}
L.~R{\"u}schendorf, \emph{Mathematical Risk Analysis}.\hskip 1em plus 0.5em
  minus 0.4em\relax Berlin Heidelberg: Springer-Verlag, 2013.

\bibitem{daietal19}
Y.~Zhu, F.~Dai, and R.~Maitra, ``Three-dimensional radial visualization of
  high-dimensional continuous or discrete datasets,'' \emph{ArXiv
  e-prints:1905.09505}, Mar. 2019.

\bibitem{nelsen06}
R.~B. Nelsen, \emph{An Introduction to Copulas}, 2nd~ed.\hskip 1em plus 0.5em
  minus 0.4em\relax New York: Springer, 2006.

\bibitem{gionis07}
A.~Gionis, H.~Mannila, and P.~Tsaparas, ``Clustering aggregation,'' \emph{ACM
  Transactions on Knowledge Discovery from Data (TKDD)}, vol.~1, no.~1, p.~4,
  2007.

\bibitem{hashlerandpiekenbrock18}
\BIBentryALTinterwordspacing
M.~Hahsler and M.~Piekenbrock, \emph{dbscan: Density Based Clustering of
  Applications with Noise (DBSCAN) and Related Algorithms}, 2018, r package
  version 1.1-3. [Online]. Available:
  \url{https://CRAN.R-project.org/package=dbscan}
\BIBentrySTDinterwordspacing

\bibitem{pendersenetal17}
\BIBentryALTinterwordspacing
T.~L. Pedersen, S.~Hughes, and X.~Qiu, \emph{densityClust: Clustering by Fast
  Search and Find of Density Peaks}, 2017, r package version 0.3. [Online].
  Available: \url{https://CRAN.R-project.org/package=densityClust}
\BIBentrySTDinterwordspacing

\bibitem{mcnicholasetal18}
\BIBentryALTinterwordspacing
P.~D. McNicholas, A.~ElSherbiny, A.~F. McDaid, and T.~B. Murphy, \emph{pgmm:
  Parsimonious Gaussian Mixture Models}, 2018, r package version 1.2.3.
  [Online]. Available: \url{https://CRAN.R-project.org/package=pgmm}
\BIBentrySTDinterwordspacing

\bibitem{franczacketal18}
\BIBentryALTinterwordspacing
B.~C. Franczak, R.~P. Browne, P.~D. McNicholas, and K.~L. Burak, \emph{MixSAL:
  Mixtures of Multivariate Shifted Asymmetric Laplace (SAL) Distributions},
  2018, r package version 1.0. [Online]. Available:
  \url{https://CRAN.R-project.org/package=MixSAL}
\BIBentrySTDinterwordspacing

\bibitem{tortoraetal19}
\BIBentryALTinterwordspacing
C.~Tortora, A.~ElSherbiny, R.~P. Browne, B.~C. Franczak, , P.~D. McNicholas,
  and D.~D. Amos., \emph{MixGHD: Model Based Clustering, Classification and
  Discriminant Analysis Using the Mixture of Generalized Hyperbolic
  Distributions}, 2019, r package version 2.3.2. [Online]. Available:
  \url{https://CRAN.R-project.org/package=MixGHD}
\BIBentrySTDinterwordspacing

\bibitem{hubertandarabie85}
L.~Hubert and P.~Arabie, ``Comparing partitions,'' \emph{Journal of
  Classification}, vol.~2, pp. 193--218, 1985.

\bibitem{steinley04}
D.~Steinley, ``Properties of the {Hubert}-{Arabie} adjusted {R}and index.''
  \emph{Psychological methods}, vol.~9, no.~3, p. 386, 2004.

\bibitem{maitra09}
R.~Maitra, ``Initializing partition-optimization algorithms,'' \emph{IEEE/ACM
  Transactions on Computational Biology and Bioinformatics}, vol.~6, pp.
  144--157, 2009.

\bibitem{zahn71}
C.~T. Zahn, ``Graph-theoretical methods for detecting and describing gestalt
  clusters,'' \emph{IEEE Transactions on computers}, vol. 100, no.~1, pp.
  68--86, 1971.

\bibitem{jainandlaw05}
A.~K. Jain and M.~H.~C. Law, ``Data clustering: A user’s dilemma,'' in
  \emph{Pattern Recognition and Machine Intelligence. PReMI 2005}, ser. Lecture
  Notes in Computer Science, S.~K. Pal, B.~S., and B.~S., Eds., vol.
  3776.\hskip 1em plus 0.5em minus 0.4em\relax Berlin, Heidelberg: Springer,
  2005, pp. 1--10.

\bibitem{changandyeung08}
\BIBentryALTinterwordspacing
H.~Chang and D.-Y. Yeung, ``Robust path-based spectral clustering,''
  \emph{Pattern Recognition}, vol.~41, no.~1, pp. 191 -- 203, 2008. [Online].
  Available:
  \url{http://www.sciencedirect.com/science/article/pii/S0031320307002038}
\BIBentrySTDinterwordspacing

\bibitem{newmanetal98}
D.~J. Newman, S.~Hettich, C.~L. Blake, and C.~J. Merz, ``{UCI} repository of
  machine learning databases,'' 1998.

\bibitem{nakaiandkinehasa91}
K.~Nakai and M.~Kinehasa, ``Expert sytem for predicting protein localization
  sites in gram-negative bacteria,'' \emph{PROTEINS: Structure, Function, and
  Genetics}, vol.~11, pp. 95--110, 1991.

\bibitem{hortonandnakai96}
P.~Horton and K.~Nakai, ``A probablistic classification system for predicting
  the cellular localization sites of proteins,'' \emph{Intelligent Systems in
  Molecular Biology}, pp. 109--115, 1985.

\bibitem{maitra02}
R.~Maitra, ``A statistical perspective to data mining,'' \emph{Journal of the
  Indian Society of Probability and Statistics}, vol.~6, pp. 28--77, 2002.

\bibitem{forinaetal88}
M.~Forina, R.~Leardi, and S.~Lanteri, ``{PARVUS} - an extendible package for
  data exploration, classification and correlation,'' Via Brigata Salerno,
  16147 Genoa, Italy, 1988.

\bibitem{aeberhardetal92}
D.~C. S.~Aeberhard and O.~de~Vel, ``Comparison of classifiers in high
  dimensional settings,'' Department of Computer Science and Department of
  Mathematics and Statistics, James Cook University of North Queensland, Tech.
  Rep. 92-02, 1992.

\bibitem{forinaandtiscornia82}
M.~Forina and E.~Tiscornia, ``Pattern recognition methods in the prediction of
  italian olive oil origin by their fatty acid content,'' \emph{Annali di
  Chimica}, vol.~72, pp. 143--155, 1982.

\bibitem{forinaetal83}
M.~Forina, C.~Armanino, S.~Lanteri, and E.~Tiscornia, ``Classification of olive
  oils from their fatty acid composition,'' in \emph{Food Research and Data
  Analysis}.\hskip 1em plus 0.5em minus 0.4em\relax London: Applied Science
  Publishers, 1983, pp. 189--214.

\bibitem{nakai96}
\BIBentryALTinterwordspacing
K.~Nakai, ``{UCI} machine learning repository,'' 1996. [Online]. Available:
  \url{http://archive.ics.uci.edu/ml}
\BIBentrySTDinterwordspacing

\bibitem{yeohetal02}
E.-J. Yeoh, M.~E. Ross, S.~A. Shurtleff, W.~Williams, D.~Patel, R.~Mahfouz,
  F.~G. Behm, S.~C. Raimondi, M.~V. Relling, A.~Patel, C.~Cheng, D.~Campana,
  D.~Wilkins, X.~Zhou, J.~Li, H.~Liu, C.-H. Pui, W.~E. Evans, C.~Naeve,
  L.~Wong, and J.~R. Downing, ``Classification, subtype discovery, and
  prediction of outcome in pediatric acute lymphoblastic leukemia by gene
  expression profiling,'' \emph{Cancer Cell}, vol.~1, no.~2, pp. 133 -- 143,
  2002.

\bibitem{alimoglu96}
F.~Alimoglu, ``Combining multiple classifiers for pen-based handwritten digit
  recognition,'' Master's thesis, Institute of Graduate Studies in Science and
  Engineering, Bogazici University, 1996.

\bibitem{alimogluandalpaydin96}
F.~Alimoglu and E.~Alpaydin, ``Methods of combining multiple classifiers based
  on different representations for pen-based handwriting recognition,'' in
  \emph{Proceedings of the Fifth Turkish Artificial Intelligence and Artificial
  Neural Networks Symposium (TAINN 96)}, Istanbul, Turkey, 1996.

\bibitem{samariaandharter94}
F.~S. Samaria and A.~C. Harter, ``Parameterization of a stochastic model for
  human face identification,'' in \emph{{P}roceedings of the {S}econd {IEEE}
  {W}orkshop on {A}pplications of {C}omputer {V}ision}, Sarasota, Florida,
  1994, pp. 138--142.

\bibitem{schwarz78}
G.~Schwarz, ``Estimating the dimensions of a model,'' \emph{Annals of
  Statistics}, vol.~6, pp. 461--464, 1978.

\bibitem{sampatetal09}
M.~P. Sampat, Z.~Wang, S.~Gupta, A.~C. Bovik, and M.~K. Markey, ``Complex
  wavelet structural similarity: A new image similarity index,'' \emph{IEEE
  Transactions on Image Processing}, vol.~18, no.~11, pp. 2385--2401, 2009.

\bibitem{wagstaff04}
K.~Wagstaff, ``Clustering with missing values: No imputation required,'' in
  \emph{Classification, clustering, and data mining applications}.\hskip 1em
  plus 0.5em minus 0.4em\relax Springer, 2004, pp. 649--658.

\bibitem{chattopadhyayetal07}
T.~Chattopadhyay, R.~Misra, A.~K. Chattopadhyay, and M.~Naskar, ``Statistical
  evidence for three classes of gamma-ray bursts,'' \emph{Astrophysical
  Journal}, vol. 667, no.~2, p. 1017, 2007.

\bibitem{piran05}
T.~Piran, ``The physics of gamma-ray bursts,'' \emph{Rev. Mod. Phys.}, vol.~76,
  pp. 1143--1210, Jan 2005.

\bibitem{mazetsetal81}
E.~P. {Mazets}, S.~V. {Golenetskii}, V.~N. {Ilinskii}, V.~N. {Panov}, R.~L.
  {Aptekar}, I.~A. {Gurian}, M.~P. {Proskura}, I.~A. {Sokolov}, Z.~I.
  {Sokolova}, and T.~V. {Kharitonova}, ``{Catalog of cosmic gamma-ray bursts
  from the KONUS experiment data. I.}'' \emph{Astrophysics and Space Science},
  vol.~80, pp. 3--83, Nov. 1981.

\bibitem{norrisetal84}
J.~P. {Norris}, T.~L. {Cline}, U.~D. {Desai}, and B.~J. {Teegarden},
  ``{Frequency of fast, narrow gamma-ray bursts},'' \emph{Nature}, vol. 308, p.
  434, Mar. 1984.

\bibitem{dezalayetal92}
J.-P. {Dezalay}, C.~{Barat}, R.~{Talon}, R.~{Syunyaev}, O.~{Terekhov}, and
  A.~{Kuznetsov}, ``{Short cosmic events - A subset of classical GRBs?}'' in
  \emph{American Institute of Physics Conference Series}, ser. American
  Institute of Physics Conference Series, W.~S. {Paciesas} and G.~J. {Fishman},
  Eds., vol. 265, 1992, pp. 304--309.

\bibitem{mukherjeeetal98}
S.~{Mukherjee}, E.~D. {Feigelson}, G.~{Jogesh Babu}, F.~{Murtagh}, C.~{Fraley},
  and A.~{Raftery}, ``Three types of gamma-ray bursts,'' \emph{Astrophyical
  Journal}, vol. 508, pp. 314--327, Nov. 1998.

\bibitem{chattopadhyayandmaitra17}
S.~Chattopadhyay and R.~Maitra, ``Gaussian-mixture-model-based cluster analysis
  finds five kinds of gamma-ray bursts in the {BATSE} catalogue,''
  \emph{Monthly Notices of the Royal Astronomical Society}, vol. 469, no.~3,
  pp. 3374--3389, 2017.

\bibitem{chattopadhyayandmaitra18}
------, ``{Multivariate t-mixture-model-based cluster analysis of BATSE
  catalogue establishes importance of all observed parameters, confirms five
  distinct ellipsoidal sub-populations of gamma-ray bursts},'' \emph{Monthly
  Notices of the Royal Astronomical Society}, vol. 481, no.~3, pp. 3196--3209,
  07 2018.

\bibitem{rafteryanddean06}
A.~E. Raftery and N.~Dean, ``Variable selection for model-based clustering,''
  \emph{Journal of the American Statistical Association}, vol. 101, pp.
  168--178, 2006.

\bibitem{berryandmaitra19}
N.~S. Berry and R.~Maitra, ``Tik-means: Transformation-infused k-means
  clustering for skewed groups,'' \emph{Statistical Analysis and Data Mining:
  The ASA Data Science Journal}, vol.~12, no.~3, pp. 223--233, 2019.

\bibitem{bandettinietal93}
P.~A. Bandettini, A.~Jesmanowicz, E.~C. Wong, and J.~S. Hyde, ``Processing
  strategies for time-course data sets in functional mri of the human brain,''
  \emph{Magnetic Resonance in Medicine}, vol.~30, pp. 161--173, 1993.

\bibitem{belliveauetal91}
J.~W. Belliveau, D.~N. Kennedy, R.~C. McKinstry, B.~R. Buchbinder, R.~M.
  Weisskoff, M.~S. Cohen, J.~M. Vevea, T.~J. Brady, and B.~R. Rosen,
  ``Functional mapping of the human visual cortex by magnetic resonance
  imaging,'' \emph{Science}, vol. 254, pp. 716--719, 1991.

\bibitem{kwongetal92}
K.~K. Kwong, J.~W. Belliveau, D.~A. Chesler, I.~E. Goldberg, R.~M. Weisskoff,
  B.~P. Poncelet, D.~N. Kennedy, B.~E. Hoppel, M.~S. Cohen, R.~Turner, H.-M.
  Cheng, T.~J. Brady, and B.~R. Rosen, ``Dynamic magnetic resonance imaging of
  human brain activity during primary sensory stimulation,'' \emph{Proceedings
  of the National Academy of Sciences of the United States of America},
  vol.~89, pp. 5675--5679, 1992.

\bibitem{ogawaetal90}
S.~Ogawa, T.~M. Lee, A.~S. Nayak, and P.~Glynn, ``Oxygenation-sensitive
  contrast in magnetic resonance image of rodent brain at high magnetic
  fields,'' \emph{Magnetic Resonance in Medicine}, vol.~14, pp. 68--78, 1990.

\bibitem{fristonetal94}
K.~J. Friston, P.~Jezzard, and R.~Turner, ``Analysis of functional {MRI}
  time-series,'' \emph{Human Brain Mapping}, vol.~1, pp. 153--171, 1994.

\bibitem{glover99}
G.~H. Glover, ``Deconvolution of impulse response in event-related {BOLD}
  {fMRI},'' \emph{Neuroimage}, vol.~9, pp. 416--429, 1999.

\bibitem{lazar08}
N.~A. Lazar, \emph{The Statistical Analysis of Functional MRI Data}.\hskip 1em
  plus 0.5em minus 0.4em\relax Springer, 2008.

\bibitem{fristonetal95}
K.~J. Friston, A.~P. Holmes, K.~J. Worsley, J.-B. Poline, C.~D. Frith, and
  R.~S.~J. Frackowiak, ``Statistical parametric maps in functional imaging: A
  general linear approach,'' \emph{Human Brain Mapping}, vol.~2, pp. 189--210,
  1995.

\bibitem{formanetal95}
S.~D. Forman, J.~D. Cohen, M.~Fitzgerald, W.~F. Eddy, M.~A. Mintun, and D.~C.
  Noll, ``Improved assessment of significant activation in functional magnetic
  resonance imaging ({f}mri): Use of a cluster-size threshold,'' \emph{Magnetic
  Resonance in Medicine}, vol.~33, pp. 636--647, 1995.

\bibitem{Genoveseetal2002}
C.~R. Genovese, N.~A. Lazar, and T.~E. Nichols, ``Thresholding of statistical
  maps in functional neuroimaging using the false discovery rate:,''
  \emph{Neuroimage}, vol.~15, pp. 870--878, 2002.

\bibitem{thirionetal14}
\BIBentryALTinterwordspacing
B.~Thirion, G.~Varoquaux, E.~Dohmatob, and J.-B. Poline, ``Which {fMRI}
  clustering gives good brain parcellations?'' \emph{Frontiers in
  Neuroscience}, vol.~8, p. 167, 2014. [Online]. Available:
  \url{https://www.frontiersin.org/article/10.3389/fnins.2014.00167}
\BIBentrySTDinterwordspacing

\bibitem{maitraetal02}
R.~Maitra, S.~R. Roys, and R.~P. Gullapalli, ``Test-retest reliability
  estimation of functional {MRI} data,'' \emph{Magnetic Resonance in Medicine},
  vol.~48, pp. 62--70, 2002.

\bibitem{maitra09b}
R.~Maitra, ``Assessing certainty of activation or inactivation in test-retest
  {fMRI} studies,'' \emph{Neuroimage}, vol.~47, no.~1, pp. 88--97, 2009.

\bibitem{almodovarandmaitra19}
I.~A. Almod\'ovar-Rivera and R.~Maitra, ``{FAST} adaptive smoothed thresholding
  for improved activation detection in low-signal f{MRI},'' \emph{{IEEE}
  Transactions on Medical Imaging}, vol.~38, no.~12, pp. 2821--2828, 2019.

\bibitem{almodovarandmaitra19a}
I.~Almod{\'o}var-Rivera and R.~Maitra, ``{RFASTfMRI}: Fast adaptive smoothing
  and thresholding for improved activation detection in low-signal {fMRI},''
  2019, {R} Package, URL http://github.com/ialmodovar/RFASTfMRI.

\bibitem{jaccard1901}
P.~Jaccard, ``\`{E}tude comparative de la distribution florale dans une portion
  des alpes et des jura,'' \emph{Bulletin del la Soci\`{e}t\`{e} Vaudoise des
  Sciences Naturelles}, vol.~37, p. 547–579, 1901.

\bibitem{smiejaandwiercioch17}
M.~{\'S}mieja and M.~Wiercioch, ``Constrained clustering with a complex cluster
  structure,'' \emph{Advances in Data Analysis and Classification}, vol.~11,
  no.~3, pp. 493--518, 2017.

\end{thebibliography}
